\documentclass[12pt]{article}

\usepackage[top=1.25in, left=1.25in, bottom=1.25in, right=1.25in]{geometry}
 
\usepackage{graphicx}
\usepackage{epstopdf}
\usepackage{amsthm}
\usepackage{amsmath}
\usepackage{amssymb}
\usepackage{amstext}
\usepackage{amsfonts}
\usepackage{enumerate}
\usepackage[onehalfspacing]{setspace}
\usepackage[shortlabels, inline]{enumitem}
\usepackage{multirow}
\usepackage{cancel}  
\usepackage[normalem]{ulem}

\DeclareUnicodeCharacter{00A0}{ }
\DeclareMathOperator*{\argmax}{arg\,max}
\DeclareMathOperator*{\argmin}{arg\,min}         

\usepackage{mathrsfs}
\usepackage{wasysym}

\usepackage[utf8]{inputenc}
\usepackage{bm, cases, mathtools, thmtools}
\usepackage{verbatim}
\usepackage{graphicx}\graphicspath{{figures/}}
\usepackage{multicol}
\usepackage{tabularx}
\usepackage[usenames,dvipsnames]{xcolor}
\usepackage{mathrsfs}
\usepackage{url}
\usepackage[normalem]{ulem}
\usepackage{xstring}
\usepackage{enumitem}

\usepackage[%
    minnames=1,maxnames=99,maxcitenames=3,
    style=authoryear,
    doi=false,url=false,
    firstinits=true,hyperref,natbib,backend=bibtex]{biblatex}
\renewbibmacro{in:}{%
  \ifentrytype{article}{}{\printtext{\bibstring{in}\intitlepunct}}}
\bibliography{biblio}

\usepackage[colorlinks,citecolor=blue,urlcolor=blue,linkcolor=RawSienna]{hyperref}
\usepackage{hypernat}
\usepackage{datetime}

\DeclareMathAlphabet\EuRoman{U}{eur}{m}{n}
\SetMathAlphabet\EuRoman{bold}{U}{eur}{b}{n}

\usepackage{euscript,microtype}

\usepackage[capitalize]{cleveref}

\crefname{assumption}{Assumption}{Assumptions}
\crefname{claim}{Claim}{Claims}

\makeatletter
\let\reftagform@=\tagform@
\def\tagform@#1{\maketag@@@{\ignorespaces\textcolor{gray}{(#1)}\unskip\@@italiccorr}}
\renewcommand{\eqref}[1]{\textup{\reftagform@{\ref{#1}}}}
\makeatother

\definecolor{WowColor}{rgb}{.75,0,.75}
\definecolor{SubtleColor}{rgb}{0,0,.50}

\newcounter{margincounter}

\declaretheorem[style=plain,numberwithin=section,name=Theorem]{theorem}
\declaretheorem[style=plain,sibling=theorem,name=Lemma]{lemma}

\declaretheorem[style=plain,sibling=theorem,name=Claim]{claim}

\declaretheorem[style=definition,sibling=theorem,name=Definition]{definition}
\declaretheorem[style=definition,name=Assumption]{assumption}

\declaretheorem[style=remark,sibling=theorem,name=Remark]{remark}

\crefformat{conditionINNER}{#2(#1)#3}
\crefmultiformat{conditionINNER}
  {(#2#1#3)}
  { and~(#2#1#3)}
  {, (#2#1#3)}
  { and~(#2#1#3)}
\Crefformat{conditionINNER}{Condition~#2(#1)#3}
\Crefmultiformat{conditionINNER}
  {Conditions~(#2#1#3)}
  { and~(#2#1#3)}
  {, (#2#1#3)}
  { and~(#2#1#3)}

\declaretheoremstyle[
    spaceabove=-6pt,
    spacebelow=6pt,
    headfont=\normalfont\bfseries,
    bodyfont = \normalfont,
    postheadspace=1em,
    qed=$\square$,
    headpunct={{}}]{myproofstyle}

\numberwithin{equation}{section}
\numberwithin{theorem}{section}

\renewcommand{\appendixname}{}

\usepackage{amssymb}
\usepackage{mathrsfs}
\usepackage{leftidx}

\renewenvironment{quote}
               {\list{}{\rightmargin\leftmargin}%
                \item\relax}
               {\endlist}

\def\[#1\]{\begin{align}#1\end{align}}
\def\*[#1\]{\begin{align*}#1\end{align*}}

\newcommand{\Reals}{\mathbb{R}}
\newcommand{\Nats}{\mathbb{N}}

\newcommand{\NNReals}{\Reals_{\ge 0}}
\newcommand{\PosReals}{\Reals_{> 0}}

\newcommand{\intr}{\textrm{int}}

\newcommand{\dee}{\mathrm{d}}

\DeclareMathOperator{\sign}{sign}

\DeclareMathOperator*{\newlim}{\mathrm{lim}\vphantom{\mathrm{infsup}}}
\DeclareMathOperator*{\newmin}{\mathrm{min}\vphantom{\mathrm{infsup}}}
\DeclareMathOperator*{\newmax}{\mathrm{max}\vphantom{\mathrm{infsup}}}
\DeclareMathOperator*{\newinf}{\mathrm{inf}\vphantom{\mathrm{infsup}}}
\DeclareMathOperator*{\newsup}{\mathrm{sup}\vphantom{\mathrm{infsup}}}
\renewcommand{\lim}{\newlim}
\renewcommand{\min}{\newmin}
\renewcommand{\max}{\newmax}
\renewcommand{\inf}{\newinf}
\renewcommand{\sup}{\newsup}

\newcommand{\defn}[1]{\emph{#1}}

\newcommand{\cF}{\mathcal F}

\newcommand{\BorelSets}[1]{\mathcal{B}[#1]}
\newcommand{\NSE}[1]{{^{*}#1}}
\newcommand{\ST}{\mathsf{st}}

\newcommand{\HReals}{\NSE{\Reals}}

\newcommand{\NS}[1]{\mathrm{NS}(#1)}

\newcommand{\cA}{\mathcal{A}}

\newcommand{\cC}{\mathcal{C}}

\newtheorem{open problem}{Open Problem}

\newcommand{\Loeb}[1]{\overline{#1}}

\newcommand{\interior}[1]{%
  {\kern0pt#1}^{\mathrm{o}}%
}

\newcommand{\refproof}[1]{See \cref{#1} for \IfSubStr{#1}{,}{proofs}{a proof}. }

\newif\iflongform
\longformtrue

\iflongform

\else

\fi

\begin{document}
\baselineskip=.22in\parindent=30pt

\renewcommand{\cite}{\textcite}

\newtheorem*{theorem*}{Theorem}
\newtheorem{tm}{Theorem}
\newcommand{\thm}{\begin{tm}}
\newcommand{\nt}{\noindent}
\newcommand{\thmm}{\end{tm}}
\newtheorem{exam}{Example}
\newcommand{\ex}{\begin{exam}}
\newcommand{\exx}{\end{exam}}
\newtheorem{innercustomex}{Example}
\newenvironment{customex}[1]
{\renewcommand\theinnercustomex{#1}\innercustomex}
{\endinnercustomex}
\newcommand{\ben}{\begin{enumerate}}
\newcommand{\een}{\end{enumerate}}
\newcommand{\cB}{\mathscr{B}}
\newcommand{\Lip}[1]{\mathcal{L}_{1}(#1)}
\newcommand{\FM}[1]{\mathcal{M}(#1)}
\newcommand{\cM}{\mathcal{M}}
\newcommand{\cN}{\mathcal{N}}
\newcommand{\cH}{\mathcal{H}}
\newcommand{\cI}{\mathcal{I}}
\newcommand{\cJ}{\mathcal{J}}
\newcommand{\cV}{\mathcal{V}}
\newcommand{\PM}[1]{\mathcal{M}_{1}(#1)}
\newcommand{\topology}{\mathcal{T}}
\newcommand{\pd}[1]{{#1}_p}
\newcommand{\ipd}[1]{{#1}^{p}}
\newcommand{\boundary}[1]{\partial #1}
\newcommand{\cU}{\mathcal{U}}
\newcommand{\comp}[1]{\hat{#1}}
\newcommand{\expect}{\mathbb{E}}
\newcommand{\E}{\mathbb{E}}
\newcommand{\Prob}{\mathbb{P}}
\newcommand{\IProb}{\mathbb{Q}}
\newcommand{\kernel}{\{P_{x}(\cdot)\}_{x \in X}}
\newcommand{\kerp}{\{G_{i}(\cdot)\}_{i \in S}}
\newcommand{\asfc}{C t_{L}}
\newcommand{\asfcalpha}{[FIX ME]}
\newcommand{\lak}{\mathcal{T}_{L}}
\newcommand{\trk}{\mathcal{T}^{(S)}}
\newcommand{\gkernel}{\{g(x,1,\cdot)\}_{x\in X}}
\newcommand{\cP}{\mathcal{P}}
\newcommand{\rbig}{\succ}
\newcommand{\rsmall}{\prec}
\newcommand{\requal}{\sim}
\newcommand{\des}[2]{\rho_{#1}(#2)}
\newcommand{\symdiff}{\bigtriangleup}
\newcommand{\mprod}{\otimes}
\newcommand{\smprod}{\otimes^{\sigma}}
\newcommand{\closure}[1]{\mathrm{cl}(#1)}
\newtheorem{question}{Question}
\newcommand{\cto}{\twoheadrightarrow}
\newcommand{\cpl}[1]{#1^{c}}
\newcommand{\monad}[1]{\mu(#1)}
\newcommand{\bqu}{\sloppy \small \begin{quote}}
\newcommand{\equ}{\end{quote} \sloppy \large}
\newcommand{\fn}{\footnote}


\def\qed{\hfill\vrule height4pt width4pt
depth0pt}
\def\reff #1\par{\noindent\hangindent =\parindent
\hangafter =1 #1\par}
\def\title #1{\begin{center}
{\Large {\bf #1}}
\end{center}}
\def\author #1{\begin{center} {\large #1}
\end{center}}
\def\date #1{\centerline {\large #1}}
\def\place #1{\begin{center}{\large #1}
\end{center}}

\def\date #1{\centerline {\large #1}}
\def\place #1{\begin{center}{\large #1}\end{center}}
\def\intr #1{\stackrel {\circ}{#1}}
\def\R{{\rm I\kern-1.7pt R}}
 \def\N{{\rm I}\hskip-.13em{\rm N}}
 \newcommand{\cprod}{\Pi_{i=1}^\ell}
\let\Large=\large
\let\large=\normalsize


\begin{titlepage}

\def\thefootnote{\fnsymbol{footnote}}
\vspace*{0.05in}

\title{On Existence of Berk-Nash Equilibria in \vspace{0.4em} Misspecified
\\ Markov Decision Processes with Infinite Spaces\fn{{\footnotesize \setstretch{0.5}{\em Acknowledgments} The authors are grateful to  Chris Carroll, Ignacio Esponda, Mira Frick, Arthur Paul Pedersen, Luciano Pomatto, Demian Pouzo,  Santanu Roy, Aaron Smith, Maxwell Stinchcombe and Metin Uyanik for helpful comments. We would also like to thank seminar participants at The University of Toronto and The Johns Hopkins University, where parts of this paper were presented. Anderson gratefully acknowledges financial support from Swiss Re through the Consortium for Data Analytics in Risk. All remaining errors are solely ours.}}}

\vskip 1em

\author{
Robert M. Anderson\fn{Department of Economics, University of California, Berkeley, CA 94720-3880, USA.}
~ 
Haosui Duanmu\fn{Institute for Advanced Study in Mathematics, Harbin Institute of Technology, Harbin 150001, China.} 
~ 
Aniruddha Ghosh\fn{Department of Economics, The Johns Hopkins University, Baltimore, MD 21218, USA.}
~
M. Ali Khan\fn{Department of Economics, The Johns Hopkins University, Baltimore, MD 21218, USA.}}

\vskip 1.00em
\date{July 13, 2023}

\vskip 1.75em

\vskip 0.50em

\baselineskip=.18in
\begin{abstract} 
Model misspecification is a critical issue in many areas of economics.  In the context of misspecified Markov Decision Processes, \citet{ep21} defined the notion of Berk-Nash equilibrium and established its existence with finite state and action spaces. However, many substantive applications (including two of the three motivating examples presented by Esponda and Pouzo) involve continuous state or action spaces, and are thus not covered by the Esponda-Pouzo existence theorem.  We extend the existence of Berk-Nash equilibrium to compact action spaces and sigma-compact state spaces, with possibly unbounded payoff functions. 
A complication arises because Berk-Nash equilibrium depends critically on Radon-Nikodym derivatives, which are bounded in the finite case but typically unbounded in misspecified continuous models. 
The proofs rely on nonstandard analysis, and draw on novel argumentation traceable to work of the second author on nonstandard representations of Markov processes.
\end{abstract}


\vskip 0.50em
\noindent {\it Journal of Economic Literature} Classification
Numbers: C02, C62, D01, D83 

\medskip


\medskip


\bigskip


\end{titlepage}




{
  \hypersetup{linkcolor=black}
  \setcounter{tocdepth}{2}
  \tableofcontents
}

\vspace{-5pt}

 --------------------------------------------------------------- 

\vspace{15pt}


\setcounter{page}{1}

\setcounter{footnote}{0}


\setlength{\abovedisplayskip}{-10pt}
\setlength{\belowdisplayskip}{-10pt}


\newif\ifall
\alltrue 

%
%
%
%


\smallskip

\nocite{ma69jmaa, ks02, ks97pnas}

\newpage

\setstretch{1.4}

\section{Introduction}

Model misspecification is a critical issue in many areas of theoretical and empirical economics.\fn{\citet{ag73, as85}; also see the influential example of \citet{ny91}. This stimulus has been followed up by \citet{hs11}, and most recently by \citet{fls21}, and their references.}  Subjective Markov Decision Processes (SMDPs) generalize Optimal Control Problems and Markov Decision Processes (hereafter, MDPs). Optimal Control Theory, developed by Pontryagin and Bellman, involves the optimal selection of a control variable in a deterministic dynamical system.  MDPs extend the Optimal Control framework to stochastic processes in which the Markov transition probabilities are given by a known model. In SMDPs, the Markov transition probabilities are unknown, but assumed to be given by a model whose parameters must be estimated.  The model may be misspecified, i.e. there may be no parameter values under which the model is true.  The agent's goal is to learn the parameter values that minimize the distance between the misspecified model and the unknown true model.   \citet{ep21} used \citet{berk66} and the notion of weighted Kullback-Leibler divergence to formulate the notion of Berk-Nash equilibrium in SMDPs.\fn{SMDPs can be potentially applied to a myriad of settings such as that of the Lucas asset pricing problem, dynamic principal-agent problems, and consumption-saving problems; See \cite{gh22} for such examples and their associated monotone comparative statics properties.} 
Although their Berk-Nash equilibrium \emph{notion} applies broadly to SMDPs with finite or infinite state and action spaces, the Esponda and Pouzo \emph{theorem} on the existence of Berk-Nash equilibrium applies only to finite state and action spaces, and hence only to bounded payoff functions.  In this paper, we extend the existence of Berk-Nash equilibrium to SMDPs with compact action spaces, sigma-compact state spaces, and potentially unbounded payoff functions. 

In order to appreciate the economic significance of the generalization to infinite state and action spaces, it is helpful to note the following facts.\fn{See \citet{pu94} for illustrations covering operations research, economics and engineering. Section \ref{secmainresults} sketches settings in economic theory that naturally demand infinite state and action spaces. Examples include the asset selling problem in \citet{ka62} and the employment seeking problem in \citet{sl89}.} Two of the three motivating examples presented by Esponda and Pouzo involve infinite state or action spaces. Risk aversion is critical to the modeling of choice under uncertainty, but standard risk preference models (such as CRRA) are unbounded functions on infinite state spaces and often with continuous action spaces. Standard models in finance involve MDPs with Gaussian or log-normal (hence unbounded) asset prices. 
It is thus highly desirable to extend the existence theorem to infinite state and action spaces and unbounded payoff functions. Our existence theorems cover all of the settings just described.  

In this paper, we consider five examples from three important economic environments: (i) neoclassical producer theory, (ii)  the  optimal savings problem, and (iii) identification and inference in econometric theory.  In the first case, we consider two instances featuring demand and supply shocks to the revenues and the costs of the producer, and note the consequences of the misspecified distributions of these shocks for the profit-maximizing choices.\fn{\citet{lo09} explores the role of productivity shocks, news shocks and sampling shocks in driving business cycles with the shocks normally distributed with the real line as their support.}  The second environment extends Example 2 of EP, which features an optimal savings problem with a binary preference shock, to shocks with continuous and unbounded support.\fn{A key paper that connects learning, optimal savings and uncertainty is \citet{kms09}. Our convergence result (Theorem 4) in  \cref{sec:onlineappendix} applies potentially to such settings.} Finally, we provide two examples in Gaussian AR(1) processes with possibly unbounded payoff functions that connect the notion of Berk-Nash equilibrium to the existence of unit roots.\fn{See Examples \ref{example-Unit_Root} and \ref{example-Unit_Root2}. Also, see \citet{fa21} for empirical illustrations connecting unit roots to model misspecification in macroeconomic settings. There are at least  three more avenues where misspecification is being explored; climate economics (\citet{bm20}), axiomatic decision theory (\citet{hmmv20}) and non-atomic anonymous games (\citet{msv20}).} As a concluding observation, we note that in all these applications, shocks are typically modeled as arising from continuous distributions with unbounded support, which can only be covered by SMDPs with an unbounded state space.

We now turn to a brief introduction to our principal results. We report three main results in this paper, all of which feature infinite state and action spaces, and two of which feature an unbounded state space: 
\begin{enumerate}
    \item In Theorem \ref{mainresults}, we establish the existence of a Berk-Nash equilibrium for \textit{regular} SMDPs with compact action, state and parameter spaces, and bounded payoff functions, but with unbounded densities (Radon-Nikodym derivatives).\footnote{Given any two Gaussian distributions with distinct variances, the Radon-Nikodym derivative of the one with the larger variance with respect to the other is unbounded.  Moreover, as we see in Example \ref{example-Unit_Root}, unbounded Radon-Nikodym derivatives arise routinely in OLS estimation.} However, the assumption of a compact state space rules out unbounded payoff function such as the CRRA utility function and  distributions that have unbounded support, including normal, exponential and log-normal distributions, which play central roles in economic theory and finance. Theorem \ref{mainresults} therefore needs extension to more general settings in order to admit applications to broader economic environments\fn{In the context of stability theorems for monotone economies, the interested reader may see \cite{ks14} which relaxes the assumption of a compact state space, and therefore admits a broader class of economic models.};
    \item Theorem \ref{mainresultsigma1} considers SMDPs with a $\sigma$-compact state space and a bounded payoff function. We establish the existence of a Berk-Nash equilibrium under a regularity condition on the state space, a tightness condition on the class of transition probability measures and either a uniform integrability or a uniqueness condition on the relative entropy condition formalized as the Kullback-Liebler divergence. The tightness condition that we impose is satisfied by many economic applications (e.g. Ornstein-Uhlenbeck and Cox-Ingersoll-Ross processes) and we provide two sufficient conditions to test its applicability in environments of interest. 
    \item Theorem \ref{mainresultsigma2} pushes the extension further: it allows for unbounded payoff functions\fn{See \cref{assumptionpayoff} and \cref{assumptionpayoffubd} in \cref{secmainresult} for further details.}. The Bellman equation may not have a solution when the underlying payoff function is unbounded. We impose several growth conditions to ensure that a solution of the Bellman equation exists and is well-behaved. These conditions are satisfied in many important examples ranging from various economic fields, as we have illustrated in Examples \ref{stochasticgrowth} and \ref{example-Unit_Root2}. Finally, Theorem 4 provides a possible learning foundation for SMDPs with compact state and action spaces that generalizes Theorem 2 in EP.\fn{It depends on a strong condition (see \cref{defunfcvg}) whose conceptual and technical underpinnings need further consideration. We report it in the \hyperref[sec:onlineappendix]{Online Appendix.}
}
\end{enumerate}

The proofs make use of nonstandard analysis, a powerful mathematical technique that originated in \citet{robinson}, and was introduced into mathematical economics in \citet{br72-ns}.  Readers whose expertise does not extend to nonstandard analysis should note the following points to the four results that we report here:

\ben   \item   The results are  standard results in that their proofs can in principle be furnished without any reference to nonstandard analysis – this is a consequence of a  meta-theorem in mathematical logic that guarantees  that there exists a standard proof, albeit a long and convoluted one,  for any nonstandard one.\fn{As we will explain in \cref{sechypresent}, it may be possible to give a standard proof of Theorem \ref{mainresults} but one needs to overcome several obstacles, especially in the case with unbounded Radon-Nikodym derivatives. We do not know how to give tractable standard proofs for Theorem \ref{mainresultsigma1} and \ref{mainresultsigma2}, as we will explain in \cref{sechypresentsigma}.}

\item Previous applications of nonstandard analysis to probability depend on theorems to lift and push down results between hyperfinite probability spaces and the (standard) Loeb measure spaces they generate, as well as between Loeb spaces and conventional probability measures such as Lebesgue or Wiener measures\footnote{
The previous applications include a complete theory of It\^{o} Processes (See \citet{andersonisrael}), stochastic differential equations (See \citet{Keisler87}), as well as a contribution on existence of equilibrium in continuous-time financial markets (\citet{anderson08}).}. The lifting and pushing down theorems for general Markov processes developed in \citet{drw21}, play essential roles in this paper. 
Markov Processes are involved in many economic problems, and we think it very likely that this work will generate many further economic applications.

\item Previous applications of nonstandard analysis in mathematical economics showed that results that are true in infinite settings but false in finite settings are approximately true in large finite settings.\footnote{For previous applications of nonstandard analysis to mathematical economics, see for example, \citet{nsexchange}, \citet{strongcore}, \citet{oligopoly}, \citet{ks01}, \citet{indmatching}, \citet{ar08}, and \citet{duffie18}.} Here, by contrast, we take results that are {\it true} in finite settings and transport them  to results in infinite settings.\fn{To repeat,  this method is applicable in situations in which the desired result is known for finite objects, its {\it proof} depends heavily on finiteness, but its {\it statement} makes sense for infinite objects.} This approach, pioneered in \citet{duanmuthesis}, works well in situations in which the statement of the result makes sense in the infinite setting, but the proof in the finite case does not readily extend.\footnote{Duanmu's technique has previously been applied to statistical decision theory (\citet{nsbayes,duanmu22}), Markov processes (\citet{drw21} and \citet{ads21avg}), and to abstract economies and Walrasian equilibrium (\citet{adku2,adku1}).}   
\een

We now conclude this introduction by laying out the plan of the paper. \cref{secenviron} furnishes the conceptual framework and the antecedent theory by EP on Markov decision processes with misspecification. \cref{secmainresults} motivates the various assumptions we make on it, and presents our three existence theorems that cover compact and $\sigma$-compact state spaces. Further, we provide two additional illustrations that are relevant for settings widely used in economics. \cref{secmethod} lays out the methodological innovations of the paper and outlines the difficulty with the standard approach and \cref{sketchofproofs} sketches the proofs. \cref{secdiscussion} briefly discusses the extensions of our single-agent results to a broader class of multi-agent misspecified environments. The Online Appendix gives a supplementary result that furnishes a learning foundation to the existence results and also contains the detailed analysis of the examples. The Appendix contains self-contained proofs of all our main results.

\section{The Basic Environment}\label{secenviron}

\subsection{Notational and Conceptual Preliminaries}
\nt We begin by describing the environment faced by the agent which mirrors the one in EP. At the start of each period $t=0,1,2, \ldots,$ the agent observes a state $s_{t}\in S$, takes an action $x_{t}\in X$ that determines the distribution of the future state $s_{t+1}$ given the transition probability function   $Q(\cdot|s_{t},x_{t})$ with the initial state $s_{0}$, drawn according to the initial probability distribution $q_{0}.$ For a given payoff function $\pi(s_{t},x_{t},x_{t+1}),$ the agent then maximizes her expected discounted utility by choosing a feasible policy function. We now formally describe these objects. 

\begin{definition}\label{defMDP}
A Markov Decision Process (MDP) is a tuple $\langle S,X,q_0,Q,\pi,\delta \rangle$, where 
\begin{enumerate}[{\normalfont (i)}, topsep=1pt]
 \setlength{\itemsep}{-2pt} 
    \item The state space $S$ is a $\sigma$-compact locally compact metric space with Borel $\sigma$-algebra $\BorelSets S$;
    \item The action space $X$ is a compact metric space with Borel $\sigma$-algebra $\BorelSets X$;
    \item The initial distribution of states $q_0$ is a probability measure on $(S, \BorelSets S)$;
    \item $Q: S\times X\to \PM{S}$ is a transition probability function, where $\PM{S}$ denotes the set of probability measures on $S$.
    That is, for each $(s,x)\in S\times X$, $Q(s,x)$ is a probability measure on $S$. We sometimes write $Q(\cdot|s, x)$ for $Q(s, x)(\cdot)$;
    \item $\pi: S\times X\times S\to \Reals$ is the per-period payoff function;
    \item The discount factor $\delta$ is in $[0,1)$.
\end{enumerate}
\end{definition}

By the principle of optimality, the agent's problem can be cast recursively as
\begin{equation}
\label{bellman}
V(s)=\max_{x\in X}\int_{S}\{\pi(s,x,s')+\delta V(s')\}Q(\dee s'|s, x), 
\end{equation}
where $V$ is the unique solution to the Bellman equation \cref{bellman}.\fn{Unlike the case considered in EP, the Bellman equation \cref{bellman} need not have a solution, especially when the payoff function is unbounded. In \cref{secmainresult}, we provide regularity conditions, which are shown to be satisfied by examples span over various fields of economics, to guarantee the existence of a solution for the Bellman equation even when the payoff function is unbounded.} We use MDP($Q$) to refer to Markov Decision Process with transition probability function $Q$.

\begin{definition}
An action $x$ is optimal given $s$ in the MDP($Q$) if 

\begin{equation}
\label{policycorr}
x\in \argmax_{\hat{x}\in X}\int_{S}\{\pi(s,\hat{x},s')+\delta V(s')\}Q(\dee s'|s, \hat{x})
\end{equation}
\end{definition}

\nt We next describe a subjective Markov Decision Process.
\begin{definition}\label{defsmdp}
A subjective Markov Decision Process is a Markov Decision Process $\langle S,X,q_0,Q,\pi,\delta \rangle$, and a nonempty family $\mathcal{Q}_{\Theta}=\{Q_{\theta}: \theta\in \Theta\}$ of transition probability functions, where each transition probability function $Q_{\theta}: S\times X\to \PM{S}$ is indexed by an element $\theta\in \Theta$. A subjective Markov Decision Process is said to be \textit{misspecified} if $Q\notin Q_{\Theta}.$\fn{In the language of the everyday, $(Q_{\Theta},\Theta)$ is the set of models.}
\end{definition}

We write SMDP($\langle S,X,q_0,Q,\pi,\delta \rangle$, $\mathcal{Q}_{\Theta}$) to denote a subjective Markov Decision Process with the Markov Decision Process $\langle S,X,q_0,Q,\pi,\delta \rangle$ and the family $\mathcal{Q}_{\Theta}$ of transition probability functions. For all $\theta\in \Theta$, all $(s, x)\in S\times X$, let $D_{\theta}(\cdot|s, x): S\to \bar{\Reals}$ be the density function if $Q(s, x)$ is dominated by $Q_{\theta}(s, x)$ and let $D_{\theta}(s'|s, x)=\infty$ otherwise.\fn{We use $\bar{\Reals}$ to denote the extended real line, equipped with the one-point compactification topology}

\begin{definition}\label{regsmdp}
A regular subjective Markov decision process (regular-SMDP $\mathcal M$) is a SMDP that satisfies the following conditions: 
\begin{enumerate}[{\normalfont (i)}, topsep=1pt]
 \setlength{\itemsep}{-2pt}
    \item The parameter space $\Theta$ is a compact metric space;
    \item The mapping $(s, x)\to Q(s, x)$ is continuous in the Prokhorov metric;
    \item The mapping $(\theta, s, x)\to Q_{\theta}(s, x)$ is continuous in the Prokhorov metric;
    \item The density function $D_{\theta}(s'|s, x)$ is jointly continuous on the set $\{(\theta, s', s, x): Q(s, x)\ \\ \text{is dominated by}\ Q_{\theta}(s, x)\}$\fn{A probability measure $Q(s, x)$ is dominated by a probability measure $Q_{\theta}(s, x)$ if for any measurable set $S$, $Q_{\theta}(S|s,x)=0,$ then $Q_{\theta}(S|s,x)=0.$ }; 
    \item\label{KLint} (Uniform integrability) For every compact set $S'\subset S$, there exists some $r>0$ such that $\big(D_{\theta}(\cdot|s, x)\big)^{1+r}$ is uniformly integrable with respect to $Q_{\theta}(s, x)$ over the set $\{(\theta, s, x): Q(s, x)\ \text{is dominated by}\ Q_{\theta}(s, x)\}$. That is, for every $\epsilon>0$, there exists $\kappa>0$ such that 
    $$
    \int_{E}\big(D_{\theta_0}(t|s_0, x_0)\big)^{1+r}Q_{\theta_0}(s_0, x_0)(\dee t)<\epsilon
    $$
    if $(\theta_0, s_0, x_0)$ is an element of the set $\{(\theta, s, x)\in \Theta\times S'\times X: Q(s, x) \mbox{ is dominated by } \\ Q_{\theta}(s, x)\}$ and $Q_{\theta_0}(s_0, x_0)(E)<\kappa$\fn{This condition is automatically satisfied if the density functions $D_{\theta}(\cdot|s, x)$ are uniformly bounded over the set $\{(\theta, s, x): Q(s, x)$ \mbox{is dominated by} $Q_{\theta}(s, x)$\}.}; 
    \item (Absolute continuity) There is a dense set $\hat{\Theta}\subset \Theta$ such that $Q(s, x)$ is dominated by $Q_{\theta}(s, x)$ for all $\theta\in \hat{\Theta}$ and $(s, x)\in S\times X$;
    \item The per-period payoff function $\pi:S\times X\times S\to \Reals$ is continuous.
\end{enumerate}
\end{definition}

\begin{remark}
In \cref{KLint}, $\kappa$ depends on both $\epsilon$ and the compact set $S'\subset S$. 
Note that we allow $D_{\theta}(\cdot|s, x)$ to take value $\infty$ even if $Q(s, x)$ is dominated by $Q_{\theta}(s, x)$. So we allow for unbounded continuous density functions even when the state space is compact. 
\end{remark}

\begin{definition}\label{defKLD}
The weighted Kullback-Leibler divergence is a mapping $K_{Q}: \PM{S\times X}\times \Theta\to \bar{\Reals}_{\geq 0}$ such that for any $m\in \PM{S\times X}$ and $\theta\in \Theta$,
$$
K_{Q}(m,\theta)=\int_{S\times X}\mathbb{E}_{Q(\cdot|s, x)}\left[\ln \big(D_{\theta}(s'|s,x)\big)\right]\\m(\dee s, \dee x).
$$
\nt The set of closest parameter values given $m\in \PM{S\times X}$ is the set\footnote{We follow the standard convention in that $\ln(0)\cdot 0=0$ and integral of infinity over a set of measure $0$ is $0$. Further, $\dfrac{0}{0}=0$, $\dfrac{1}{0}=\infty$, $\log \infty =\infty.$ }
$$
\Theta_{Q}(m)=\argmin_{\theta\in \Theta}K_{Q}(m,\theta).
$$
\end{definition}
For $(s, x)\in S\times X$ and $\theta\in \Theta$, the relative entropy (Kullback-Leibler divergence) from $Q_{\theta}(s, x)$ to $Q(s, x)$ is:
$$
\mathcal{D}_{\mathrm{KL}}\big(Q(s, x), Q_{\theta}(s, x)\big)=\mathbb{E}_{Q(\cdot|s, x)}\left[\ln \big(D_{\theta}(s'|s,x)\big)\right].
$$
If $Q(s, x)$ is dominated by $Q_{\theta}(s, x)$, then we have
$$
\mathcal{D}_{\mathrm{KL}}\big(Q(s, x), Q_{\theta}(s, x)\big)=\int_{S}D_{\theta}(s'|s, x)\\ \ln\big(D_{\theta}(s'|s,x)\big) Q_{\theta}(\dee s'|s, x).
$$ 
and otherwise, it equals infinity. Moreover, by \cref{KLint} in \cref{regsmdp}, the function $D_{\theta}(\cdot|s, x)\ln\big(D_{\theta}(\cdot|s,x)\big)$ is integrable  with respect to $Q_{\theta}(s, x)$. 
For $m\in \PM{S\times X}$, let $\Theta_{m}=\{\theta\in \Theta: K_{Q}(m, \theta)<\infty\}$. 
By \cref{regsmdp}, we have $\hat{\Theta}\subset \Theta_{m}$ and $K_{Q}(m, \theta)$ is a continuous function of $\theta$ on $\Theta_{m}$.  
Finally, by Jensen's inequality, the relative entropy $\mathcal{D}_{\mathrm{KL}}\big(Q(s, x), Q_{\theta}(s, x)\big)$ is non-negative for all $(s, x)\in S\times X$. 

\begin{definition}
A probability measure $m\in \PM{S\times X}$ is a Berk-Nash equilibrium of the SMDP($\langle S,X,q_0,Q,\pi,\delta \rangle$, $\mathcal{Q}_{\Theta}$) if there exists a belief $\nu\in \PM{\Theta}$ such that
\begin{enumerate}[{\normalfont (i)}, topsep=1pt]
 \setlength{\itemsep}{-2pt}
    \item \textbf{Optimality}: For all $(s,x)\in S\times X,$ that is in the support of $m$, $x$ is optimal given $s$ in the MDP($\bar{Q}_{\nu}$), where $\bar{Q}_{\nu}=\int_{\Theta}Q_{\theta}\nu(\dee \theta)$;
    \item \textbf{Belief Restriction}: We have $\nu\in \PM{\Theta_{Q}(m)}$;
    \item \textbf{Stationarity}: For all $A\in \BorelSets S$, $m_{S}(A)=\int_{S\times X}Q(A|s, x)m(\dee s, \dee x)$, where $m_{S}$ denote the marginal measure of $m$ on $S$.\fn{The MDP (SMDP) is stationary in the classical Blackwell sense.}
\end{enumerate}
\end{definition}

\subsection{Antecedent Results: Esponda-Pouzo (2021)}

In this section, we begin by illustrating two environments from EP that naturally feature infinite state and action spaces. Example \ref{stochasticgrowth} illustrates an optimal consumption-savings problem while Example \ref{costsproduction} frames an example where the costs of a producer are misspecified. It is imperative to note that while the illustrations we provide extend the environment of these examples to more general environments, the original environments are themselves outside the realm of EP's original finite existence theorem. Following the illustrations, we state the main existence theorem in EP for SMDPs with finite state and action spaces.

\ex[Optimal Savings (Example 2), Esponda-Pouzo (2021)]
 \label{stochasticgrowth}\normalfont  The Markov decision process is as follows. A state space $S=(y,z)\in Y\times Z=(0, \infty)\times [0,1]$, where $y$ and $z$ denote the wealth and preference shocks, respectively. For each $y \in Y$, the agent chooses $x \in X = [0,1]$, with $x$ representing the \textit{fraction} of $y$ the agent chooses to save, so that the agent saves $k=xy$ and consumes $y-k$.\footnote{In EP, the agent chooses how much $k \in [0,y]$ to save.  Here, we recast the problem in terms of the fraction saved in order to ensure that the action set $X$ is compact and independent of the state.}  The payoff function $\pi$ is $\pi\left(y, x, z\right)=z \ln \left(y-k\right)=z \ln \left(y-xy\right).$\fn{When $z=0$, we again use the standard convention that $0\ln 0=0$. When $z\neq 0$, we approximate the action space $X=[0,1]$ by closed intervals $\{[0,1-\epsilon]: \epsilon>0\}$.} 
 
 We next describe the true transition function that describes the evolution of the state variables. $Q\left(y^{\prime}, z^{\prime} \mid y, z, x\right)$ is such that $y^{\prime}$ and $z^{\prime}$ are independent, $y^{\prime}$ has a log-normal distribution with mean $\alpha^{*}+\beta^{*} \ln (xy)+\gamma^{*} z$ and unit variance, and $z^{\prime}$ is \textit{uniform} on $[0,1]$. That is, the next period wealth, $y_{t+1}$, is given by $\ln y_{t+1}=\alpha^{*}+\beta^{*} \ln x_{t}y_{t}+\varepsilon_{t},$ where $\varepsilon_{t}=\gamma^{*} z_{t}+\xi_{t}$ is an unobserved i.i.d. productivity shock, $\xi_{t}\sim N(0,1), \gamma^{*}\neq 0,$ and $0\leq\beta^{*}<1,\ \delta\beta^{*}<1,$ where $\delta \in[0,1)$ is the discount factor. The agent maximizes their discounted expected utility by choosing optimal proportion of savings, $x.$\fn{It is the restriction on $0\leq\beta<1$ that gives us stationarity. The detailed analysis is in the Supplementary Appendix.} The Bellman equation for this MDP is as follows.

$$
V(y, z)=\max _{0 \leq x \leq 1} z \ln (y-xy)+\delta \expect\left[V\left(y^{\prime}, z^{\prime}\right) \mid x\right],
$$

However, the agent believes (SMDP) that $\ln y_{t+1}=\alpha+\beta \ln (x_{t}y_{t})+\varepsilon_{t}$ where $\varepsilon_{t} \sim N(0,1)$ and is independent of the preference shock. Further, the agent knows the distribution of the preference shock but is uncertain about $\beta \in \Theta$.\fn{A compact set in  $\mathbb{R}.$}  The \textit{subjective} transition probability function $Q_{\theta}\left(y^{\prime}, z^{\prime} \mid y, z, x\right)$ is such that $y^{\prime}$ and $z^{\prime}$ are independent, $y^{\prime}$ has a log-normal distribution with mean $\alpha+\beta \ln (xy)$ and unit variance, and $z^{\prime}$ is \textit{uniform} on $[0,1]$. The agent has a misspecified model since she believes that the productivity and utility shocks are independent, when in fact $\gamma^{*} \neq 0$. Here we diverge from the EP example by having preference shocks $z$ distributed uniformly over $[0,1].$ This example extends EP's example for a continuum of preference shocks in an optimal consumption-savings model.


\exx

\ex[Misspecified Costs (Example 3), Esponda-Pouzo (2021)]
\label{costsproduction}\normalfont

 Consider the following Markov decision process.
Every period, an agent observes a productivity shock $z \in \mathbb{Z}=[0,1]$ and chooses an input $x \in X$  $\subset \mathbb{R}_{+} $ which results in the agent obtaining a payoff of $r(x)-$ $c(x)$ every period, where $c(x)=\phi(x)\epsilon$ is the cost of choosing $x$,  $r(x)=z\ln(x)$ where $\ln(x)$ is the production function, $z$ is the productivity shock in $[0,1]$ and $\epsilon$ is a random, independent shock to the cost  distributed according to the (true) distribution $d^{*}$, which has support equal to $[0, b], 0\leq b \leq \infty \mbox{ and } 0< \expect_{d^{*}}[\epsilon] < \infty.$\fn{We assume that the true distribution $d^{*}$ satisfies conditions in \cref{regsmdp}.} The state space $S=[0,1]\times [0,b], b<\infty $  is the support of the cost shock. The action space $X$ and the parameter space $\Theta$ are chosen as such to be compact\footnote{The details are supplied in the Online Appendix.} and the payoff function $\pi(s,x,s')=z\ln x-c(x)$.
The Bellman equation is given by:
$$
V(z, \epsilon)=\max _{x} \int_{[0,1]\times [0,b]}\left(z f(x)-c^{\prime}+\delta V\left(z^{\prime}, \epsilon'\right)\right) Q\left(\dee z^{\prime} \mid z\right) Q^{C}\left(\dee \epsilon' \mid x\right)
$$
Let $Q\left(z^{\prime} \mid z\right)$ be the probability that tomorrow's productivity shock is $z^{\prime}$ given the current shock $z$. We assume that there is a unique stationary distribution over these productivity shocks which is uniform,  $U[0,1]$. Similarly, let $Q^{C}(\epsilon'\mid x)$ denote the transition function for the cost shock, $\epsilon'.$ The agent believes in a misspecified cost function (SMDP), that is,  $c_{\theta}(x)=x \epsilon$ and $\epsilon \sim d_{\theta}$, where $d_{\theta}$ has support equal to $[0, b]$ where $b=k\theta,$ $0\leq k<\infty.$  We assume that $\epsilon$ follows a truncated exponential distribution, $d_{\theta}(\epsilon)=\dfrac{(1 / \theta) e^{-(1 / \theta) \epsilon}}{1-e^{(-b/\theta)}}.$\fn{In the context of this particular example, the agent's model can be misspecified if either cost functions are nonlinear, true distribution of cost shocks are not a part of the exponential family, or if the support assumed is incorrect.} This example extends EP finite productivity shocks to a continuum of shocks in the realm of a producer's problem. However, instead of having unbounded support for the cost shock as in EP, we restrict it to a bounded support. 
\exx

Both these examples are outside the scope of the existence theorem in EP with finite states and actions. We now spell out their existence result for finite SMDPs. It is straightforward to verify that regular-SMDPs with finite state and action spaces, as defined in EP, are regular in the sense of \cref{regsmdp}. EP's Theorem 1 proves the following result.

\begin{theorem*}[EP (2021)]\label{EPmain}
Suppose ($\langle S,X,q_0,Q,\pi,\delta \rangle$, $\mathcal{Q}_{\Theta}$) is a regular-SMDP such that
\begin{enumerate}[{\normalfont (i)}, topsep=1pt]
 \setlength{\itemsep}{-2pt}
    \item The state space $S$ and the action space $X$ are both finite;
    \item The parameter space is a compact subset of Euclidean space.
\end{enumerate}
Then there exists a Berk-Nash equilibrium. 
\end{theorem*}

Note that every finite set can be embedded into a Euclidean space. Thus, the following \textit{finite} result is an immediate consequence of the above theorem. 
\begin{lemma}\label{EPfinitelemma}
Suppose $(\langle S,X,q_0,Q,\pi,\delta \rangle$, $\mathcal{Q}_{\Theta})$ is a regular-SMDP with finite state, action and parameter spaces.
Then there exists a Berk-Nash equilibrium. 
\end{lemma}

\section{The Main Results and Applications}\label{secmainresults}

In this section, we present the three main existence results: the first pertains to a compact state space and the other two, to $\sigma$-compact state spaces. We consider three more substantial examples on neoclassical producer theory and econometric theory, and establish the existence of Berk-Nash equilibria for all examples in this paper by applying our main theorems. 

\subsection{The Main Results}
\label{secmainresult}

Our first main result extends the finite existence result of EP to a broader class of environments, SMDPs with a compact state space and unbounded densities (Radon-Nikodym derivatives). 

\thm
\label{mainresults}
Every regular-SMDP $\mathcal M$ with a compact state space has a Berk-Nash equilibrium. 
\thmm
We extend Theorem \ref{mainresults} to a regular SMDP $\mathcal{M}=(\langle S,X,q_0,Q,\pi,\delta \rangle$, $\mathcal{Q}_{\Theta})$ with a non-compact state space $S$. We assume that the density functions $\{D_{\theta}(\cdot|s, x)\}$ take value in $\Reals$. 
That is, for all $\theta\in \Theta$ and all $(s, x)\in S\times X$, let $D_{\theta}(\cdot|s, x): S\to \Reals$ be the density function if $Q(s, x)$ is dominated by $Q_{\theta}(s, x)$ and let $D_{\theta}(s'|s, x)=\infty$ if $Q(s, x)$ is not dominated by $Q_{\theta}(s, x)$. 
We start with the following assumption on the state space $S$. 

\begin{assumption}\label{assumptionstate}
There exists a non-decreasing sequence $\{S_n\}_{n\in \Nats}$ of compact subsets of $S$ such that 
\begin{enumerate}[{\normalfont (i)}, topsep=1pt]
 \setlength{\itemsep}{-2pt}
    \item $\bigcup_{n\in \Nats}S_n=S$;
    \item $q_0(S_n)>0$ for all $n\in \Nats$;
    \item\label{minbound} There exists $r>0$ such that $Q(s, x)(S_n)>r$ and $Q_{\theta}(s, x)(S_n)>r$ for all $n\in \Nats$, all $(s, x)\in S_n\times X$ and all $\theta\in \Theta$;
    \item \label{contset} For all $n\in \Nats$, $S_n$ is a continuity set of $Q(s, x)$ and $Q_{\theta}(s, x)$ for all $(s, x)\in S_n\times X$ and all $\theta\in \Theta$.
\end{enumerate}
\end{assumption}

\cref{assumptionstate} imposes four technical conditions on the state space that are satisfied for most applications in the literature.\fn{Unlike the finite proof in EP which doesn't involve $q_{0}$, we use it here to deconstruct the state space into subsets with positive initial state visitations.} It requires that the state space $S$ can be deconstructed into a countable, non-decreasing sequence of subsets such that their union is the state space $S$. Items \ref{minbound} and \ref{contset} of \cref{assumptionstate} jointly imply that the true and model transition probability functions are well-behaved for the truncation of the SMDP $\mathcal{M}$ defined on the sequence $\{S_n\}_{n\in \Nats}$. Thus, it is natural to approximate $\mathcal{M}$ using a sequence of SMDPs with state spaces $\{S_n\}_{n\in \Nats}$. In particular, for $n\in \Nats$, define $\mathcal{M}_{\Theta'}^{n}=(\langle S_n,X,q_0^{n},Q^{n},\pi_{n},\delta \rangle$, $\mathcal{Q}^{n}_{\Theta'})$ to be the SMDP such that:
\begin{enumerate}[{\normalfont (i)}, topsep=1pt]
 \setlength{\itemsep}{-2pt}
    \item The state space is $S_n$, endowed with Borel $\sigma$-algebra $\BorelSets {S_n}$;
    \item The action space is $X$, endowed with Borel $\sigma$-algebra $\BorelSets X$;
    \item $q_0^{n}(A)=\frac{q_0(A)}{q_0(S_n)}$ for all $A\in \BorelSets {S_n}$;
    \item The parameter space $\Theta'$ is a finite subset of $\hat{\Theta}$;
    \item $Q^{n}: S_n\times X\to \PM{S_n}$ is the transition probability function defined as $Q^{n}(s, x)(A)=\frac{Q(s, x)(A)}{Q(s, x)(S_n)}$ for all $A\in \BorelSets {S_n}$;
    \item For every $\theta\in \Theta'$, $Q_{\theta}^{n}: S_n\times X\to \PM{S_n}$ is defined as 
    $Q_{\theta}^{n}(s, x)(A)=\frac{Q_{\theta}(s, x)(A)}{Q_{\theta}(s, x)(S_n)}$ for all $A\in \BorelSets {S_n}$ and  
    let $\mathcal{Q}^{n}_{\Theta'}=\{Q_{\theta}^{n}: \theta\in \Theta'\}$; 
    \item $\pi_{n}: S_n\times X\times S_n\to \Reals$ is the restriction of $\pi$ to $S_n\times X\times S_n$;
    \item  $\delta\in [0,1)$ is the discount factor.
\end{enumerate}

\begin{remark}\label{doubleapprox}
For $\theta\in \Theta$, it is possible that $Q^{n}(s, x)$ is dominated by $Q_{\theta}^{n}(s, x)$ but $Q(s, x)$ is not dominated by $Q_{\theta}(s, x)$. 
Thus, we need to approximate the state and parameter spaces of the full SMDP $\mathcal{M}$ by carefully chosen subsets, simultaneously. 
So we choose to approximate the state space $S$ by the sequence $\{S_n\}_{n\in \Nats}$ of compact sets and approximate the parameter space $\Theta$ by finite subsets of $\hat{\Theta}$. 
\end{remark}

To ensure that the Markov decision process has a stationary measure, a sequence of stationary measures for the truncated Markov decision processes should have a convergent subsequence. 
For every $n\in \Nats$ and every $P\in \PM{S_n\times X}$, let $R_{n}(P)$ be the probability measure on $S_n$ such that 
$
R_{n}(P)(A)=\int_{S_n\times X}Q^{n}(A|s,x)P(\dee s, \dee x).
$
Let $P_{S}$ denote the marginal measure of $P$ on $S$. 
We impose the following tightness assumption: 

\begin{assumption}\label{assumptiontight} (Tightness)
The family 

$$
\mathcal{R}=\{R_{n}(P): n\in \Nats, P\in \PM{S_n\times X}, P_{S}=R_{n}(P)\}
$$ 
is tight.
\end{assumption}
\nt\cref{assumptiontight} ensures that any sequence of stationary measures for the truncated transition probability functions $\{Q^{n}(s, x): s\in S_n, x\in X\}$ is tight, which further implies that any sequence of stationary measures has a convergent subsequence. Tightness may sometimes be hard to verify directly and therefore, we provide two sufficient conditions. The second sufficient condition is satisfied for most applications.

\vspace{1em}

\nt{\textbf{Condition 1 (Reversible)}:}  The transition probability function $\{Q(s, x)\}$ has a unique stationary measure $\pi$ and is \textit{reversible} with respect to $\pi$. That is, there exists a unique $\pi\in \PM{S\times X}$ such that
    $
    \pi_{S}(A)=\int_{S\times X}Q(s, x)(A)\pi(\dee s, \dee x)
    $
    for all $A\in \BorelSets S$. Moreover, for all $A_1, A_2\in \BorelSets S$, we have
    $
    \int_{A_1\times X}Q(s, x)(A_2)\pi(\dee s, \dee x)=\int_{A_2\times X}Q(s, x)(A_1)\pi(\dee s, \dee x). 
    $

\vspace{1em}

\nt{\textbf{Condition 2 (Lyapunov)}:}  The transition probability function $\{Q(s, x)\}$ satisfies the 
    \emph{Lyapunov condition}, that is, there exist a non-negative continuous norm-like function $V,$\footnote{A function $V: S\to \NNReals$ is \textit{norm-like} if $\{s\in S: V(s)\leq B\}$ is precompact for every $B>0$.}  and constants $0<\alpha\leq1$, $\beta\geq 0$ such that 
    $\label{lyapueq}
    \int_{S}V(y)Q(s, x)(\dee y)\leq (1-\alpha)V(s)+\beta
    $
    for all $s\in S$ and $x\in X$. Moreover, the sequence $(\{s\in S: V(s)\leq n\})_{n\in \Nats}$ of sets satisfies \cref{assumptionstate}. Then, by taking $S_n=\{s\in S: V(s)\leq n\}$, \cref{assumptiontight} is satisfied. 

\vspace{0.5em}

To establish belief restriction for the SMDP $\mathcal{M}$, 
we impose the following assumption on the relative entropy.
\begin{assumption}\label{assumptionrebound} (Uniform-integrability)
For all $\theta\in\hat{\Theta},$ the family of relative entropy
$\{\mathcal{D}_{\mathrm{KL}}\big(Q(s, x), Q_{\theta}(s, x)\big)\}$ is uniformly integrable with respect to all stationary $P\in \Delta(S\times X)$. 
That is, for every $\epsilon>0$, there exists $\kappa>0$ such that 
$
\int_{E}\mathbb{E}_{Q(\cdot|s, x)}\left[\ln\big(D_{\theta}(s'|s, x)\big)\right]P(\dee s, \dee x)<\epsilon
$
for all $\theta\in \hat{\Theta}$ and all stationary $P\in \Delta(S\times X)$ with $P(E)<\kappa$.\fn{One sufficient condition for \cref{assumptionrebound} is to assume that the relative entropy is uniformly bounded on the set $\{(\theta, s, x)\in \Theta\times S\times X: Q(s, x)\ \text{is dominated by}\ Q_{\theta}(s, x)\}$. If the true transition probability function $Q$ and every element in $\mathcal{Q}_{\Theta}$ do not depend on the current state,  as the action and parameter spaces are compact, this sufficient condition is usually satisfied.}
\end{assumption}

The candidate Berk-Nash equilibrium for $\mathcal{M}$ is the weak limit of Berk-Nash equilibrium for the sequence of truncated SMDPs. \cref{assumptionrebound} allows us to approximate the weighted Kullback-Leibler divergence of $\mathcal{M}$ from the weighted Kullback-Leibler divergence of truncated SMDPs. Alternatively, we can also establish belief restriction under the following assumption.

\begin{assumption}\label{assumptionuniquemin} (Uniqueness)
There exists a unique $\theta_0\in \Theta$ that minimizes the relative entropy $\mathcal{D}_{\mathrm{KL}}\big(Q(s, x), Q_{\theta}(s, x)\big)$ for all $(s, x)\in S\times X$. Moreover, for every $n\in \Nats$, $\theta_0$ uniquely minimizes $\mathcal{D}_{\mathrm{KL}}\big(Q^{n}(s, x), Q_{\theta}^{n}(s, x)\big)$ for all $(s, x)\in S_n\times X$.
\end{assumption}

\begin{remark}
Under \cref{assumptionuniquemin}, the set of closest parameters for $\mathcal{M}$ and all truncated SMDPs is the same singleton set $\{\theta_0\}$.  
If the model is \textit{correctly specified}, then \cref{assumptionuniquemin} is trivially satisfied with the set of the closest parameters being a singleton, that is, the true parameter value. 
\end{remark}

We establish optimality under two different set of conditions: Theorem 2 assumes the payoff function is bounded continuous while Theorem 3 allows for unbounded payoff functions under a fairly general norm-restriction assumption on the state space. 
For the bounded payoff function case, we assume: 

\begin{assumption}\label{assumptionpayoff} (Boundedness)
The payoff function $\pi: S\times X\times S\to \Reals$ is a bounded continuous function.
\end{assumption}

We now present our second main result on the existence of Berk-Nash Equilibrium for SMDPs with a bounded payoff function and a $\sigma$-compact state space, thus, allowing instances when the state space is the real line, $\mathbb R.$ 

\thm
\label{mainresultsigma1} Any regular SMDP that satisfies \cref{assumptionstate}, tightness (\cref{assumptiontight}) and has a bounded payoff function (\cref{assumptionpayoff}) has a Berk-Nash equilibrium if either the SMDP 
is correctly specified or if one of the assumptions of uniform integrability (\cref{assumptionrebound}) or uniqueness (\cref{assumptionuniquemin}) holds. 
\thmm

Unbounded payoff functions appear naturally in many important economic applications;\fn{See \cite{ry74} and the subsequent response to it in \cite{ar74} for further elaboration on the relevance of unbounded utility functions in the context of  expected utility theory. The interested reader is referred to a discussion on these issues in \cite{pst20}.} our  Examples \ref{stochasticgrowth} (optimal savings) and \ref{example-Unit_Root2} (AR (1) process) feature unbounded payoff functions. However, the Bellman equation need not have a solution when the payoff function is unbounded. We impose a few assumptions to ensure the existence of a solution for the Bellman equation. Moreover, these assumptions allow for approximation of the solution of the Bellman equation for $\mathcal{M}$ by solutions of Bellman equations for truncated SMDPs. We assume the state space $S$ is a norm space and use $\|s\|$ to denote the norm of an element $s\in S$. The first assumption puts an upper bound on the growth rate of the payoff function:

\begin{assumption}\label{assumptionpayoffubd} (State-boundedness)
The payoff function $\pi: S\times X\times S\to \Reals$ is a jointly continuous function and there exist $A, B\in \PosReals$ such that for all $(s, x, s')\in S\times X \times S'$, $|\pi(s, x, s')|\leq A+B \max\{\|s\|, \|s'\|\}.$\fn{Instances where such an assumption are satisfied are common: (i) a monopolist's payoff in \citet{ny91}, (ii)  CRRA utility  in stochastic growth and optimal savings environments.}

\end{assumption}

The next assumption asserts that the rate of the family of subjective transition probability functions drifting to infinity is bounded by a linear function. 

\begin{assumption}\label{assumptionsint} (Fold-boundedness)
Let $B\in \PosReals$ be given in \cref{assumptionpayoffubd}.
There exist some $C, D\in \PosReals$ such that  $\int_{S}\|s'\|Q_{\theta}(\dee s'|s, x)\leq C+\frac{D}{(1+\delta)B+\delta D}\|s\|$ for all $x\in X$, $\theta\in \Theta$ and $s\in S$.\footnote{It is clear that $\frac{D}{(1+\delta)B+\delta D}\to \frac{1}{\delta}$ as $D\to \infty$. 
If we assume that there exist $C\in \PosReals$ and $D<\frac{1}{\delta}$ such that $\int_{S}\|s'\|Q_{\theta}(\dee s'|s, x)\leq C+D\|s\|$, then it implies \cref{assumptionsint} with a suitably chosen $D'$. } 
\end{assumption}



The final assumption impose a stronger continuity condition on the family $\mathcal{Q}_{\Theta}=\{Q_{\theta}: \theta\in \Theta\}$. 
Let $d_{S}$ denote the metric on $S$ generated from the norm. 

\begin{assumption}\label{assumptioncontwass} (W-continuity)
The mapping $(\theta, s, x)\to Q_{\theta}(s, x)$ is continuous in the $1$-Wasserstein metric. Moreover, $Q_{\theta}(s, x)$ has finite first moment for all $(\theta, s, x)\in \Theta\times S\times X$. That is, for every $(\theta, s, x)\in \Theta\times S\times X$, there exists\fn{By the triangle inequality, under \cref{assumptioncontwass}, we have $\int_{S}d_{S}(t, s')Q_{\theta}(s, x)(\dee t)<\infty$ for all $s'\in S$ and all
$(\theta,s,x)\in \Theta\times S\times X$.  
} some $s_0\in S$ such that $\int_{S}d_{S}(t, s_0)Q_{\theta}(s, x)(\dee t)<\infty$. 
\end{assumption}

Our final main result establishes the existence of a Berk-Nash equilibrium for SMDPs for a sigma-compact state spaces with unbounded payoff function.

\thm
\label{mainresultsigma2}
Theorem 2 holds if the boundedness of the payoff function (\cref{assumptionpayoff}) is weakened to state-boundedness (\cref{assumptionpayoffubd}) but with subjective transition probability functions fold-bounded (\cref{assumptionsint}) and W-continuous (\cref{assumptioncontwass}). 
\thmm

\subsection{Applications to Economic Settings}

In this subsection, in addition to the previous two EP illustrations, we present examples to demonstrate the applicability of our main results to a variety of problems encountered in economic theory. Examples \ref{example-Unit_Root} and \ref{example-Unit_Root2} connect the existence of a Berk-Nash equilibrium with the existence of unit roots for an AR(1) process. Example \ref{pricingshocks} considers a revenue analogue of Example \ref{costsproduction}.

\ex[AR(1) Process with a Bounded Payoff Function] \label{example-Unit_Root}\normalfont
In this example, we show that for a AR(1) process, a Berk-Nash equilibrium exists if and only if the AR(1) process does not have a unit root. In addition, we illustrate that the unbounded density functions arise naturally in correctly specified econometric inference problems.  

Consider a SMDP with  state space $S=\mathbb R,$ a singleton action space $X=\{0\},$ and a payoff function $\pi: S\times X\times S$ that equals 0 for all $(s,x,s)\in S\times X \times S$. For every $s\in S$, the true transition probability function $Q(s)$ is the distribution of $a_0s+b_0\xi$, where $a_0\in [0, 2], b_0\in [0, 1]$ and $\xi=\mathcal{N}(0,1)$ has the standard normal distribution.\fn{In other words, this is simply an inference problem about a Markov process, specifically an AR(1) process, rather than a full Markov {\em decision problem}.} The parameter space $\Theta$ is $[0, 2]\times [0, 1]$ and for every $(a, b)\in \Theta$, the transition probability function $Q_{(a, b)}(s)$ is the distribution of $as+b\xi$.

Consider first the degenerate case $b_0 = 0$. In this case, the evolution of the state is deterministic. 
When $a_0<1$, the Dirac measure $\delta_{(0,0)}$ at $(0,0)$ is a Berk-Nash equilibrium, supported by the belief $\delta_{(a_0, 0)}$. When $a_0=1$, every Dirac measure $\delta_{(s, 0)}$ for $s\in S$ is a Berk-Nash equilibrium supported by the belief $\delta_{(1, 0)}$. There is no Berk-Nash equilibrium if $a_0>1$. 

Now, we turn to the non-degenerate case $b_0 > 0$.  
The true transition probability function $Q(s) = Q_{(a_0,b_0)}(s)$ is absolutely continuous with respect to $Q_{(a, b)}(s)$ for all $(a, b)\in \hat \Theta = [0,2] \times (0,1]$, and the density function is jointly continuous function where it is defined. Note, however, that the density function is unbounded on $\hat \Theta$, tending to infinity as $b \rightarrow 0.$  The unboundedness arises here, even though the model is correctly specified.
This situation arises ubiquitously in econometric inference.  In any OLS estimation, we test (among other things) whether or not the regression coefficient $\alpha_1$ of the dependent variable $y$ on a given independent variable $x_1$ is, or is not, zero. This estimation requires us to include in $\Theta$ the {\em possibility} that $\alpha_1$ is zero.  If $\alpha_1$ is, in fact, not zero, the Radon-Nikodym derivative will typically be unbounded.

It is straightforward to verify that this is a regular SMDP in the sense of \cref{regsmdp}. As the state space is not compact, we need to check the conditions under which this example satisfies the assumptions for Theorem 2. 
We will see that the example satisfies those conditions if and only if $a_0 <1$.  Note that if $a_0 \geq 1$, then since $b_0 \not = 0$, the Markov process has no stationary distribution, and hence there is no Berk-Nash equilibrium. Note also that $a_0 \geq 1$ if and only if the AR(1) process has a unit root, which implies that the usual method for estimating the parameters of the AR(1), ordinary least squares, yields spurious results.  Thus, our example has a Berk-Nash equilibrium if and only if the AR(1) process does not have a unit root.

To see this, suppose $0 \leq a_0 < 1$ (and recall that $b_0 > 0$). For each $n\in \Nats$, let $S_n=[-n, n]$; it is straightforward to see that the sequence $\{S_n: n\in \Nats\}$ satisfies \cref{assumptionstate}. 
\cref{assumptiontight} is satisfied by taking the Lyapunov function $V(s)=|s|$. 
It can be verified that \cref{assumptionrebound} is satisfied, which establishes belief restriction\fn{In \cref{detailexample}, we provide rigorous verification for \cref{assumptiontight} and \cref{assumptionrebound}.}.
We can also establish belief restriction from the fact that the SMDP is correctly specified. 
We can easily modify this example to a SMDP with misspecification, in which case belief restriction follows from \cref{assumptionuniquemin}.
The payoff function is constant, and hence satisfies Assumption \ref{assumptionpayoff}.  By Theorem 2, there exists a Berk-Nash equilibrium for this SMDP. 
The Berk-Nash equilibrium is $\mu\times \delta_{0}$ where $\mu=\mathcal{N}(0, \frac{b_0^2}{1-a_0^2})$, supported by the belief $\delta_{(a_0,b_0)}$.

\exx

\ex[AR(1) Process with an Unbounded Payoff Function]
\label{example-Unit_Root2}\normalfont
We modify Example \ref{example-Unit_Root} by setting the action space $X=[-1,1]$ and $\Theta = [0,2] \times [0,1] \times [-1,1]$. The true probability transition $Q(s, x)$ has the distribution of $a_0s+b_0\xi+c_0 x$, with $c_0 \in [-1,1]$, and the payoff is $\pi(s,x,s') = s'$. For every $(a, b, c)\in \Theta$, the transition probability function $Q_{(a,b,c)}(s, x)=as+b\xi+cx$.  
The degenerate case $b_0 = 0$ is handled in the same way as in Example \ref{example-Unit_Root}. 

Now suppose that $b_0>0$.  If $a_0\geq 1$, then since $b_0\neq 0$, the Markov decision process has no stationary distribution, and hence there is no Berk-Nash equilibrium.
If $a_0<1$, we restrict $\Theta$ to $\Theta' = [0,1] \times [0,1] \times [-1,1]$, and verify that the Assumptions of Theorem \ref{mainresultsigma2} are satisfied for this modified SMDP.
Letting $V(s)=\|s\|$ and $S_n=[-n, n]$, Condition 2 for  \cref{assumptiontight} is satisfied by essentially the same calculation as in Example \ref{example-Unit_Root}.
The payoff function $\pi$ clearly satisfies \cref{assumptionpayoffubd} and \cref{assumptioncontwass}.
By a similar calculation as in Example \ref{example-Unit_Root}, \cref{assumptionrebound} is satisfied. We can establish belief restriction from the fact that this SMDP is correctly specified or via \cref{assumptionuniquemin}.
By essentially the same calculation as in Example \ref{example-Unit_Root}, \cref{assumptionsint} is satisfied. Thus, the restricted SMDP has a Berk-Nash equilibrium by Theorem 3.

When $b_0 > 0$ and $a_0<1$, observe that the action choice  $x= \sign{c_0}$\footnote{$\sign x = x/|x|$ when $x \not = 0$, $\sign 0 = 0$.} is a dominant strategy,\footnote{It is weakly dominant if $c_0=0$, dominant otherwise.} so let $\mu$ be the unique stationary distribution on $S$ induced by the action choice $\sign c_0$.  When $c_0 \not = 0$, the Berk-Nash equilibrium is $\mu \times \delta_{\sign c_0}$; when $c_0 = 0$, the set of Berk-Nash equilibria is $\mu \times \Delta(X)$; in both cases, the Berk-Nash equilibria are supported by the belief $\delta_{(a_0,b_0,c_0)}$.
Note that, in this example, the set of closest parameter values $\Theta'_{Q}(\mu)=\{(a_0,b_0,c_0)\}$ for the restricted SMDP is the same as the set of closest parameter values $\Theta_{Q}(\mu)$ for the original SMDP. Hence, the equilibrium is a Berk-Nash equilibrium of the original SMDP. Therefore, by Theorem 3, the problem has a Berk-Nash equilibrium if and only if $a_0 < 1$, i.e. if and only if the problem does not have a unit root.

\exx

\ex[Misspecified Revenue]
\label{pricingshocks}\normalfont
In this example, we incorporate misspecification in the payoff function with misspecified pricing shocks. The Markov Decision Process is as follows. Every period, an agent observes a productivity shock $z \in \mathbb{Z}=[0,1]$ and chooses an input $x \in X$  $\subset \mathbb{R}_{+} $ which results in the agent receiving a payoff of $r(x)-$ $c(x),$ where $c(x)=x^{2}$ is the cost of choosing $x$, and $r(x)=zf(x)\epsilon$ where $f(x)$ is the production function, $\epsilon$ is a random, independent shock to the price (which we set as 1) distributed according to the (true) distribution $d^{*}$, which has support equal to $[0, b], 0\leq b \leq \infty \mbox{ and } 0<\expect_{d^{*}}[\epsilon] < \infty.$\fn{We assume that the true distribution $d^{*}$ satisfies conditions in \cref{regsmdp}.} Therefore, the state space is given by, $(z,\epsilon)\in S=[0,1]\times [0,b].$  Let $Q\left(z^{\prime} \mid z\right)$ be the probability that tomorrow's productivity shock is $z^{\prime},$ given the current shock $z$ and similarly, let $Q^{R}(\epsilon'\mid x)$ denote the transition function for the price shock, $\epsilon'$. We follow EP in framing the price shock $\epsilon$ as a part of the state variables along with the productivity shock $z$ and define the Bellman equation below.
$$
V(z, \epsilon)=\max _{x} \int_{[0,1]\times [0,b]}\left(z f(x)\epsilon'-c+\delta V\left(z^{\prime}, \epsilon^{\prime}\right)\right) Q\left(\dee z^{\prime} \mid z\right) Q^{R}\left(\dee \epsilon^{\prime} \mid x\right)
$$
We assume that there is a unique stationary distribution over these productivity shocks, denoted by $z\sim U[0,1]$. Next, we describe the SMDP of our environment. The agent believes (SMDP) that $f(x)=
x $ and $\epsilon \sim d_{\theta}$, where $d_{\theta}$ has support equal to $[0, b]$ where $b=k\theta.$ The parameter space $\Theta$ and the action space $X$ are chosen as such to be compact.\fn{The details are supplied in the Online Appendix.} We assume that $\epsilon$ follows a truncated exponential distribution, $d_{\theta}(\epsilon)=\dfrac{(1 / \theta) e^{-(1 / \theta) \epsilon}}{1-e^{(-b/\theta)}} .$  Here, the agent's model can be misspecified if either the true production function is not linear or if the  true distribution of revenue shocks are not a part of the exponential family, or if  support assumed  of the model transition functions is different from the true transition function. Given these primitives, it is easy to verify that this is a regular SMDP in the sense of \cref{regsmdp}. Therefore, from Theorem \ref{mainresults}, a Berk-Nash equilibrium exists.

\exx

\begin{customex}{1 (\normalfont contd)}\normalfont 
Following Examples 3 and 4, it is easy to see that the SMDP satisfies \cref{regsmdp} and given the normality of the transition probability function, Assumptions 1-3 hold. The state space is not compact, and therefore, we need to check whether Theorem 2 or Theorem 3 applies. The payoff function is unbounded. However, it is state-bounded and therefore, satisfies \cref{assumptionpayoffubd}. Finally, Assumptions 7-8 hold as well as illustrated in Example \ref{example-Unit_Root}. Therefore, by Theorem 3, a Berk-Nash equilibrium exists. We next characterize the Berk-Nash equilibrium for this instance.\fn{The details are given in \cref{sec:onlineappendix}.}

In this case, the Berk-Nash equilibrium  is characterized by the optimal policy function, $k=x^{*}y=A_{z}(\beta^{m})y=\dfrac{0.5\delta\beta^{m}}{(1-\delta\beta^{m})z+0.5\delta\beta^{m}}y,$ where there exists a $\beta^{m}\in (0,\beta^{*}).$ Indeed, note that the true transition probability function $Q(s)$ has a unique stationary measure $\mu$.  So, the Berk-Nash equilibrium for this SMDP is $\mu\times \delta_{x^{*}}$, supported by the belief $\delta_{(\beta^{m})}$.
\end{customex}

\begin{customex}{2 (\normalfont contd)}\normalfont
Given the primitives, it is easy to verify that this is a regular SMDP in the sense of \cref{regsmdp}. Given the state space  $S=[0,1]\times[0,b],$ for this misspecified SMDP is compact, therefore, from Theorem 1, a Berk-Nash equilibrium exists. The full characterization of the Berk-Nash equilibrium for the case when the \textit{true} cost function is convex and quadratic is sketched in  \cref{detailexample}. 
\end{customex}

\section{Methodological Contribution}\label{secmethod}

While nonstandard analysis has been used in mathematical economics since the 1970s, this paper relies on a new nonstandard technique pioneered in \citet{duanmuthesis} to extend theorems from finite mathematical structures to infinite mathematical structures.\footnote{This paper is part of an ongoing program applying nonstandard analysis to resolve important problems in Markov processes (\citet{drw21}, \citet{anderson2018mixhit}, \citet{ads21avg} and \citet{ads21gibbs}), statistics (\citet{nsbayes} and \citet{nscredible}) and mathematical economics (\citet{adku1} and \citet{adku2}). }  Candidates for this technique have the following properties:  

\begin{enumerate}[{\normalfont (i)}, topsep=1pt]
 \setlength{\itemsep}{-2pt}
\item The theorem is known on a finite (or finite-dimensional) space, \textit{and}
\item The theorem \textit{statement} does not rely heavily on the space being finite, \textit{but}
\item The existing proof(s) \textit{do} rely heavily on the space being finite.
\end{enumerate}
For results with these properties, Duanmu's technique allows one to directly translate the \textit{statement} of the theorem without having to translate the \textit{details} of the proof.


\defn{Nonstandard} models satisfy three principles: \defn{extension}, which associates to every ordinary mathematical object a nonstandard counterpart called its extension; \defn{transfer}, which preserves the truth values of first-order logic statements between standard and nonstandard models; and \defn{saturation}, which gives us a powerful mechanism for proving the existence of nonstandard objects defined in terms of finitely satisfiable collections of first-order formulas. In a suitably saturated nonstandard model, one can construct a hyperfinite probability space, which satisfies all the first order logical properties of a finite probability space, but which can be simultaneously viewed as a measure-theoretical probability space via the Loeb measure construction. In particular, Duanmu's technique invokes the following proof strategy:
\begin{enumerate}[{\normalfont (i)}, topsep=1pt]
 \setlength{\itemsep}{-2pt}
\item Start with a standard infinite (e.g. measure-theoretic) object.
\item Construct a \textit{lifting}, embedding our standard object in a hyperfinite object.
\item Use the \textit{transfer principle} to obtain the theorem for the hyperfinite object, essentially for free. 
\item Use the Loeb measure construction to \textit{push down} the theorem for the hyperfinite object to obtain the result in the original standard setting.
\end{enumerate}   

The \emph{truncation argument} uses a sequence of mathematical objects on larger and larger compact spaces to approximate an object on the $\sigma$-compact space. It is a standard methodology to derive results on $\sigma$-compact space using analogous results on compact space. Nonstandard analysis allows for the construction of a single nonstandard object such that:
\begin{enumerate}[{\normalfont (i)}, topsep=1pt]
 \setlength{\itemsep}{-2pt}
    \item The nonstandard object is an ``infinite" element of the sequence of mathematical objects on compact spaces;
    \item The nonstandard object sits on a $\NSE{}$compact set, which contains the original $\sigma$-compact space as a subset.
\end{enumerate}
To derive the desired result for the mathematical object on the $\sigma$-compact space, we apply the following proof strategy: 
\begin{enumerate}[{\normalfont (i)}, topsep=1pt]
 \setlength{\itemsep}{-2pt}
    \item Use the \emph{transfer principle} to obtain the theorem for the nonstandard object on the $\NSE{}$compact space, essentially for free;
    \item Use the Loeb measure construction to \emph{push down} the theorem for the nonstandard object to obtain result for the original mathematical object on the $\sigma$-compact space.
    To ensure the push-down is the desired standard object, we need to impose reasonable regularity conditions on the original standard mathematical object. These regularity conditions also guarantee the sequence of objects on compact sets converges in appropriate sense. 
\end{enumerate}

This paper is another example of these novel approaches. To prove Theorem \ref{mainresults}, we start with a regular SMDP  with compact state and action spaces, embed it in a hyperfinite SMDP, transfer existing results from EP to this hyperfinite SMDP, then conclude by a ``push down" argument to obtain a Berk-Nash equilibrium. On the other hand, to prove Theorem \ref{mainresultsigma1} and \ref{mainresultsigma2}, we start with regular SMDP with a $\sigma$-compact state space, a compact action space and a possibly unbounded payoff function, embed it in a nonstandard SMDP with a $\NSE{}$compact state space, transfer Theorem \ref{mainresults} to this nonstandard SMDP, then obtain a Berk-Nash equilibrium for the standard regular SMDP by a ``push down" argument.

There are several difficulties one faces in formulating the extension for infinite spaces. First, as the state and action spaces are infinite, it is not straightforward to show that the weighted Kullback-Leibler divergence is jointly lower semi-continuous, since the standard proof depends heavily on the finiteness of the state and action spaces.\fn{Claim A of Lemma 1 (Page 742, line 9) in EP.} Allowing for unbounded Radon-Nikodym derivatives only adds to the difficulty of this approach, thus making the direct generalization difficult. Second, the fixed point argument crucially relies on the upper hemicontinuity of policy correspondence in Equation (\ref{uhc}); which in the case of an unbounded state space is difficult to show.\fn{See  Claim B on (Page 774, line 25) in EP.} 
\begin{equation}
\label{uhc}
(s, Q) \mapsto M(s, Q) \equiv \arg \max _{\hat{x} \in \mathbb{X}} \int_{\mathbb{S}}\left\{\pi\left(s, \hat{x}, s^{\prime}\right)+\delta V\left(s^{\prime}\right)\right\} Q\left(d s^{\prime} \mid s, \hat{x}\right)
\end{equation}

\nt In Equation (\ref{uhc}), the solution $V$ to the Bellman equation varies as $(s, Q)$ varies, and therefore one can not simply apply Berge's maximal theorem as is usual in the standard textbook treatment (\cite{sl89}) where only $s$ varies. The subtlety arises from the fact that $V$ depends on \textit{both} $s$ and $Q.$ While  this additional dependence is easy to accomodate in the finite case but it is precisely here that the difficulty arises in the infinite one, especially with an unbounded state space. The crucial difficulty lies in establishing that the Bellman equation of the SMDP has a solution and that the solution varies continuously with respect to the appropriate transition probabilities (referencing measure). For bounded payoff functions, the Banach fixed point theorem guarantees the existence of a solution for the Bellman equation. But establishing the continuity of the solution of the Bellman equation with respect to the referencing measures is not straightforward. With unbounded payoff functions--which feature naturally in many important applications in economics--this issue becomes even more formidable. The Bellman equation need not have a solution when the payoff function is unbounded. The Banach fixed point theorem does not apply in this case, and one needs to impose a growth condition on the payoff function to ensure the existence of a solution for the Bellman equation. Moreover, it is much more difficult to show that the solution of the Bellman equation varies continuously in terms of the referencing measure in the Bellman equation, since the tail behavior of the Bellman equation is non-negligible.

 An alternative standard method to prove Theorem 1 is to use a sequence of finite SMDPs to approximate the regular SMDP $\mathcal{M}$, and construct a Berk-Nash equilibrium for $\mathcal{M}$ from a sequence of Berk-Nash equilibria of the sequence of finite SMDPs. However, as mentioned in Remark \ref{doubleapprox} in the paper, we need to partition the state, action and parameter spaces simultaneously to obtain the desired sequence of finite SMDPs, which makes the construction as well as the analysis of the sequence of finite SMDPs complicated. Nonstandard analysis allows for the construction of a hyperfinite SMDP $\mathcal{M}'$, where its state, action and parameter spaces are chosen to avoid all pathological aspects that may arise in the sequential approximation by finite SMDPs. We transfer EP's Theorem 1 to establish the existence of a hyperfinite Berk-Nash equilibrium for $\mathcal{M}'$, then construct a Berk-Nash equilibrium for the regular SMDP $\mathcal{M}$. Third, for Theorems 2 and 3, if we were try to establish them using standard methods, the obvious choice is via truncation and Theorem 1. That is, we would first construct a truncated sequence of SMDPs with larger and larger compact state spaces, and using Theorem 1 to guarantee the existence of a sequence of Berk-Nash equilibria, we would construct a Berk-Nash equilibrium for the original SMDP with a $\sigma$-compact state space. 
However, for all this we must perform a simultaneous ``double" approximation on the state and the parameter space, which makes the construction as well as the analysis of the truncated sequence of SMDPs extremely complicated. On the other hand, nonstandard analysis provides an elegant alternative approach by using a single nonstandard SMDP with a ``large" nonstandard compact state space to approximate the original SMDP. The nonstandard SMDP can be viewed informally as the limiting object of a sequence of truncated SMDPs, but avoids the aforementioned technical difficulties that arise in the standard approach. In Section 5.3, where we provide sketch of proofs for Theorems 2 and 3, we illustrate after Theorem 5.8, the parallelism between the nonstandard approach and the truncated sequence approach.

\section{Sketch of Proofs}
\label{sketchofproofs}
In this section, we provide sketch of proofs for main results of the paper:  Theorem \ref{mainresults}, Theorem \ref{mainresultsigma1} and Theorem \ref{mainresultsigma2}. After setting out some basic preliminaries of non-standard analysis for the lay reader in \cref{appdnotation}, we turn to an overview of the basic argumentation.
In \cref{sechypresent}, we sketch the proof for Theorem \ref{mainresults}. The detailed proof is presented in \cref{appendixA1}. In \cref{sechypresentsigma}, we sketch proofs for Theorem \ref{mainresultsigma1} and Theorem \ref{mainresultsigma2}. The detailed proofs are presented in \cref{appendixA2}.

\subsection{Preliminaries on Nonstandard Analysis}\label{appdnotation}
For those who are not familiar with nonstandard analysis, \citet{adku2, adku1} provide reviews tailored to economists. \citet{NDV, NSAA97, NAW} provide thorough introductions. 
We use $\NSE{}$ to denote the nonstandard extension map taking elements, sets, functions, relations, etc., to their nonstandard counterparts.
In particular, $\HReals$ and $\NSE{\Nats}$ denote the nonstandard extensions of the reals and natural numbers, respectively.
An element $r\in \HReals$ is \emph{infinite} if $|r|>n$ for every $n\in \Nats$ and is \emph{finite} otherwise. An element $r \in \HReals$ with $r > 0$ is \textit{infinitesimal} if $r^{-1}$ is infinite. For $r,s \in \HReals$, we use the notation $r \approx s$ as shorthand for the statement ``$|r-s|$ is infinitesimal,'' and use use $r \gtrapprox s$ as shorthand for the statement ``either $r \geq s$ or $r \approx s$.''

Given a topological space $(X,\topology)$,
the monad of a point $x\in X$ is the set $\bigcap_{ U\in \topology \, : \, x \in U}\NSE{U}$.
An element $x\in \NSE{X}$ is \emph{near-standard} if it is in the monad of some $y\in X$.
We say $y$ is the standard part of $x$ and write $y=\ST(x)$. 
Note that such $y$ is unique provided that $X$ is a Hausdorff space.
The \emph{near-standard part} $\NS{\NSE{X}}$ of $\NSE{X}$ is the collection of all near-standard elements of $\NSE{X}$.
The standard part map $\ST$ is a function from $\NS{\NSE{X}}$ to $X$, taking near-standard elements to their standard parts.
In both cases, the notation elides the underlying space $Y$ and the topology $\topology$,
because the space and topology will always be clear from context.
For a metric space $(X,d)$, two elements $x,y\in \NSE{X}$ are \emph{infinitely close} if $\NSE{d}(x,y)\approx 0$.
An element $x\in \NSE{X}$ is near-standard if and only if it is infinitely close to some $y\in X$.
An element $x\in \NSE{X}$ is finite if there exists $y\in X$ such that $\NSE{d}(x,y)<\infty$ and is infinite otherwise.

Let $X$ be a topological space endowed with Borel $\sigma$-algebra $\BorelSets X$ and
let $\FM{X}$ denote the collection of all finitely additive probability measures on $(X,\BorelSets X)$.
An internal probability measure $\mu$ on $(\NSE{X},\NSE{\BorelSets X})$ is an element of $\NSE{\FM{X}}$.
The Loeb space of the internal probability space $(\NSE{X},\NSE{\BorelSets X}, \mu)$ is a countably additive probability space $(\NSE{X},\Loeb{\NSE{\BorelSets X}}, \Loeb{\mu})$ such that
$
\Loeb{\NSE{\BorelSets X}}=\{A\subset \NSE{X}|(\forall \epsilon>0)(\exists A_i,A_o\in \NSE{\BorelSets X})(A_i\subset A\subset A_o\wedge \mu(A_o\setminus A_i)<\epsilon)\}
$ and
$
\Loeb{\mu}(A)=\sup\{\ST(\mu(A_i))|A_i\subset A,A_i\in \NSE{\BorelSets X}\}=\inf\{\ST(\mu(A_o))|A_o\supset A,A_o\in \NSE{\BorelSets X}\}.
$

Every standard model is closely connected to its nonstandard extension via the \emph{transfer principle}, which asserts that a first order statement is true in the standard model if and only if it is true in the nonstandard model.
Given a cardinal number $\kappa$, a nonstandard model is called $\kappa$-saturated if the following condition holds:
let $\cF$ be a family of internal sets, if $\cF$ has cardinality less than $\kappa$ and $\cF$ has the finite intersection property, then the total intersection of $\cF$ is non-empty. In this paper, we assume our nonstandard model is as saturated as we need (see \textit{e.g.} {\citet[][Thm.~1.7.3]{NSAA97}} for the existence of $\kappa$-saturated nonstandard models for any uncountable cardinal $\kappa$).

The concept of ``push-down,'' through which a standard object is constructed from a nonstandard object, is at the heart of nonstandard analysis and will be employed in the proofs of our theorems.

\begin{definition}\label{defpushdown}
Let $Y$ be a Hausdorff space endowed with Borel $\sigma$-algebra $\BorelSets Y$. 
Let $P$ be an internal probability measure on $(\NSE{Y}, \NSE{\BorelSets Y})$. 
The \textbf{push-down measure} of $P$ is defined to be a standard measure $\pd{P}$ on $(Y, \BorelSets Y)$ such that $\pd{P}(A)=\Loeb{P}(\ST^{-1}(A))$ for all $A\in \BorelSets Y$. 
\end{definition}

If the space of probability measures on $(Y, \BorelSets Y)$ is endowed with the Prokhorov metric, then an internal probability measure is in the monad of its push-down measure with respect to the Prokhorov metric, provided that the push-down measure is a probability measure. 
 
\begin{lemma}[{\citep[][Lemma.~6.1]{nsbayes}}]\label{compactpd}
Let $Y$ be a compact Hausdorff space endowed with Borel $\sigma$-algebra $\BorelSets Y$. 
Let $P$ be an internal probability measure on $(\NSE{Y}, \NSE{\BorelSets Y})$. 
Then $\pd{P}$ is a probability measure on $(Y, \BorelSets Y)$. 
\end{lemma}


\subsection{Sketch of the Proof for Theorem \ref{mainresults}}\label{sechypresent}
In this section, we consider a regular SMDP $\mathcal{M}=(\langle S,X,q_0,Q,\pi,\delta \rangle, \mathcal{Q}_{\Theta})$ with a compact state space. This is the environment for Examples \ref{costsproduction} and \ref{pricingshocks}. Moreover, existence of Berk-Nash equilibrium on $\mathcal{M}$ is an important intermediate step towards establishing existence of Berk-Nash equilibrium for SMDPs with a $\sigma$-compact state space. If we were to prove Theorem \ref{mainresults} using standard method, the two obvious choices are:
\begin{enumerate}[{\normalfont (i)}, topsep=1pt]
 \setlength{\itemsep}{-2pt}
    \item We may be able to generalize EP's proof to prove Theorem \ref{mainresults}. However, such generalization is far from trivial. For example, it is not straightforward to show that the weighted Kullback-Leibler divergence is jointly lower semi-continuous. Moreover, allowing for unbounded Radon-Nikodym derivatives adds difficulty to this approach;
    \item The other alternative is to use a sequence of finite SMDPs to approximate the regular SMDP $\mathcal{M}$, and construct a Berk-Nash equilibrium for $\mathcal{M}$ from a sequence of Berk-Nash equilibria of the sequence of finite SMDPs. However, for the same reason as mentioned in \cref{doubleapprox}, we need to partition the state, action and parameter spaces simultanesously to obtain the desired sequence of finite SMDPs, which makes the construction as well as the analysis of the sequence of finite SMDPs complicated. 
\end{enumerate}
Following the proof strategy outlined in \cref{secmethod}, we use nonstandard analysis to construct a hyperfinite SMDP $\mathcal{M}'$. The state, action and parameter spaces of $\mathcal{M}'$ are chosen via saturation to avoid all pathological aspects that may arise in the sequential approximation by finite SMDPs. We transfer EP's Theorem 1 to establish the existence of a hyperfinite Berk-Nash equilibrium for $\mathcal{M}'$, then construct a Berk-Nash equilibrium for the regular SMDP $\mathcal{M}$. We provide a sketch of proof in this section. We start with the following definition of a hyperfinite representation of compact metric spaces.\fn{
Roughly speaking, we construct a hyperfinite representation by first partitioning $\NSE{Y}$ into hyperfinitely many pieces of sets with infinitesimal radius, then picking one point from each element of the partition to form the hyperfinite representation.}  

\begin{definition}\label{hyperapproxsp}
Let $(Y,d)$ be a compact metric space with Borel $\sigma$-algebra $\BorelSets Y$.
A hyperfinite representation of $Y$ is a tuple $(T_Y,\{B_Y(t)\}_{t\in T_Y})$ such that

\begin{enumerate}[{\normalfont (i)}, topsep=1pt]
 \setlength{\itemsep}{-2pt}
\item $T_Y$ is a hyperfinite subset of $\NSE{Y}$ and $Y$ is a subset of $T_Y$;
\item $t\in B_Y(t)\in \NSE{\BorelSets Y}$ for every $t\in T_Y$;
\item For every $t\in T_Y$, the diameter of $B_Y(t)$ is infinitesimal;
\item For every $t\in T_Y$, $B_{Y}(t)$ contains an $\NSE{}$open set;
\item The hyperfinite collection $\{B_Y(t): t\in T_Y)\}$ forms a $\NSE{}$partition of $\NSE{Y}$. 
\end{enumerate}
For every $y\in \NSE{Y}$, we use $t_{y}$ to denote the unique element in $T_{Y}$ such that $y\in B_{Y}(t_y)$.
\end{definition}
 
The next result from {\citet[][Thm.~6.6]{drw21}} guarantees the existence of a hyperfinite representation when the underlying space is a compact metric space. 

\begin{lemma}\label{exhyper}
Let $Y$ be a compact metric space with Borel $\sigma$-algebra $\BorelSets Y$. 
Then there exists a hyperfinite representation $(T_{Y}, \{B_{Y}(t)\}_{t\in T_{Y}})$ of $Y$. 
\end{lemma}

A hyperfinite Markov decision process is a $\NSE{}$Markov decision process where the state and action spaces are hyperfinite.\fn{A hyperfinite Markov decision process can be viewed informally as an ``infinite" element in a sequence of finite Markov decision processes.  Moreover, every hyperfinite Markov decision process has the same first-order logic properties as a finite Markov decision process.} We construct a hyperfinite Markov decision process (HMDP) from the Markov decision process $\langle S,X,q_0,Q,\pi,\delta \rangle$:
\begin{enumerate}[{\normalfont (i)}, topsep=1pt]
 \setlength{\itemsep}{-2pt}
    \item Let $(T_{S}, \{B_{S}(s)\}_{s\in T_{S}})$ and $(T_{X}, \{B_{X}(x)\}_{x\in S_{X}})$ to be two hyperfinite representations of $S$ and $X$, respectively, as in \cref{exhyper}. $T_S$ is the hyperfinite state space and $T_X$ is the hyperfinite action space; 
    \item Define $h_0(\{s\})=\NSE{q}_0(B_{S}(s))$ for every $s\in T_S$. Note that $h_0$ is an internal probability measure on $T_S$. $h_0$ denotes the initial distribution of states;
    \item For every $s, s'\in T_S, x\in T_X$, let $\mathbb{Q}(s,x)(s')=\NSE{Q}(s,x)(B_{S}(s'))$ and $\mathbb{Q}(s, x)(A)=\sum_{s'\in A}\mathbb{Q}(s, x)(s')$ for all internal $A\subset T_S$. We write $\mathbb{Q}(A|s, x)$ for $\mathbb{Q}(s, x)(A)$. Then, $\mathbb{Q}: T_S\times T_X\to \NSE{\PM{T_S}}$ is an internal transition probability function;
    \item Define $\Pi: T_S\times T_X\times T_S\to \NSE{\Reals}$ to be the restriction of $\NSE{\pi}$ on $T_S\times T_X\times T_S$. $\Pi$ denotes the hyperfinite per-period payoff function;
    \item The discount factor $\delta$ remains the same as in \cref{defMDP}. 
\end{enumerate}

We now construct a \emph{hyperfinite subjective Markov decision process} (HSMDP) $\mathcal{M}'$ from the regular SMDP $\mathcal{M}$:
\begin{enumerate}[{\normalfont (i)}, topsep=1pt]
 \setlength{\itemsep}{-2pt}
    \item The hyperfinite parameter space is chosen to be $T_{\Theta}$, where $(T_{\Theta}, \{B_{\Theta}(\theta)\}_{\theta\in T_{\Theta}})$ be a hyperfinite representation of $\Theta$. By \cref{hyperapproxsp}, $B_{\Theta}(\theta)$ contains an $\NSE{}$open set for all $\theta\in T_{\Theta}$. Thus, we have $B_{\Theta}(\theta)\cap \NSE{\hat{\Theta}}\neq \emptyset$ for all $\theta\in T_{\Theta}$. So, without loss of generality, we can assume $T_{\Theta}\subset \NSE{\hat{\Theta}}$;
    \item For every $\theta\in T_{\Theta}$, every $s, s'\in T_S$ and every $x\in T_X$, define $\mathbb{Q}_{\theta}(s,x)(s')=\NSE{Q}_{\theta}(s,x)(B_{S}(s'))$ and let $\mathbb{Q}_{\theta}(s, x)(A)=\sum_{s'\in A}\mathbb{Q}_{\theta}(s, x)(s')$ for all internal $A\subset T_S$. We sometimes write $\mathbb{Q}_{\theta}(A|s, x)$ for $\mathbb{Q}_{\theta}(s, x)(A)$. The family $\mathscr{Q}_{T_{\Theta}}=\{\mathbb{Q}_{\theta}: \theta\in T_{\Theta}\}$ is the family of internal transition probability functions.
\end{enumerate}
The HSMDP $\mathcal{M}'$ is chosen to be $(\langle T_S, T_X, h_0, \mathbb{Q}, \Pi, \delta \rangle, \mathscr{Q}_{T_{\Theta}})$. 
The agent's problem can be cast recursively as
$
\mathbb{V}(t)=\max_{x\in T_X}\sum_{s'\in T_S}\{\Pi(s,x,s')+\delta \mathbb{V}(s')\}\mathbb{Q}(s'|s, x)
$
where $\mathbb{V}: T_S\to \Reals$ is the unique solution to the hyperfinite Bellman equation. 
\begin{definition}\label{defhyoptimal}
An action $x$ is $\NSE{}$optimal given $s$ in the HMDP($\mathbb{Q}$) if 

$$x\in \argmax_{\hat{x}\in T_X}\sum_{s'\in T_S}\{\Pi(s,\hat{x},s')+\delta \mathbb{V}(s')\}\mathbb{Q}(s'|s, \hat{x})
$$
\end{definition}

The definition of hyperfinite weighted Kullback-Leibler divergence is simply the transfer of the definition of weighted Kullback-Leibler divergence for finite SMDPs. 

\begin{definition}\label{defhyKL}
The hyperfinite weighted Kullback-Leibler divergence is a mapping $\mathbb{K}_{\mathbb{Q}}: \NSE{\PM{T_S\times T_X}}\times T_{\Theta}\to \NSE{\NNReals}$ 
such that for any $m\in \NSE{\PM{T_S\times T_X}}$ and $\theta\in T_{\Theta}$:

$$
\mathbb{K}_{\mathbb{Q}}(m, \theta)=\sum_{(s, x)\in T_S\times T_X}\mathbb{E}_{\mathbb{Q}(\cdot|s, x)}\left[\ln \big(\frac{\mathbb{Q}(s'|s,x)}{\mathbb{Q}_{\theta}(s'|s,x)}\big)\right]m(\{(s, x)\}).
$$
 The set of closest parameter values given $m\in \NSE{\PM{T_S\times T_X}}$ is
$
T_{\Theta}^{\mathbb{Q}}(m)=\argmin_{\theta\in T_{\Theta}}\mathbb{K}_{\mathbb{Q}}(m, \theta).
$
\end{definition}
For all $\theta\in T_{\Theta}$, $(s,x,s')\in T_S\times T_X\times T_S$, 
$\mathbb{Q}_{\theta}(s'|s,x)=0$ implies that $\mathbb{Q}(s'|s,x)=0$. 
So the hyperfinite relative entropy, $\mathbb{E}_{\mathbb{Q}(\cdot|s, x)}\left[\ln \big(\frac{\mathbb{Q}(s'|s,x)}{\mathbb{Q}_{\theta}(s'|s,x)}\big)\right]$, is well-defined, since the hyperfinite relative entropy is interpreted as $0$ if $\mathbb{Q}(s'|s,x)=0$. 
Note that the hyperfinite relative entropy is always non-negative. 
By transferring the finite existence result in \cref{EPfinitelemma}, we have the following theorem. The detailed proof of which is presented in \cref{appendixA1}. 

\begin{theorem}\label{hyperberkNash}
The hyperfinite Markov decision process 
$
\mathcal{M}'=(\langle T_S, T_X, h_0, \mathbb{Q}, \Pi, \delta \rangle, \mathscr{Q}_{T_{\Theta}})
$
has a hyperfinite Berk-Nash equilibrium. 
That is, there exists some $m\in\NSE{\PM{T_S\times T_X}}$ and some hyperfinite belief $\nu\in\NSE{\PM{T_{\Theta}}}$ such that 
\begin{enumerate}[{\normalfont (i)}, topsep=1pt]
 \setlength{\itemsep}{-2pt}
    \item \textbf{Optimality}: For all $(s,x)\in T_S\times T_X$ such that $m(\{(s, x)\})>0$, $x$ is $\NSE{}$optimal given $s$ in the $\text{HMDP}(\bar{\mathbb{Q}}_{\nu})$, where 
    $\bar{\mathbb{Q}}_{\nu}=\sum_{\theta\in T_{\Theta}}\mathbb{Q}_{\theta}\nu(\{\theta\})$;
    
    \item \textbf{Belief Restriction}: We have $\nu\in \NSE{\PM{T_{\Theta}^{\mathbb{Q}}(m)}}$;
    
    \item \textbf{Stationarity}: $m_{T_S}(\{s'\})=\sum_{(s,x)\in T_S\times T_X}\mathbb{Q}(s'|s, x)m(\{(s, x)\})$ for all $s'\in T_S$. 
\end{enumerate}
\end{theorem}

Following the proof strategy outlined in \cref{secmethod}, we can establish the existence of a Berk-Nash equilibrium in the regular SMDP $\mathcal{M}$:

\begin{proof}[\textbf{Proof of Theorem \ref{mainresults}}]
By \cref{hyperberkNash}, let $m$ be the hyperfinite Berk-Nash equilibrium for the hyperfinite SMDP with the associated hyperfinite belief $\nu$. By \cref{compactpd}, $\pd{m}$ and $\pd{\nu}$ are probability measures on $S\times X$ and $\Theta$, respectively. The stationarity of $\pd{m}$ is established in \cref{stationm}, the belief restriction of $\pd{\nu}$ is established in \cref{beliefmain} and the optimality is established in \cref{optmain}. Hence, $\pd{m}$ is a Berk-Nash equilibrium for $\mathcal{M}$ with the associated belief $\pd{\nu}$\fn{\label{sberknote} We may define a \emph{Berk-Nash S-equilibrium} to be an ``almost" hyperfinite Berk-Nash equilibrium. That is,  replacing equal signs by $\approx$ in the definition of a hyperfinite Berk-Nash equilibrium. By essentially the same proof, we can show that the push-down of a Berk-Nash S-equilibrium is a Berk-Nash equilibrium.}.
\end{proof}

If we were to prove Theorem \ref{mainresults} using standard method, we would consider a sequence of carefully chosen finite SMDPs. As both the state and action spaces are compact, the sequence of Berk-Nash equilibria for the sequence of finite SMDPs has a convergent sub-sequence. The limit of this sub-sequence would be a Berk-Nash equilibrium for $\mathcal{M}$.

\subsection{Sketch of the Proofs for Theorem \ref{mainresultsigma1} and Theorem \ref{mainresultsigma2}}\label{sechypresentsigma}

In this section, we consider a regular SMDP $\mathcal{M}=(\langle S,X,q_0,Q,\pi,\delta \rangle, \mathcal{Q}_{\Theta})$ with a $\sigma$-compact state space, and a possibly unbounded payoff function. As discussed in previous sections, this is the environment of many important examples in various fields of economics. If we were to establish the existence of a Berk-Nash equilibrium for $\mathcal{M}$ using standard method, the obvious choice is via truncation, that is, to construct a sequence $\{\mathcal{M}_{n}\}_{n\in \Nats}$ of SMDPs with compact state spaces $\{S_n\}_{n\in \Nats}$ such that $\bigcup_{n\in \Nats}S_n=S$. By Theorem \ref{mainresults}, there exists a sequence of Berk-Nash equilibria $\{m_{n}\}_{n\in \Nats}$ with associated belief $\{\nu_{n}\}_{n\in \Nats}$ for the sequence $\{\mathcal{M}_{n}\}_{n\in \Nats}$ of SMDPs. We then hope to construct a Berk-Nash equilibrium for $\mathcal{M}$ from the sequence $\{m_n\}_{n\in \Nats}$ under moderate regularity conditions. However, as pointed out in \cref{doubleapprox}, we must perform a simultaneous ``double" approximation on the state and the parameter space of $\mathcal{M}$, which makes the construction as well as the analysis of the sequence $\{\mathcal{M}_{n}\}_{n\in \Nats}$ extremely complicated. On the other hand, as discussed in \cref{secmethod}, nonstandard analysis provides an elegant alternative approach by using a single nonstandard SMDP with a ``large" $\NSE{}$compact state space to approximate $\mathcal{M}$. The nonstandard SMDP can be viewed informally as the limiting object of a sequence $\{\mathcal{M}_{n}\}_{n\in \Nats}$ of truncated SMDPs, but avoids many technical difficulties that arise in the standard approach. 

To construct the desired nonstandard SMDP, we first extend the sequence $\{S_n\}_{n\in \Nats}$ in \cref{assumptionstate} to an internal sequence $\{\NSE{S}_{n}\}_{n\in \NSE{\Nats}}$. 
By the transfer principle, $\NSE{S}_n$ is a $\NSE{}$compact set for all $n\in \NSE{\Nats}$. 
Pick some $N\in \NSE{\Nats}\setminus \Nats$. 
As $\{S_n\}_{n\in \Nats}$ is a sequence of non-decreasing sets, we have $\NSE{S}_n\subset \NSE{S_N}$ for all $n\in \Nats$, which implies that $\NS{\NSE{S}}\subset \NSE{S}_N$. 
As a result, the state space $S$ is a subset of $\NSE{S}_N$. 
The nonstandard subjective Markov decision process (NSMDP) 
$\mathcal{M}_{T_{\Theta}}^{N}=(\langle \NSE{S}_N,\NSE{X},\NSE{q}_0^{N},\NSE{Q}^{N},\NSE{\pi}_{N},\delta \rangle$, $\NSE{\mathcal{Q}^{N}}_{T_{\Theta}})$ is defined as: 

\begin{enumerate}[{\normalfont (i)}, topsep=1pt]
 \setlength{\itemsep}{-2pt}
    \item The state space is $\NSE{S}_N$, endowed with $\NSE{}$Borel $\sigma$-algebra $\NSE{\BorelSets{\NSE{S}_N}}$;
    \item The action space is $\NSE{X}$, endowed with $\NSE{}$Borel $\sigma$-algebra $\NSE{\BorelSets{\NSE{X}}}$;
    \item The parameter space $T_{\Theta}$ is the hyperfinite representation of $\Theta$ chosen in \cref{sechypresent}. Note that $T_{\Theta}\subset \NSE{\hat{\Theta}}$;
    \item $\NSE{q}_0^{N}(A)=\frac{\NSE{q}_0(A)}{\NSE{q}_0(\NSE{S}_N)}$ for all $A\in \NSE{\BorelSets {\NSE{S}_N}}$;
    \item $\NSE{Q}^{N}: \NSE{S}_N\times \NSE{X}\to \NSE{\PM{\NSE{S_N}}}$ is the $\NSE{}$transition probability function defined as 
    $\NSE{Q}^{N}(s, x)(A)=\frac{\NSE{Q}(s, x)(A)}{\NSE{Q}(s, x)(\NSE{S_N})}$ for all $A\in \NSE{\BorelSets{\NSE{S_N}}}$;
    \item The payoff function $\NSE{\pi}_{n}: \NSE{S_N}\times \NSE{X}\times \NSE{S_N}\to \NSE{\Reals}$ is the restriction of $\NSE{\pi}$ to $\NSE{S_N}\times \NSE{X}\times \NSE{S_N}$;
    \item the discounting factor $\delta$ remains the same;
    \item For every $\theta\in T_{\Theta}$, $\NSE{Q}_{\theta}^{N}: \NSE{S}_N\times \NSE{X}\to \NSE{\PM{\NSE{S_N}}}$ is the $\NSE{}$transition probability function defined as 
    $\NSE{Q}_{\theta}^{N}(s, x)(A)=\frac{\NSE{Q}_{\theta}(s, x)(A)}{\NSE{Q}_{\theta}(s, x)(\NSE{S_N})}$ for all $A\in \NSE{\BorelSets{\NSE{S_N}}}$. 
    Let $\NSE{\mathcal{Q}^{N}_{T_{\Theta}}}=\{\NSE{Q}_{\theta}^{N}: \theta\in T_{\Theta}\}$. 
\end{enumerate}

In \cref{trucequi}, we establish that, under \cref{assumptionstate}, every truncation of $\mathcal{M}$ is a regular SMDP and has a Berk-Nash equilibrium. By the transfer principle:
\begin{theorem}\label{nsmdpeqexst}
Suppose \cref{assumptionstate} holds.
Then $\mathcal{M}_{T_{\Theta}}^{N}$ is $\NSE{}$regular and has a Berk-Nash $\NSE{}$equilibrium. 
\end{theorem}

Our aim is to construct a Berk-Nash equilibrium for the standard regular SMDP $\mathcal{M}$ from the Berk-Nash $\NSE{}$equilibrium for the nonstandard SMDP $\mathcal{M}_{T_{\Theta}}^{N}$. Such construction depends crucially on assumptions presented in \cref{secmainresult}: 
\begin{enumerate}[{\normalfont (i)}, topsep=1pt]
 \setlength{\itemsep}{-2pt}
    \item \cref{assumptiontight} ensures the push-down of the Berk-Nash $\NSE{}$equilibrium of the nonstandard SMDP $\mathcal{M}_{T_{\Theta}}^{N}$ is a probability measure. If we were to tackle the problem using the standard truncation argument, 
    \cref{assumptiontight} would ensure the sequence of Berk-Nash equilibria for truncated SMDPs is tight, hence has a convergent subsequence;
    \item As the parameter space $\Theta$ is compact, the push-down of the $\NSE{}$belief that associated with the Berk-Nash $\NSE{}$equilibrium of $\mathcal{M}_{T_{\Theta}}^{N}$ is a probability measure on $\Theta$. To establish belief restriction for $\mathcal{M}$, \cref{assumptionrebound} and \cref{assumptionuniquemin} provide two alternative approaches:
    \begin{itemize}
        \item \cref{assumptionrebound} allows for the approximation of the weighted Kullback-Leibler divergence of $\mathcal{M}$ by the nonstandard weighted Kullback-Leibler divergence of $\mathcal{M}_{T_{\Theta}}^{N}$, which in turn guarantees the set of closest parameter for $\mathcal{M}$ is close to the set of closest paramter for $\mathcal{M}_{T_{\Theta}}^{N}$. If we were to tackle the problem using the standard truncation argument,  \cref{assumptionrebound} allows for the approximation of the weighted Kullback-Leibler divergence of $\mathcal{M}$ by the weighted Kullback-Leiber divergence of truncated SMDPs;
        \item \cref{assumptionuniquemin} guarantees that the set of closest parameter for $\mathcal{M}$, $\mathcal{M}_{T_{\Theta}}^{N}$ and all truncated SMDPs is the same singleton set, which implies belief restriction;
    \end{itemize}
    \item We establish optimality of the candidate Berk-Nash equilibrium of $\mathcal{M}$ for bounded and unbounded payoff functions under different sets of assumptions: 
    \begin{itemize}
        \item If the payoff function is bounded, then the Banach fixed point theorem guarantees the existence of a solution for the Bellman equation. We approximate the solution of the Bellman equation for $\mathcal{M}$ by the solution of the nonstandard Bellman equation for $\mathcal{M}_{T_{\Theta}}^{N}$, which further implies optimality of the candidate Berk-Nash equilibrium of $\mathcal{M}$. If we were to tackle the problem using standard truncation argument, we show that the sequence of solutions of the Bellman equations for the truncated SMDPs converges pointwise to an integrable function, which is the solution of the Bellman equation for $\mathcal{M}$;
        \item If the payoff function is unbounded, the Bellman equation need not have a solution. We impose \cref{assumptionpayoffubd}, \cref{assumptionsint} and \cref{assumptioncontwass} to guarantee the existence of a solution of the Bellman equation for $\mathcal{M}$. By similar but more complicated arguments as in the bounded payoff function case, these assumptions allow us to establish optimality of the candidate Berk-Nash equilibrium of $\mathcal{M}$.
    \end{itemize}
\end{enumerate}

We now sketch proofs for Theorem \ref{mainresultsigma1} and \ref{mainresultsigma2}. The detailed proof for these two theorems are postponed to \cref{appendixA2}. 

\begin{proof}[\textbf{The Proofs of Theorem \ref{mainresultsigma1} and  Theorem \ref{mainresultsigma2}}]
By \cref{nsmdpeqexst}, $\mathcal{M}_{T_{\Theta}}^{N}$ has a Berk-Nash $\NSE{}$equilibrium $m$ with the associated $\NSE{}$belief function $\nu\in \NSE{\PM{T_{\Theta}}}$. 
Then, we have
$$
m_{\NSE{S}}(A)=m_{\NSE{S}_{N}}(A)=\int_{\NSE{S}_N\times \NSE{X}}\NSE{Q}^{N}(A|s, x)m(\dee s, \dee x)
$$
for all $A\in \NSE{\BorelSets {\NSE{S_N}}}$. 
Thus, $m$ is an element of $\NSE{\mathcal{R}}$, where $\mathcal{R}$ is the set in \cref{assumptiontight}.
Under \cref{assumptiontight}, $\Loeb{m_{\NSE{S}}}(\ST^{-1}(S))=1$, hence the push down $\pd{m}$ is a probability measure on $S\times X$. As $\Theta$ is compact, by \cref{compactpd}, $\pd{\nu}$ is a probability measure on $\Theta$. To prove Theorems \ref{mainresultsigma1} and \ref{mainresultsigma2}, it is sufficient to show that $\pd{m}$ is a Berk-Nash equilibrium for the regular SMDP $\mathcal{M}$ with the belief $\pd{\nu}$ on $\Theta$. The stationarity of $\pd{m}$ follows from \cref{thmstationsigma}. \cref{beliefmainsigma} and \cref{beliefunique} establish belief restriction for $\pd{\nu}$ under uniform integrability (\cref{assumptionrebound}) and uniqueness (\cref{assumptionuniquemin}), respectively. Note that a correctly specified SMDP satisfies \cref{assumptionuniquemin}.  
Finally, for a bounded payoff function, optimality follows from \cref{optimalsigma}, proving Theorem \ref{mainresultsigma1}.
For an unbounded payoff function satisfying \cref{assumptionpayoffubd}, optimality follows from \cref{optunbdpayoff}, proving Theorem \ref{mainresultsigma2}.
\end{proof}






\section{Open Questions and Concluding Remarks}
\label{secdiscussion}
This paper uses a novel technique in nonstandard analysis to extend the existence results for Berk-Nash equilibrium from finite state and action spaces to sigma-compact state and compact action spaces, thereby allowing coverage of a wide range of natural examples in macroeconomics, microeconomics, and finance. This paper suggests the following promising directions for future work. First, like EP, we considers a single-agent environment. In future work, we hope to extend these results to the case for a continuum of agents, in particular, to the canonical static multi-agent game setting of \cite{ep16}, and to recursive equilibrium framework in macroeconomics  (\citet{mo19}). Second, as pointed out in \cref{sechypresent}, the standard analogue of our nonstandard approach towards Theorem \ref{mainresults} is to use a sequence of finite SMDPs to approximate the SMDP $\mathcal{M}$. The Berk-Nash equilibrium of $\mathcal{M}$ is the weak limit of the sequence of Berk-Nash equilibria for finite SMDPs. If we further understand the convergence rate of the sequence of Berk-Nash equilibria for finite SMDPs, we can approximate the Berk-Nash equilibrium for $\mathcal{M}$ by computing the Berk-Nash equilibrium for a sufficiently large but finite SMDP. This may also have implications for computational aspects of such equilibria for infinite spaces that are in practice can be approximated with sufficiently large but a finite setting. Third,  \cref{blfupdatemain} in the online appendix provides a possible learning foundation for SMDPs with compact state and action spaces.  Unfortunately, it relies on an implausibly strong condition, convergence in the total variation norm on measures. It is of great interest to develop a learning foundation under a weaker convergence condition such as convergence in the Prokhorov metric. This may have further implications for environments that are characterized by \textit{slow} learning as in \citet{fii20}.
Finally, another line of future research is to investigate the \textit{global stability} of the Berk-Nash equilibria with the tools developed in \cite{ks14} for Markov decision processes with unbounded state spaces; a setting for which our results in this paper have been developed.

\appendix

\section{Appendix}\label{appendixA}
We present proofs that are omitted from the main body of the paper.  
Most of the proofs make use of nonstandard analysis for which notations are introduced in \cref{appdnotation}. 

\subsection{Proof of Theorem 1}\label{appendixA1}
First, we provide a rigorous proof to Theorem \ref{mainresults}.  The following two lemmas are key to prove the existence of a hyperfinite Berk-Nash equilibrium in \cref{hyperberkNash}. The first lemma follows from the fact that $T_{\Theta}$ is hyperfinite. 

\begin{lemma}\label{Qcontinuity}
Suppose $(\langle S,X,q_0,Q,\pi,\delta \rangle$, $\mathcal{Q}_{\Theta})$ is a regular-SMDP. 
Then, for all $(s, x, s')\in T_S\times T_X\times T_S$, the function $\mathbb{Q}_{\theta}(s'|s, x)$ is $\NSE{}$continuous function of $\theta$.  
\end{lemma}

\begin{lemma}\label{Qpositive}
Suppose $(\langle S,X,q_0,Q,\pi,\delta \rangle$, $\mathcal{Q}_{\Theta})$ is a regular SMDP.
Then, for all $\theta\in T_{\Theta}$, $\mathbb{Q}_{\theta}(s'|s,x)>0$ for all $(s,s',x)\in T_S\times T_S\times T_X$ such that $\mathbb{Q}(s'|s,x)>0$.
\end{lemma}
\begin{proof}
Pick some $\theta\in T_{\Theta}$ and $(s, s', x)\in T_S\times T_S\times T_X$ with $\mathbb{Q}(s'|s, x)>0$. 
Note that $\mathbb{Q}(s'|s, x)=\NSE{Q}(s, x)(B_{S}(s'))$. 
As $T_{\Theta}\subset \NSE{\hat{\Theta}}$, by the transfer principle, we have $\NSE{Q}_{\theta}(s, x)(B_{S}(s'))>0$.
As $\mathbb{Q}_{\theta}(s'|s, x)=\NSE{Q}_{\theta}(s, x)(B_{S}(s'))$, we have the result. 
\end{proof}

\begin{proof}[\textbf{Proof of \cref{hyperberkNash}}]
Note that $T_{\Theta}$ is a hyperfinite set. 
Then the result follows from \cref{Qcontinuity}, \cref{Qpositive} and the transfer of \cref{EPfinitelemma}. 
\end{proof}

Next, we establish stationarity, optimality and belief restriction of the candidate Berk-Nash equilibrium $(\pd{m}, \pd{\nu})$ in the following three subsections, hence proving Theorem \ref{mainresults}.

\subsubsection{Stationarity}

Recall that $(\pd{m})_{S}$ denotes the marginal measure of $\pd{m}$ on $S$.
In this section, we establish the stationarity of $(\pd{m})_{S}$. 
We use $m_{T_S}$ to denote the marginal measure of $m$ on $T_S$. 

\begin{lemma}\label{changeTd}
For any $A\in \BorelSets S$, $\pd{(m_{T_S})}(A)=(\pd{m})_{S}(A)$.
\end{lemma}
\begin{proof}
We have $\pd{(m_{T_S})}(A)=\Loeb{m_{T_S}}(\ST^{-1}(A)\cap T_S)=\Loeb{m}\big((\ST^{-1}(A)\cap T_S)\times T_X\big)$ for every $A\in \BorelSets S$.
On the other hand, we have $(\pd{m})_{S}(A)=\pd{m}(A\times X)=\Loeb{m}\big((\ST^{-1}(A)\cap T_S)\times T_X\big)$ for all $A\in \BorelSets S$.
Hence, we have the desired result. 
\end{proof}

\begin{lemma}\label{stationLoeb}
Let $A$ be a (possibly external) subset of $T_S$. 
Suppose there exists a sequence $\{A_k: k\in \Nats\}$ of non-decreasing internal subsets of $T_S$ such that $\bigcup_{k\in \Nats} A_k=A$.
Then $\Loeb{m_{T_S}}(A)=\int_{T_S\times T_X}\Loeb{\mathbb{Q}(s, x)}(A\cap T_S)\Loeb{m}(\dee s, \dee x)$
\end{lemma}
\begin{proof}
By the continuity of probability, we have $\Loeb{m_{T_S}}(A)=\lim_{k\to\infty}\Loeb{m_{T_S}}(A_k)$. 
For each $k\in \Nats$, by the S-stationarity of $m$, we have
$
\Loeb{m_{T_S}}(A_k)\approx \int_{T_S\times T_X}\mathbb{Q}(s,x)(A_k)m(\dee s, \dee x)
\approx \int_{T_S\times T_X}\Loeb{\mathbb{Q}(s,x)}(A_k)\Loeb{m}(\dee s, \dee x).
$
Thus, we have $\Loeb{m_{T_S}}(A)=\lim_{k\to \infty}\int_{T_S\times T_X}\Loeb{\mathbb{Q}(s,x)}(A_k)\Loeb{m}(\dee s, \dee x)$.
The result then follows from the dominated convergence theorem. 
\end{proof}

To complete the proof, we need to make an assumption on the topological structure of $S$. We start with the following definition. 
\begin{definition}\label{pisystem}
A $\pi$-system on a set $\Omega$ is a non-empty collection $P$ of subsets of $\Omega$ that is closed under finite intersection. 
\end{definition}

\begin{lemma}[The Uniqueness Lemma]\label{uniqlemma}
Let $(\Omega, \Sigma)$ be a measure space with $\Sigma$ generated from some $\pi$-system $\Pi$. 
Let $\mu$ and $\nu$ be two probability measures that agree on $\Pi$. 
Then $\mu$ and $\nu$ agree on $\Sigma$. 
\end{lemma}

\begin{assumption}\label{assumptioneps}
There exists a $\pi$-system $\cF$ on $S$ that generates $\BorelSets S$ such that, for every $A\in \cF$, $\ST^{-1}(A)=\bigcup_{k\in \Nats}A_k$ for some non-decreasing sequence $\{A_k: k\in \Nats\}\subset \NSE{\BorelSets{\NSE{S}}}$ of sets.
\end{assumption}

Although \cref{assumptioneps} is stated in nonstandard terminology, it is satisfied by many standard topological spaces. 
In fact, all metric spaces which are endowed with the Borel $\sigma$-algebra satisfy \cref{assumptioneps}.
\begin{theorem}\label{metricsatisfy}
Let $Y$ be a metric space endowed with the Borel $\sigma$-algebra $\BorelSets Y$. 
Then $(Y, \BorelSets Y)$ satisfies \cref{assumptioneps}. 
\end{theorem}
\begin{proof}
Let $\cF$ be the $\pi$-system generated by the collection of open balls. 
Clearly, $\cF$ generates $\BorelSets Y$. 
Let $B(a,\eta)$ be an open ball centered at $a$ with radius $\eta$. 
For each $n\in \Nats$, let $C_n$ be the closure of $B(a, \eta-\frac{1}{n})$. 
Then, we have $\ST^{-1}\big(B(a,\eta)\big)=\bigcup_{n\in \Nats}\NSE{C_n}$. 
Pick some $U\in \cF$. 
Then $U=\bigcap_{i\leq n}U_i$ for some $n\in \Nats$, where $U_i$ is an open ball for all $i\leq n$. 
For each $i\leq n$, there is a sequence $\{A^{i}_{k}: k\in \Nats\}\subset \NSE{\BorelSets{\NSE{Y}}}$ such that $\ST^{-1}(U_i)=\bigcup_{k\in \Nats}A^{i}_{k}$. 
Then $U$ equals to the union of the countable collection $\{\bigcap_{i\leq n}A^{i}_{k_i}: k_1, k_2, \dotsc, k_n\in \Nats\}$. 
\end{proof}

\begin{lemma}\label{stationSTmap}
$(\pd{m})_{S}(A)=\int_{T_S\times T_X}\Loeb{\mathbb{Q}(s, x)}(\ST^{-1}(A)\cap T_S)\Loeb{m}(\dee s, \dee x)$ 
for all $A\in \BorelSets S$. 
\end{lemma}
\begin{proof}
By \cref{metricsatisfy}, let $\cF$ denote the $\pi$-system in \cref{assumptioneps}.
By \cref{changeTd}, we have $(\pd{m})_{S}(A)=\pd{(m_{T_S})}(A)=\Loeb{m_{T_S}}(\ST^{-1}(A)\cap T_S)$ for every $A\in \BorelSets S$. 
Pick some $B\in \cF$. 
By \cref{assumptioneps}, there is a sequence $\{B_k: k\in \Nats\}\subset \NSE{\BorelSets{\NSE{S}}}$ of non-decreasing sets such that $\ST^{-1}(B)=\bigcup_{k\in \Nats}B_k$. 
By \cref{stationLoeb}, we have
$
(\pd{m})_{S}(B)=\Loeb{m_{T_S}}(\ST^{-1}(B)\cap T_S)=\int_{T_S\times T_X}\Loeb{\mathbb{Q}(s, x)}(\ST^{-1}(B)\cap T_S)\Loeb{m}(\dee s, \dee x).
$
Define $P(A)=\int_{T_S\times T_X}\Loeb{\mathbb{Q}(s, x)}(\ST^{-1}(A) \cap T_S)\Loeb{m}(\dee s, \dee x)$ for every $A\in \BorelSets S$.  
It is easy to verify that $P$ is a well-defined a probability measure on $(S, \BorelSets S)$. 
As $(\pd{m})_{S}$ and $P$ agree on $\cF$, by \cref{uniqlemma}, we have the desired result. 
\end{proof}

Next, we quote the following  results from nonstandard analysis which will be used for the subsequent proofs.


\begin{theorem}[{\citet[][Prop.~8.4]{anderson82}}]\label{pdint}
Let $Y$ be a compact Hausdorff space endowed with Borel $\sigma$-algebra $\BorelSets Y$, 
let $\nu$ be an internal probability measure on $(\NSE{Y}, \NSE{\BorelSets Y})$, and let $f: Y\to \Reals$ be a bounded measurable function. 
Define $g: \NSE{Y}\to \Reals$ by $g(s)=f(\ST(s))$. Then we have $\int f(y)\pd{\nu}(\dee y)=\int g(y)\Loeb{\nu}(\dee y)$. 
\end{theorem}

\begin{theorem}[{\citet[][Corollary.~5]{nsweak}}]\label{nsweaklemma}
Let $Y$ be a compact Hausdorff space endowed with Borel $\sigma$-algebra $\BorelSets Y$, 
let $\{P_n\}_{n\in \Nats}$ be a sequence of probability measures on $(Y, \BorelSets Y)$. 
Then the sequence $\{P_n\}_{n\in \Nats}$ converges weakly to a probability measure $P$ on $(Y, \BorelSets Y)$ 
if and only if $P(A)=\Loeb{\NSE{P}_{N}}(\ST^{-1}(A))$ for all $A\in \BorelSets Y$ and $N\in \NSE{\Nats}\setminus \Nats$. 
\end{theorem}

Recall that we assume the mappings $(s, x)\to Q(s, x)$ and $(\theta, s, x)\to Q_{\theta}(s, x)$ are continuous in the Prokhorov metric. 
By \cref{nsweaklemma}, we have the following result:
\begin{lemma}\label{wkconverge}
For every $(s, x)\in T_S\times T_X$, every $\theta\in T_{\Theta}$ and every $A\in \BorelSets S$, 
we have $Q(\ST(s), \ST(x))(A)=\Loeb{\mathbb{Q}(s, x)}(\ST^{-1}(A)\cap T_S)$ and $Q_{\ST(\theta)}(\ST(s), \ST(x))(A)=\Loeb{\mathbb{Q}_{\theta}(s, x)}(\ST^{-1}(A)\cap T_S)$.
\end{lemma}
\begin{proof}
Pick $(s_0, x_0)\in T_S\times T_X$, $\theta_0\in T_{\Theta}$ and $A_0\in \BorelSets S$.
By \cref{nsweaklemma}, we have $Q(\ST(s_0), \ST(x_0))(A_0)=\Loeb{\NSE{Q}(s_0, x_0)}(\ST^{-1}(A_0))$
and $Q_{\ST(\theta_0)}(\ST(s_0), \ST(x_0))(A_0)=\Loeb{\NSE{Q}_{\theta_0}(s_0, x_0)}(\ST^{-1}(A_0))$.
As $\ST^{-1}(A_0)=\bigcup\{B_{S}(s): s\in \ST^{-1}(A_0)\cap T_S\}$, by construction, we obtain the desired result. 
\end{proof}

We now prove the main result of this section, which establishes stationarity of $\pd{m}$:  

\begin{theorem}\label{stationm}
$(\pd{m})_{S}(A)=\int_{S\times X}Q(A|s, x)\pd{m}(\dee s, \dee x)$ for every $A\in \BorelSets S$. 
\end{theorem}
\begin{proof}
By \cref{stationSTmap}, we have $(\pd{m})_{S}(A)=\int_{T_S\times T_X}\Loeb{\mathbb{Q}(s, x)}(\ST^{-1}(A)\cap T_S)\Loeb{m}(\dee s, \dee x)$ for all $A\in \BorelSets S$. 
Thus, it is sufficient to show that 
$
\int_{T_S\times T_X}\Loeb{\mathbb{Q}(s, x)}(\ST^{-1}(A)\cap T_S)\Loeb{m}(\dee s, \dee x)=\int_{S\times X}Q(s,x)(A)\pd{m}(\dee s, \dee x).
$
This follows from \cref{pdint} and \cref{wkconverge}. 
\end{proof}

\subsubsection{Belief Restriction}

Recall that $\nu$ is the hyperfinite belief as in \cref{hyperberkNash}. 
As $\Theta$ is compact, $\pd{\nu}$ is a well-defined probability measure on $\Theta$. 
In this section, we show that the support of $\pd{\nu}$ is a subset of $\Theta_{Q}(\pd{m})$. 
We start with the following result, which is closely related to \citet{zimradon}, on hyperfinite representation of density functions.

\begin{theorem}\label{dsyapprox}
For all $\theta\in T_{\Theta}$, all $(s,x,s')\in T_S\times T_X\times T_S$ such that 
\begin{enumerate}[{\normalfont (i)}, topsep=1pt]
 \setlength{\itemsep}{-2pt}
    \item $\mathbb{Q}_{\theta}(s'|s, x)>0$;
    \item $Q(\ST(s), \ST(x))$ is dominated by $Q_{\ST(\theta)}(\ST(s), \ST(x))$;
    \item $D_{\ST(\theta)}(\ST(s')|\ST(s), \ST(x))$ is finite. 
\end{enumerate}
Then, we have $\frac{\mathbb{Q}(s'|s, x)}{\mathbb{Q}_{\theta}(s'|s, x)}\approx D_{\ST(\theta)}(\ST(s')|\ST(s), \ST(x))$.
\end{theorem}
\begin{proof}
Pick some $\theta_0\in T_{\Theta}$, some $(s_0, x_0, s'_0)\in T_S\times T_X\times T_S$ that satisfy the assumptions of the theorem. 
As $T_{\Theta}\subset \NSE{\hat{\Theta}}$, by the transfer principle, we have 
$
\mathbb{Q}(s'_0|s_0, x_0)=\NSE{Q}(s_0, x_0)(B_{S}(s'_0))=\int_{B_{S}(s'_0)}\NSE{D}_{\theta_0}(y|s_0, x_0)\NSE{Q}_{\theta_0}(\dee y|s_0, x_0).
$
We also have $\mathbb{Q}(s'_0|s_0, x_0)=\int_{B_{S}(s'_0)}\frac{\mathbb{Q}(s'_0|s_0, x_0)}{\mathbb{Q}_{\theta_0}(s'_0|s_0, x_0)}\NSE{Q}_{\theta_0}(\dee y|s_0, x_0)$.
Note that the $D_{\ST(\theta_0)}(\ST(s'_0)|\ST(s_0), \ST(x_0))$ is finite and the density function $D_{\theta}(s'|s, x)$ is jointly continuous on $\{(\theta, s', s, x): Q(s, x) \text{ is  dominated by }   Q_{\theta}(s, x)\}$. Thus, we have $\NSE{D}_{\theta_0}(y|s_0, x_0)\approx \frac{\mathbb{Q}(s'_0|s_0, x_0)}{\mathbb{Q}_{\theta_0}(s'_0|s_0, x_0)}$ for all $y\in B_{S}(s'_0)$. Hence, we conclude that
$\frac{\mathbb{Q}(s'_0|s_0, x_0)}{\mathbb{Q}_{\theta_0}(s'_0|s_0, x_0)}\approx D_{\ST(\theta_0)}(\ST(s'_0)|\ST(s_0), \ST(x_0))$, completing the proof. 
\end{proof}

We now introduce the notion of S-integrability from nonstandard analysis. 

\begin{definition}\label{defsint}
Let $(\Omega,\cA,P)$ be an internal probability space 
and let $F: \Omega\to \NSE{\Reals}$ be an internally integrable function such that $\ST(F)$ exists $\Loeb{P}$-almost surely. 
Then $F$ is \emph{S-integrable} with respect to $P$ if $\ST({F})$ is $\Loeb{P}$-integrable, and $\int |F|\dee P\approx \int \ST({|F|})\dee \Loeb{P}$.
\end{definition}

We now show that the hyperfinite Kullback-Leibler divergence is infinitely close to the standard Kullback-Leibler divergence.
Recall that $\Theta_{m}=\{\theta\in \Theta: K_{Q}(m, \theta)<\infty\}$
for $m\in \PM{S\times X}$. 
Note that $\hat{\Theta}\subset \Theta_{m}$ for all $m\in \PM{S\times X}$. 

\begin{theorem}\label{KLapprox}
Let $\lambda$ be an element of $\NSE{\PM{T_S\times T_X}}$. 
Then, we have 
\begin{enumerate}[{\normalfont (i)}, topsep=1pt]
 \setlength{\itemsep}{-2pt}
    \item $\mathbb{K}_{\mathbb{Q}}(\lambda, \theta)\gtrapprox K_{Q}(\pd{\lambda}, \ST(\theta))$ for all $\theta\in T_{\Theta}$ such that $\ST(\theta)\in \Theta_{\pd{\lambda}}$;
    \item $\mathbb{K}_{\mathbb{Q}}(\lambda, \theta)\approx K_{Q}(\pd{\lambda}, \ST(\theta))$ for all $\theta\in T_{\Theta}$ such that $\ST(\theta)\in \hat{\Theta}$.
\end{enumerate}
\end{theorem}
\begin{proof}
Pick $\theta\in T_{\Theta}$ such that $\ST(\theta)\in \Theta_{\pd{\lambda}}$. 
As $K_{Q}(\pd{\lambda}, \ST(\theta))<\infty$, this implies that $Q(s, x)$ is dominated by $Q_{\ST(\theta)}(s, x)$ for $\pd{\lambda}$-almost all $(s, x)\in S\times X$. The proof of the theorem relies essentially on the following claim which is proved in the supplementary material, \ref{dominate}.
\begin{claim}\label{eptapprox}
For every $(s, x)\in T_S\times T_X$ such that $Q(\ST(s), \ST(x))$ is dominated by $Q_{\ST(\theta)}(\ST(s), \ST(x))$, 
$\mathbb{E}_{\mathbb{Q}(\cdot|s, x)}\left[\ln \big(\frac{\mathbb{Q}(s'|s,x)}{\mathbb{Q}_{\theta}(s'|s,x)}\big)\right]\approx 
\mathbb{E}_{Q(\cdot|\ST(s), \ST(x))}\left[\ln \big(D_{\ST(\theta)}(s'|\ST(s), \ST(x))\big)\right]$.
\end{claim}
\nt Define $g: S\times X\to \Reals$ to be $g(s, x)=\mathbb{E}_{Q(\cdot|s, x)}\left[\ln \big(D_{\ST(\theta)}(s'|s, x)\big)\right]$ for all $(s, x)\in S\times X$ such that $Q(s, x)$ is dominated by $Q_{\ST(\theta)}(s, x)$ and $g(s, x)=0$ otherwise. 
For each $n\in \Nats$, define $g_n: S\times X\to \Reals$ to be $g_{n}(s, x)=\min\{g(s, x), n\}$. 
As $K_{Q}(\pd{\lambda}, \ST(\theta))<\infty$, we conclude that $K_{Q}(\pd{\lambda}, \ST(\theta))=\lim_{n\to \infty}\int g_n(s, x)\pd{\lambda}(\dee s, \dee x)$. 
Note that each $g_n$ is a bounded measurable function. Similarly, we define  $G: T_S\times T_X\to \NSE{\Reals}$ to be $G(s, x)=\mathbb{E}_{\mathbb{Q}(\cdot|s, x)}\left[\ln \big(\frac{\mathbb{Q}(s'|s,x)}{\mathbb{Q}_{\theta}(s'|s,x)}\big)\right]$. 
For each $n\in \Nats$, let $G_n: T_S\times T_X\to \NSE{\Reals}$ be $G_n(s, x)=\min\{G(s, x), n\}$. 
By \cref{eptapprox} and \cref{pdint}, we have $\int_{T_S\times T_X}G_n(s, x)\lambda(\dee s, \dee x)\approx \int_{S\times X}g_n(s, x)\pd{\lambda}(\dee s, \dee x)$ for all $n\in \Nats$. 
Note that $\mathbb{K}_{\mathbb{Q}}(\lambda, \theta)\geq \lim_{n\to\infty}\int_{T_S\times T_X}G_n(s, x)\lambda(\dee s, \dee x)$. 
Thus, we have $\mathbb{K}_{\mathbb{Q}}(\lambda, \theta)\gtrapprox K_{Q}(\pd{\lambda}, \ST(\theta))$. 

For the special case that $\ST(\theta)\in \hat{\Theta}$,
by Arkeryd et al. (1997, Section 4, Corollary 6.1) and \cref{eptapprox}, $G$ is S-integrable with respect to $\lambda$. So
$\mathbb{K}_{\mathbb{Q}}(\lambda, \theta)\approx \lim_{n\to\infty}\int_{T_S\times T_X}G_n(s, x)\lambda(\dee s, \dee x)$ follows from Arkeryd et al. (1997, Section 4, Theorem 6.2).
Hence, $\mathbb{K}_{\mathbb{Q}}(\lambda, \theta)\approx K_{Q}(\pd{\lambda}, \ST(\theta))$ when $\ST(\theta)\in \hat{\Theta}$. 
\end{proof}
We now prove the main result of this section. 
\begin{theorem}\label{beliefmain}
The support of $\pd{\nu}$ is a subset of $\Theta_{Q}(\pd{m})$. 
\end{theorem}
\begin{proof}
Pick $\theta_0\in \Theta$ such that $\theta_0$ is in the support of $\pd{\nu}$. 
As $\pd{\nu}(A)=\Loeb{\nu}(\ST^{-1}(A))$ for all $A\in \BorelSets {\Theta}$,
by \cref{hyperberkNash}, there exists $\theta_1\approx \theta_0$ such that $\theta_1\in T_{\Theta}^{\mathbb{Q}}(m)$.
That is, we have $\mathbb{K}_{\mathbb{Q}}(m, \theta_1)= \min_{\theta\in T_{\Theta}}\mathbb{K}_{\mathbb{Q}}(m, \theta)$. 
Suppose there exists $\theta'\in \Theta$ such that $K_{Q}(\pd{m},\theta')<K_{Q}(\pd{m},\theta_0)-\frac{1}{n}$ for some $n\in \Nats$. 
Note that $K_{Q}(\pd{m}, \theta)$ is a continuous function of $\theta$ on $\Theta_{\pd{m}}$. 
As $\hat{\Theta}\subset \Theta_{\pd{m}}$ and $\hat{\Theta}$ is a dense subset of $\Theta$, 
there exists some $\hat{\theta}\in \hat{\Theta}$ such that $K_{Q}(\pd{m},\hat{\theta})<K_{Q}(\pd{m},\theta_0)-\frac{1}{2n}$.
Let $t_{\hat{\theta}}\in T_{\Theta}$ be the unique element such that $\hat{\theta}\in B_{\Theta}(t_{\hat{\theta}})$.
By \cref{KLapprox}, we have 
$
\mathbb{K}_{\mathbb{Q}}(m, t_{\hat{\theta}})\approx K_{Q}(\pd{m}, \hat{\theta})<K_{Q}(\pd{m},\theta_0)-\frac{1}{2n}\lessapprox \mathbb{K}_{\mathbb{Q}}(m, \theta_1)-\frac{1}{2n}.
$
This is a contradiction, so the support of $\pd{\nu}$ is a subset of $\Theta_{Q}(\pd{m})$.
\end{proof}

\subsubsection{Optimality}

In this section, we establish the optimality of the candidate Berk-Nash equilibrium $\pd{m}$. 

\begin{lemma}\label{Qbarlemma}
For every $\lambda\in \NSE{\PM{T_{\Theta}}}$ and every $(t, x)\in T_S\times T_X$,
$\pd{(\mathbb{Q}_{\lambda}(t, x))}=Q_{\pd{\lambda}}(\ST(t), \ST(x))$. 
That is, the push-down of $\mathbb{Q}_{\lambda}(t, x)$ is the same as $Q_{\pd{\lambda}}(\ST(t), \ST(x))$.
\end{lemma}
\begin{proof}
Fix $\lambda\in \NSE{\PM{T_{\Theta}}}$ and $(t, x)\in T_S\times T_X$.
Pick $A\in \BorelSets S$. 
By the construction of the Loeb measure, we have 
$
\pd{(\mathbb{Q}_{\lambda}(t, x))}(A)=\Loeb{\mathbb{Q}_{\lambda}(t,x)}(\ST^{-1}(A)\cap T_S)=\int_{S_{\Theta}}\Loeb{\mathbb{Q}_{i}(t,x)}(\ST^{-1}(A)\cap T_S)\Loeb{\lambda}(\dee i). 
$
By \cref{wkconverge}, we have $Q_{\ST(i)}(\ST(t), \ST(x))(A)=\Loeb{\mathbb{Q}_{i}(t, x)}(\ST^{-1}(A)\cap T_S)$ for all $i\in T_{\Theta}$.
Thus, by \cref{pdint}, we have
$
\int_{T_{\Theta}}\Loeb{\mathbb{Q}_{i}(t,x)}(\ST^{-1}(A)\cap T_S)\Loeb{\lambda}(\dee i)=\int_{\Theta}Q_{\theta}(\ST(t), \ST(x))(A)\pd{\lambda}(\dee \theta)=Q_{\pd{\lambda}}(\ST(t), \ST(x))(A).
$
Hence, we have the desired result. 
\end{proof}

Recall that $\nu\in \NSE{\PM{T_{\Theta}}}$ is the hyperfinite belief function that associates with the Berk-Nash S-equilibrium $m$.
We consider the Bellman equation
$$
V(s)=\max_{x\in X}\int_{S}\{\pi(s,x,s')+\delta V(s')\}\bar{Q}_{\pd{\nu}}(\dee s'|s, x).
$$

By the Banach fixed point theorem, there exists an unique $V\in \cC[T]$ that is a solution to this Bellman equation.  We fix $V$ for the rest of this section. 

\nt Similarly, we consider the hyperfinite Bellman equation
$$
\mathbb{V}(s)=\max_{x\in T_X}\int_{T_S}\{\Pi(s,x,s')+\delta \mathbb{V}(s')\}\bar{\mathbb{Q}}_{\nu}(\dee s'|s, x)
$$
where $\mathbb{V}: T_S\to \Reals$ is the unique solution to the hyperfinite Bellman equation. 
The existence of such $\mathbb{V}$ is guaranteed by the transfer principle. 
We fix $\mathbb{V}$ for the rest of this section. 
Define $\mathbb{V}': \NSE{S}\to \NSE{\Reals}$ by letting $\mathbb{V}'(s)=\mathbb{V}(t_s)$ for all $s\in \NSE{T}$, where $t_s$ is the unique element in $T_S$ such that $s\in B_{S}(t_s)$.

\begin{lemma}\label{BEclose}
For all $s\in \NSE{S}$, $\mathbb{V}'(s)\approx \NSE{V}(s)$. 
\end{lemma}

\begin{proof}
Let $V_0$ be the restriction of $\NSE{V}$ on $T_S$. 
For all $(s, x)\in T_S\times T_X$, by \cref{Qbarlemma} and \cref{pdint}, we have
$\int_{T_S}\{\Pi(s,x,s')+\delta V_0(s')\}\bar{\mathbb{Q}}_{\nu}(\dee s'|s, x)\approx \int_{S}\{\pi(\ST(s),\ST(x),s')+\delta V(s')\}\bar{Q}_{\pd{\nu}}(\dee s'|\ST(s), \ST(x)).$ Hence, we have, $\max_{x\in T_X}\int_{T_S}\{\Pi(s,x,s')+\delta V_0(s')\}\bar{\mathbb{Q}}_{\nu}(\dee s'|s, x)\approx \max_{x\in X}\int_{S}\{\pi(\ST(s),\ST(x),s')+\delta V(s')\}\bar{Q}_{\pd{\nu}}(\dee s'|\ST(s), \ST(x))=V(\ST(s))\approx V_0(s).$

Let $G(f)(s)=\max_{x\in T_X}\int_{T_S}\{\Pi(s,x,s')+\delta f(s')\}\bar{\mathbb{Q}}_{\nu}(\dee s'|s, x)$ for all internal function $f: T_S\to \NSE{\Reals}$.
Note that we have $\NSE{d_{\sup}}(G(f_1), G(f_2))\leq \delta\NSE{d}_{\sup}(f_1, f_2)$ 
for all internal functions $f_1, f_2: T_S\to \NSE{\Reals}$.
Moreover, we can find $\mathbb{V}$ as following: start with $V_0$ and define a sequence $\{V_n\}_{n\in \NSE{\Nats}}$ by $V_{n+1}=G(V_n)$. 
Then $\mathbb{V}$ is the $\NSE{}$limit of $\{V_n\}_{n\in \NSE{\Nats}}$. 
So:
$
\NSE{d_{\sup}}(V_0, \mathbb{V})\leq \frac{1}{1-\delta}\NSE{d_{\sup}}(V_1, V_0)\approx 0. 
$
As $V$ is continuous, we conclude that  $\mathbb{V}'(s)\approx \NSE{V}(s)$ for all $s\in \NSE{S}$. 
\end{proof}
\nt We now  prove the main result of this section. 
\begin{theorem}\label{optmain}
For every $(s, x)\in S\times X$ that is in the support of $\pd{m}$, $x$ is optimal given $s$ in the MDP($\bar{Q}_{\pd{\nu}}$).
\end{theorem}
\begin{proof}
Pick some $(s, x)\in S\times X$ in the support of $\pd{m}$. 
Then there exists some $(a, b)\in T_S\times T_X$ such that $(a, b)\approx (s, x)$ and $m\big(\{(a, b)\}\big)>0$. 
As $m$ is a hyperfinite Berk-Nash equilibrium, $b$ is optimal given $a$ in HMDP($\bar{\mathbb{Q}}_{\nu}$). That is, we have
$\int_{T_S}\{\Pi(a,b,s')+\delta \mathbb{V}(s')\}\bar{\mathbb{Q}}_{\nu}(\dee s'|a,b)= \max_{y\in T_X}\int_{T_S}\{\Pi(a,y,s')+\delta \mathbb{V}(s')\}\bar{\mathbb{Q}}_{\nu}(\dee s'|a,y). 
$
By \cref{BEclose}, \cref{Qbarlemma} and \cref{pdint}, we have
$\int_{T_S}\{\Pi(a,y,s')+\delta \mathbb{V}(s')\}\bar{\mathbb{Q}}_{\nu}(\dee s'|a,y)\approx \int_{S}\{\pi(s,\ST(y),s')+\delta V(s')\}\bar{Q}_{\pd{\nu}}(\dee s'|s,\ST(y))
$
for all $y\in T_X$. Thus, we have $x\in \argmax_{\hat{x}\in X}\int_{S}\{\pi(s,\hat{x},s')+\delta V(s')\}\bar{Q}_{\pd{\nu}}(\dee s'|s,\hat{x})$,
which implies that $x$ is optimal given $s$ in the MDP($\bar{Q}_{\pd{\nu}}$).
\end{proof}

By \cref{stationm}, \cref{beliefmain} and \cref{optmain}, $\pd{m}$ is a Berk-Nash equilibrium for $\mathcal{M}$, hence we have a complete proof of Theorem \ref{mainresults}.

\subsection{Proofs of Theorems 2 and 3}\label{appendixA2}
In this section, we provide rigorous proofs to Theorem \ref{mainresultsigma1} and \ref{mainresultsigma2}. We first show that every truncation of $\mathcal{M}$ has a Berk-Nash equilibrium, which immediately leads to a proof of \cref{nsmdpeqexst}.
For every $n\in \Nats$ and every finite $\Theta'\subset \hat{\Theta}$, we denote the truncation by $\mathcal{M}_{\Theta'}^{n}$. 
\begin{lemma}\label{wkctspreserve}
Suppose \cref{assumptionstate} holds.
Then, for every $n\in \Nats$,  the mappings $(s, x)\to Q^{n}(s, x)$ and $(\theta, s, x)\to Q^{n}_{\theta}(s, x)$ are continuous in Prokhorov metric. 
\end{lemma}
\begin{proof}
Let $(s_m, x_m)_{m\in \Nats}$ be a sequence of points in $S_n\times X$ that converges to some point $(s, x)\in S_n\times X$. Let $A$ be a continuity set of $Q^{n}(s, x)$. 
As $S_n$ is a continuity set of $Q^{n}(s, x)$, $A$ is a continuity set of $Q(s, x)$. 
Thus, we have 
$
\lim_{m\to \infty}Q^{n}(s_m, x_m)(A)=\lim_{m\to \infty}\frac{Q(s_m, x_m)(A)}{Q(s_m, x_m)(S_n)}=\frac{Q(s, x)(A)}{Q(s, x)(S_n)}=Q^{n}(s, x)(A).
$
The mapping $(s, x)\to Q^{n}(s, x)$ is continuous in Prokhorov metric. 
By the same argument, the mapping $(\theta, s, x)\to Q^{n}_{\theta}(s, x)$ is continuous in Prokhorov metric. 
\end{proof}

\begin{lemma}\label{restrictdominate}
Suppose \cref{assumptionstate} holds.
For every $n\in \Nats$, $Q^{n}(s, x)$ is dominated by $Q_{\theta}^{n}(s, x)$ for all $\theta\in \Theta'$ and all $(s, x)\in S_n\times X$. 
\end{lemma}
\begin{proof}
Pick $n\in \Nats$, $\theta\in \Theta'$ and $(s, x)\in S_n\times X$. 
Pick some $A\in \BorelSets {S_n}$ such that $Q_{\theta}^{n}(s, x)(A)=\frac{Q_{\theta}(s, x)(A)}{Q_{\theta}(s, x)(S_n)}=0$. 
This implies that $Q_{\theta}(s, x)(A)=0$. 
As $\theta\in \Theta'\subset \hat{\Theta}$, we have $Q(s, x)(A)=0$, which implies that $Q^{n}(s, x)(A)=0$
\end{proof}

For every $n\in \Nats$, $\theta\in \Theta'$ and every $(s, x)\in S_n\times X$,
we use $D_{\theta, n}(\cdot|s, x)$ to denote the density function of $Q^{n}(s, x)$ with respect to $Q_{\theta}^{n}(s, x)$. 

\begin{lemma}\label{dstctspreserve}
Suppose \cref{assumptionstate} holds.
For every $n\in \Nats$ and $\theta\in \Theta'$, $D_{\theta, n}(s'|s, x)$ is a jointly continuous function of $s'$, $s$ and $x$. 
\end{lemma}
\begin{proof}
Pick $n\in \Nats$ and $\theta\in \Theta'$. 
For any $A\in \BorelSets {S_n}$ and any $(s, x)\in S_n\times X$, we have 
$
Q^{n}(s, x)(A)=\int_{A}\frac{D_{\theta}(s'|s, x)}{Q(s, x)(S_n)}Q_{\theta}(s, x)(\dee s')
=\int_{A}\frac{D_{\theta}(s'|s, x)}{Q(s, x)(S_n)}Q_{\theta}(s, x)(S_n)Q_{\theta}^{n}(s, x)(\dee s').
$
So $D_{\theta, n}(s'|s, x)=\frac{D_{\theta}(s'|s, x)}{Q(s, x)(S_n)}Q_{\theta}(s, x)(S_n)$. 
Note that $Q(s, x)(S_n)>0$ and $Q_{\theta}(s, x)(S_n)>0$, and $S_n$ is a continuity set for both $Q(s, x)$ and $Q_{\theta}(s, x)$. 
Thus, $D_{\theta, n}(s'|s, x)$ is a jointly continuous function of $s'$, $s$ and $x$. 
\end{proof}

\nt Hence, by Theorem \ref{mainresults}, we have the following result.\fn{As $\Theta'$ is finite, $S_n$ and $X$ are compact, by \cref{dstctspreserve}, $D_{\theta, n}(s'|s, x)$ is bounded. 
Hence, \cref{KLint} of \cref{regsmdp} is automatically satisfied for the SMDP $\mathcal{M}_{\Theta'}^{n}$. 
Moreover, the payoff function $\pi_n$ is continuous on $S_n\times X\times S_n$. }

\begin{lemma}\label{trucequi}
Suppose \cref{assumptionstate} holds.
For every $n\in \Nats$ and every finite $\Theta'\subset \hat{\Theta}$, the SMDP $\mathcal{M}_{\Theta'}^{n}$ is regular and has a Berk-Nash equilibrium. 
\end{lemma}
\cref{nsmdpeqexst} then follows from the transfer of \cref{trucequi}.
Let $m\in \NSE{\PM{\NSE{S_N}\times \NSE{X}}}$ denote the Berk-Nash $\NSE{}$equilibrium of $\mathcal{M}_{T_{\Theta}}^{N}$, with the associated $\NSE{}$belief $\nu$. \cref{assumptiontight} guarantees that the push-down, $\pd{m}$, of $m$ is a probability measure on $\PM{S\times X}$. 
To show that $\pd{m}$ is a Berk-Nash equilibrium for the original SMDP $\mathcal{M}$ with the associated belief function $\pd{\nu}$, we break the proof into following subsections which will establish stationarity, optimality and belief restriction, respectively.

\subsubsection{Stationarity}

In this section, we show that $\pd{m}$ satisfies stationarity. 
Using essentially the same argument as in \cref{stationSTmap}, we have the following result. 

\begin{lemma}\label{stationSTmapQN}
For all $A\in \BorelSets S$, $(\pd{m})_{S}(A)=\int_{\NSE{S}\times \NSE{X}}\Loeb{\NSE{Q}^{N}(s, x)}(\ST^{-1}(A))\Loeb{m}(\dee s, \dee x)$, where $(\pd{m})_{S}$ denote the marginal measure of $\pd{m}$ on $S$. 
\end{lemma}

\begin{lemma}\label{wkconvergeQN}
For every $(s, x)\in \NS{\NSE{S}}\times \NSE{X}$, every $\theta\in T_{\Theta}$ and every $A\in \BorelSets S$, 
we have $Q(\ST(s), \ST(x))(A)=\Loeb{\NSE{Q}^{N}(s, x)}(\ST^{-1}(A))$ and 
$Q_{\ST(\theta)}(\ST(s), \ST(x))(A)=\Loeb{\NSE{Q}^{N}_{\theta}(s, x)}(\ST^{-1}(A))$.
\end{lemma}
\begin{proof}
Pick some $(s_0, x_0)\in \NS{\NSE{S}}\times \NSE{X}$, some $\theta_0\in T_{\Theta}$ and some $A_0\in \BorelSets S$.
By \cref{nsweaklemma}, we have $Q(\ST(s_0), \ST(x_0))(A_0)=\Loeb{\NSE{Q}(s_0, x_0)}(\ST^{-1}(A_0))$.
As $\NSE{Q}^{N}(s_0, x_0)(\NSE{S}_{N})\approx 1$, we have $Q(\ST(s_0), \ST(x_0))(A_0)=\Loeb{\NSE{Q}^{N}(s_0, x_0)}(\ST^{-1}(A_0))$.
By the same argument, we have $Q_{\ST(\theta_0)}(\ST(s_0), \ST(x_0))(A_0)=\Loeb{\NSE{Q}^{N}_{\theta_0}(s_0, x_0)}(\ST^{-1}(A_0))$.
\end{proof}

\begin{theorem}\label{thmstationsigma}
Suppose \cref{assumptionstate} holds. Then,  $(\pd{m})_{S}(A)=\int_{S\times X}Q(A|s,x)\pd{m}(\dee s, \dee x)$ for every $A\in \BorelSets S$, 
\end{theorem}
\begin{proof}
Pick $A\in \BorelSets S$.
By \cref{stationSTmapQN}, $(\pd{m})_{S}(A)=\int_{\NSE{S}\times \NSE{X}}\Loeb{\NSE{Q}^{N}(s, x)}(\ST^{-1}(A))\Loeb{m}(\dee s, \dee x)$
As $\Loeb{m}(\ST^{-1}(S)\times \NSE{X})=1$, we have
$$\int_{\NSE{S}\times \NSE{X}}\Loeb{\NSE{Q}^{N}(s,x)}(\ST^{-1}(A))\Loeb{m}(\dee s, \dee x)=\lim_{n\to \infty}\int_{\NSE{S_n}\times \NSE{X}}\Loeb{\NSE{Q}^{N}(s, x)}(\ST^{-1}(A))\Loeb{m}(\dee s, \dee x).$$

\nt By \cref{wkconvergeQN} and \cref{pdint}, we have 
$\int_{\NSE{S_n}\times \NSE{X}}\Loeb{\NSE{Q}^{N}(s, x)}(\ST^{-1}(A))\Loeb{m}(\dee s, \dee x)=\int_{S_n\times X}Q(s, x)(A)\pd{m}(\dee s, \dee x).
$
Note that we also have
$
\lim_{n\to\infty}\int_{S_n\times X}Q(s, x)(A)\pd{m}(\dee s, \dee x)=\int_{S\times X}Q(s, x)(A)\pd{m}(\dee s, \dee x).
$
So, we have the desired result. 
\end{proof}

\subsubsection{Belief Restriction under \cref{assumptionrebound}}

In this section, we establish belief restriction assuming uniformly bounded relative entropy. 
Recall that $\nu$ is the hyperfinite belief that associates with the Berk-Nash $\NSE{}$equilibrium $m$ of the nonstandard SMDP $\mathcal{M}_{T_{\Theta}}^{N}$. 
Recall that $T_{\Theta}\subset \NSE{\hat{\Theta}}$. 
By the transfer of \cref{restrictdominate}, $\NSE{Q}^{N}(s, x)$ is $\NSE{}$dominated by $\NSE{Q}_{\theta}^{N}(s, x)$ for all $\theta\in T_{\Theta}$ and $(s, x)\in \NSE{S}_{N}\times \NSE{X}$. We use $\mathbb{D}_{\theta}(\cdot|s, x)$ to denote the $\NSE{}$density function of $\NSE{Q}^{N}(s, x)$ with respect to $\NSE{Q}_{\theta}^{N}(s, x)$. By the transfer of \cref{dstctspreserve}, $\mathbb{D}_{\theta}(s'|s, x)$ is jointly $\NSE{}$continuous on $(s',s, x)$.  

\begin{lemma}\label{densityapprox}
Suppose \cref{assumptionstate} holds.
For all $\theta\in T_{\Theta}$ and all $(s, x)\in \NS{\NSE{S}}\times \NSE{X}$, we have
$\mathbb{D}_{\theta}(s'|s, x)\approx \NSE{D}_{\theta}(s'|s, x)$
on a $\NSE{Q}_{\theta}^{N}(s, x)$ measure $1$ set. 
\end{lemma}
\begin{proof}
Pick $\theta\in T_{\Theta}$ and $(s, x)\in \NS{\NSE{S}}\times \NSE{X}$.
Note that $\NSE{Q}(s, x)(\NSE{S_N})\approx 1$ and $\NSE{Q}_{\theta}(s, x)(\NSE{S_N})\approx 1$. 
For every $A\in \NSE{\BorelSets{\NSE{S_N}}}$, we have:
$
\NSE{Q}^{N}(s, x)(A)=\frac{\NSE{Q}(s, x)(A)}{\NSE{Q}(s, x)(\NSE{S_N})}
\approx \int_{\NSE{S}_{N}}\NSE{D}_{\theta}(s'|s, x)\NSE{Q}_{\theta}^{N}(s, x)(\dee s'). 
$
Note that $\NSE{Q}^{N}(s, x)(A)=\int_{\NSE{S_N}}\mathbb{D}_{\theta}(s'|s, x)\NSE{Q}_{\theta}^{N}(s, x)(\dee s')$. 
Thus, we conclude that $\mathbb{D}_{\theta}(s'|s, x)\approx \NSE{D}_{\theta}(s'|s, x)$ on some $\NSE{Q}_{\theta}^{N}(s, x)$ measure $1$ set. 
\end{proof}

For every $\theta\in T_{\Theta}$, let the nonstandard Kullback-Leibler divergence be:
$$
\NSE{K}_{N}(m, \theta)=\int_{\NSE{S_N}\times \NSE{X}}\NSE{\mathbb{E}}_{\NSE{Q}^{N}(\cdot|s, x)}\left[\ln \big(\mathbb{D}_{\theta}(s'|s,x)\big)\right]m(\dee s, \dee x).
$$
The set of closest parameter values given $m$ is the set $T_{\Theta}^{N}(m)=\argmin_{\theta\in T_{\Theta}}\NSE{K}_{N}(m, \theta)$.
Recall that we use $\Theta_{\pd{m}}$ to denote the set $\{\theta\in \Theta: K_{Q}(\theta, \pd{m})<\infty\}.$ The proof of the following Lemma \ref{exnsexeq} and Theorem \ref{stKLnsKL} are straightforward and therefore, provided in \cref{infinitesimaldominate}.

\begin{lemma}\label{exnsexeq}
Suppose \cref{assumptionstate} and \cref{assumptiontight} hold.
For every $(\theta,s,x)\in T_{\Theta}\times \NS{\NSE{S}}\times \NSE{X}$, if $\ST(\theta)\in \Theta_{\pd{m}}$ and 
$Q(\ST(s), \ST(x))$ is dominated by $Q_{\ST(\theta)}(\ST(s), \ST(x))$, then $\NSE{\mathbb{E}}_{\NSE{Q}^{N}(\cdot|s, x)}\left[\ln \big(\mathbb{D}_{\theta}(s'|s,x)\big)\right]\approx \mathbb{E}_{Q(\cdot|\ST(s), \ST(x))}\left[\ln \big(D_{\ST(\theta)}(s'|\ST(s),\ST(x))\big)\right]$.
\end{lemma}

By the transfer principle, $\NSE{Q}^{N}(s, x)(\NSE{S_N})>r$ for all $(s, x)\in \NSE{S_N}\times \NSE{X}$. 
As $T_{\Theta}\subset \NSE{\hat{\Theta}}$, following the calculation in \cref{dstctspreserve}, $|\mathbb{D}_{\theta}(s'|s,x)|\leq \frac{1}{r}|\NSE{D}_{\theta}(s'|s, x)|$ for all $(s',s,\theta, x)\in \NSE{S_N}\times \NSE{S_N}\times T_{\Theta}\times \NSE{X}$.
We now establish the connections between the nonstandard weighted Kullback-Leibler divergence and the standard weighted Kullback-Leibler divergence. 

\begin{theorem}\label{stKLnsKL}
Suppose \cref{assumptionstate}, \cref{assumptiontight} and \cref{assumptionrebound} hold. 
For every $\theta\in T_{\Theta}$, if $\ST(\theta)\in \Theta_{\pd{m}}$, then  $\NSE{K}_{N}(m, \theta)\approx K_{Q}(\pd{m}, \ST(\theta))$.
\end{theorem}

We now prove the main result of this section, which establishes belief restriction. 

\begin{theorem}\label{beliefmainsigma}
Suppose \cref{assumptionstate}, \cref{assumptiontight} and \cref{assumptionrebound} hold.
The support of $\pd{\nu}$ is a subset of $\Theta_{Q}(\pd{m})$. 
\end{theorem}
\begin{proof}
Pick $\theta_0\in \Theta$ such that $\theta_0$ is in the support of $\pd{\nu}$. 
As $\pd{\nu}(A)=\Loeb{\nu}(\ST^{-1}(A))$ for all $A\in \BorelSets \Theta$,
by \cref{hyperberkNash}, there exists $\theta_1\approx \theta_0$ such that $\theta_1\in T_{\Theta}^{N}(m)$.
That is, we have $\NSE{K}_{N}(m, \theta_1)=\argmin_{t\in T_{\Theta}}\NSE{K}_{N}(m, t)$. 
Suppose there exists $\theta'\in \Theta$ such that $K_{Q}(\pd{m},\theta')<K_{Q}(\pd{m},\theta_0)-\frac{1}{n}$ for some $n\in \Nats$. 
Clearly, both $\theta'$ and $\theta_0$ belong to $\Theta_{\pd{m}}$. 
Let $t_{\theta'}\in T_{\Theta}$ be the unique element such that $\theta'\in B_{\Theta}(t_{\theta'})$. 
By \cref{stKLnsKL}, we have 
$
\NSE{K}_{N}(m, t_{\theta'})\approx K_{Q}(\pd{m}, \theta')<K_{Q}(\pd{m},\theta_0)-\frac{1}{n}\lessapprox \NSE{K}_{N}(m, \theta_1)-\frac{1}{n}.
$
This is a contradiction, hence we conclude that $\pd{\nu}$ is a subset of $\Theta_{Q}(\pd{m})$.
\end{proof}

\subsubsection{Belief Restriction without \cref{assumptionrebound}}

In this section, we establish belief restriction of the SMDP $\mathcal{M}$ if $\mathcal{M}$ is either correctly specified or satisfies \cref{assumptionuniquemin}. 
We first assume that $\mathcal{M}$ is correctly specified.

\begin{theorem}\label{beliefcorrect}
Suppose the SMDP $\mathcal{M}$ is correctly specified, \cref{assumptionstate} and \cref{assumptiontight} hold. 
Then, the support of $\pd{\nu}$ is a subset of $\Theta_{Q}(\pd{m})$. 
\end{theorem}
\begin{proof}
As the SMDP $\mathcal{M}$ is correctly specified and $\Theta\subset T_{\Theta}$, we have $\min_{t\in T_{\Theta}}\NSE{K}_{N}(m,t)=0$.
Pick $\theta_0\in \Theta$ and $(s_0, x_0)$ such that $\theta_0$ in the support of $\pd{\nu}$ and $(s_0, x_0)$ in the support of $\pd{m}$. 
Then, by \cref{hyperberkNash}, there exist $\theta_1\approx \theta_0$ and $(s_1, x_1)\in \NSE{S}_N\times \NSE{X}$ such that $\theta_1\in T_{\Theta}^{N}(m)$ and $(s_1, x_1)$ in the $\NSE{}$support of $m$. 
By the transfer of Lemma 1 in EP, we have $\NSE{Q}_{\theta_1}^{N}(s_1, x_1)=\NSE{Q}^{N}(s_1, x_1)$.
As $s_1$ is near-standard, by \cref{regsmdp}, we conclude that $Q_{\theta_0}(s_0, x_0)=Q(S_0, x_0)$ and therefore, $K_{Q}(\pd{m}, \theta_0)=0$. 
\end{proof}
We now assume that \cref{assumptionuniquemin} holds but $\mathcal{M}$ may be misspecified. 

\begin{theorem}\label{beliefunique}
Suppose \cref{assumptionstate}, \cref{assumptiontight} and \cref{assumptionuniquemin} hold.
Then the support of $\pd{\nu}$ is a subset of $\Theta_{Q}(\pd{m})$.
\end{theorem}
\begin{proof}
As $\Theta\subset T_{\Theta}$, by \cref{assumptiontight}, $T_{\Theta}^{N}(m)=\{\theta_0\}$. 
Thus, we have $\pd{\nu}(\{\theta_0\})=\nu(\{\theta_0\})=1$. 
By \cref{assumptionuniquemin} again, $\Theta_{Q}(\pd{m})=\{\theta_0\}$, completing the proof. 
\end{proof}

\subsubsection{Optimality with Bounded Payoff Function}

In this section, we establish optimality of the candidate Berk-Nash equilibrium $\pd{m}$ assuming bounded and continuous payoff function. We start with the follwing lemma:

\begin{lemma}\label{Qbarlemmasigma}
For every $\lambda\in \NSE{\PM{T_{\Theta}}}$ and every $(s, x)\in \NS{\NSE{S}}\times \NSE{X}$,
$\pd{(\bar{\NSE{Q}}_{\lambda}^{N}(s, x))}=\bar{Q}_{\pd{\lambda}}(\ST(s), \ST(x))$. 
That is, the push-down of $\bar{\NSE{Q}}_{\lambda}^{N}(s, x)$ is the same as $\bar{Q}_{\pd{\lambda}}(\ST(s), \ST(x))$.
\end{lemma}
\begin{proof}
Fix $(\lambda,s,x)\in \NSE{\PM{T_{\Theta}}}\times \NS{\NSE{S}}\times \NSE{X}$ and $A\in \BorelSets S$. 
As $\Loeb{\bar{\NSE{Q}}_{\lambda}^{N}(s, x)}(\ST^{-1}(S))=1$, we have $\pd{(\bar{\NSE{Q}}_{\lambda}^{N}(s, x))}(A)=\Loeb{\bar{\NSE{Q}}_{\lambda}(s,x)}(\ST^{-1}(A))=\int_{T_{\Theta}}\Loeb{\NSE{Q}_{i}(s,x)}(\ST^{-1}(A))\Loeb{\lambda}(\dee i)$.
By \cref{wkconverge}, we have $Q_{\ST(i)}(\ST(s), \ST(x))(A)=\Loeb{\NSE{Q}_{i}(s, x)}(\ST^{-1}(A))$ for all $i\in T_{\Theta}$.
Thus, by \cref{pdint}, we have
$
\int_{T_{\Theta}}\Loeb{\NSE{Q}_{i}(s,x)}(\ST^{-1}(A))\Loeb{\lambda}(\dee i)=\int_{\Theta}Q_{\theta}(\ST(s), \ST(x))(A)\pd{\lambda}(\dee \theta)=\bar{Q}_{\pd{\lambda}}(\ST(s), \ST(x))(A).
$
Hence, we have the desired result. 
\end{proof}

\nt We now consider the Bellman equation,

$$\label{boundedbellman}
V(s)=\max_{x\in X}\int_{S}\{\pi(s,x,s')+\delta V(s')\}\bar{Q}_{\pd{\nu}}(\dee s'|s, x). 
$$

Let $\mathcal{C}_0[S]$ denote the set of bounded continuous functions on $S$ equipped with the sup-norm.
Then $\mathcal{C}_0[S]$ is a complete metric space. 
Under \cref{assumptionpayoff}, the map $F(g)(s)=\max_{x\in X}\int_{S}\{\pi(s,x,s')+\delta g(s')\}\bar{Q}_{\pd{\nu}}(\dee s'|s, x)$ is a contraction mapping from $\mathcal{C}_0[S]$ to $\mathcal{C}_0[S]$. By the Banach fixed point theorem, there is an unique $V\in \mathcal{C}_0[S]$ that is a solution to the Bellman equation. We fix $V$ for the rest of this section. 
The nonstandard Bellman equation is:
$$\label{boundednsbellman}
\mathbb{V}(s)=\max_{x\in \NSE{X}}\int_{\NSE{S_N}}\{\NSE{\pi}_N(s,x,s')+\delta \mathbb{V}(s')\}\bar{\NSE{Q}}_{\nu}^{N}(\dee s'|s, x),
$$

where $\mathbb{V}\in \NSE{\mathcal{C}_0}[\NSE{S_N}]$ is the unique solution of the nonstandard Bellman equation. 
The existence of such $\mathbb{V}$ is guaranteed by the transfer principle. 
We also fix $\mathbb{V}$ for the rest of this section. 

\begin{lemma}\label{keyapproxlemma}
Suppose \cref{assumptionstate} and \cref{assumptionpayoff} hold.
For every $(s, x)\in \NS{\NSE{S}}\times \NSE{X}$:
$\int_{\NSE{S_N}}\{\NSE{\pi}_N(s,x,s')+\delta \NSE{V}(s')\}\bar{\NSE{Q}}_{\nu}^{N}(\dee s'|s, x)\approx \int_{S}\{\pi(\ST(s), \ST(x), s')+\delta V(s')\}\bar{Q}_{\pd{\nu}}(\dee s'|\ST(s), \ST(x)).$
\end{lemma}
\begin{proof}
Pick $(t, x)\in \NS{\NSE{T}}\times \NSE{X}$. 
As $\pi$ is bounded, $\NSE{\pi}_{N}(t, x, \cdot)$ is bounded. 
By Arkeryd et al. (1997, Section 4, Corollary 6.1), $\NSE{\pi}_N(t,x,t')+\delta \NSE{V}(t')$ is S-integrable with respect to $\bar{\NSE{Q}}_{\nu}^{N}(\dee t'|t, x)$. 
By Arkeryd et al. (1997, Section 4, Corollary 6.1), \cref{Qbarlemma} and \cref{pdint} , we have 
\begin{align*}
\int_{\NSE{T_N}}\{\NSE{\pi}_N(t,x,t')+\delta \NSE{V}(t')\}\bar{\NSE{Q}}_{\nu}^{N}(\dee t'|t, x)
& \approx \lim_{n\to\infty}\ST\big(\int_{\NSE{T_n}}\{\NSE{\pi}_N(t,x,t')+\delta \NSE{V}(t')\}\bar{\NSE{Q}}_{\nu}^{N}(\dee t'|t, x)\big)\\
&=\int_{T}\{\pi(\ST(t), \ST(x), t')+\delta V(t')\}\bar{Q}_{\pd{\nu}}(\dee t'|\ST(t), \ST(x)).
\end{align*}
\nt Hence, we have the desired result. 
\end{proof}

The set $\mathcal{C}_0[S]$ is a complete metric space under the metric $d_{\sup}$. 
Recall that, under \cref{assumptionstate}, $S_n$ is a non-decreasing sequence of compact subsets of $S$ such that $S=\bigcup_{n\in \Nats}S_n$. 
For two elements $g_1, g_2\in \mathcal{C}_0[S]$, define $d_{\sup, n}(g_1, g_2)=\sup_{s\in S_n}|g_1(s)-g_2(s)|$. 
Define $d_{\mathrm{unif}}(g_1, g_2)=\sum_{n\in \Nats}\frac{\min\{1, d_{\sup, n}(g_1, g_2)\}}{2^{n}}$. 
Note that $d_{\mathrm{unif}}$ is a well-defined complete metric on $\mathcal{C}_0[S]$. 
For every $f\in \mathcal{C}_0[S]$, under the topology generated by the metric $d_{\mathrm{unif}}$, $F\in \NSE{\mathcal{C}_0}[\NSE{S}]$ is in the monad of $\NSE{f}$ if $F(s)\approx \NSE{f}(s)$ for all $s\in \NS{\NSE{S}}$. 

\begin{lemma}\label{VVapproxsigma}
Suppose \cref{assumptionstate} and \cref{assumptionpayoff} hold.
Then $\NSE{V}(s)\approx \mathbb{V}(s)$ for all $s\in \NS{\NSE{S}}$. 
\end{lemma}
\begin{proof}
Let $V_0$ be the restriction of $\NSE{V}$ to $\NSE{S_N}$. 
For all $(s, x)\in \NS{\NSE{S}}\times \NSE{X}$, by \cref{keyapproxlemma}, we have
$
\max_{x\in \NSE{X}}\int_{\NSE{S_N}}\{\NSE{\pi}_N(s,x,s')+\delta V_0(s')\}\bar{\NSE{Q}}_{\nu}^{N}(\dee s'|s, x)
=V(\ST(s))\approx V_0(s). 
$ Let $G(f)(s)=\max_{x\in \NSE{X}}\int_{\NSE{S_N}}\{\NSE{\pi}_N(s,x,s')+\delta \NSE{f}(s')\}\bar{\NSE{Q}}_{\nu}^{N}(\dee s'|s, x)$ for all $f\in \NSE{\mathcal{C}_0}(\NSE{S_N})$.
Consider the following internal iterated process: start with $V_0$ and define a sequence  $\{V_n\}_{n\in \NSE{\Nats}}$ by $V_{n+1}=G(V_n)$.
As $\delta\in [0,1)$ and $\NSE{S_N}$ is a $\NSE{}$compact set, there exists some $K\in \NSE{\Nats}$ such that $\NSE{d}_{\sup}(V_K, V_{K+1})<1$. 
Hence the internal sequence $\{V_n\}_{n\in \NSE{\Nats}}$ is a $\NSE{}$Cauchy sequence with respect to the $\NSE{}$metric $\NSE{d}_{\mathrm{unif}}$.
As $\NSE{\mathcal{C}_0}(\NSE{S_N})$ is $\NSE{}$complete with respect to $\NSE{d}_{\mathrm{unif}}$, the internal sequence $\{V_n\}_{n\in \NSE{\Nats}}$ has a $\NSE{}$limit. 
Note that $\NSE{d}_{\mathrm{unif}}(G(f_1), G(f_2))\leq \NSE{d}_{\mathrm{unif}}(f_1, f_2)$ for all $f_1, f_2\in \NSE{\mathcal{C}_0}(\NSE{S_N})$. 
So $G$ is a $\NSE{}$continuous function, hence the $\NSE{}$limit of the internal sequence $\{V_n\}_{n\in \NSE{\Nats}}$ is the $\NSE{}$fixed point $\mathbb{V}$. 
As $\NSE{d_{\mathrm{unif}}}(V_0, \mathbb{V})\approx \NSE{d_{\mathrm{unif}}}(V_1, V_0)\approx 0$, we have $\NSE{V}(s)\approx \mathbb{V}(s)$ for all $s\in \NS{\NSE{S}}$.
\end{proof}

\begin{theorem}\label{optimalsigma}
Suppose \cref{assumptionstate} and \cref{assumptionpayoff} hold.
For every $(s, x)\in S\times X$ that is in the support of $\pd{m}$,  $x$ is optimal given $s$ in the $\text{MDP}(\bar{Q}_{\pd{\nu}})$.
\end{theorem}
\begin{proof}
Pick $(s, x)\in S\times X$ that is in the support of $\pd{m}$. 
Then there exists some $(a, b)\in \NS{\NSE{S}}\times \NSE{X}$ such that $(a, b)\approx (s, x)$ and $(a, b)$ is in the $\NSE{}$support of $m$. 
Thus, we have $b\in \argmax_{y\in \NSE{X}}\int_{\NSE{S_N}}\{\NSE{\pi}_N(a,y,s')+\delta \mathbb{V}(s')\}\bar{\NSE{Q}}_{\nu}^{N}(\dee s'|a, y)$.
\begin{claim}\label{Vbdresult}
$\mathbb{V}$ is bounded. 
\end{claim}
\nt{\textit{Proof of Claim \ref{Vbdresult}}.}
Let $G(f)(s)=\max_{x\in \NSE{X}}\int_{\NSE{S_N}}\{\NSE{\pi}_N(s,x,s')+\delta \NSE{f}(s')\}\bar{\NSE{Q}}_{\nu}^{N}(\dee s'|s, x)$ for all $f\in \NSE{\mathcal{C}_0}(\NSE{S_N})$ and $F_0: \NSE{S_N}\to \NSE{\Reals}$ be the constant $0$ function. 
Consider the following internal iterated process: start with $F_0$ and define a sequence  $\{F_n\}_{n\in \NSE{\Nats}}$ by $F_{n+1}=G(F_n)$.
The $\NSE{}$limit (with respect to $\NSE{d}_{\sup}$) of the internal sequence $\{F_n\}_{n\in \NSE{\Nats}}$ is $\mathbb{V}$. 
By the transfer of the Banach fixed point theorem, we know that $\NSE{d}_{\sup}(F_0, \mathbb{V})\leq \frac{1}{1-\delta}\NSE{d}_{\sup}(F_0, F_1)$. 
As $\NSE{\pi}_{N}$ is bounded, we conclude that $\mathbb{V}$ is bounded. 

By \cref{Vbdresult}, Arkeryd et al. (1997, Section 4, Corollary 6.1), \cref{VVapproxsigma}, \cref{Qbarlemmasigma} and \cref{pdint}:
\begin{align*}
\int_{\NSE{S_N}}\{\NSE{\pi}_N(a,y,s')+\delta \mathbb{V}(s')\}\bar{\NSE{Q}}_{\nu}^{N}(\dee s'|a, y)&\approx \lim_{n\to \infty}\int_{S_n}\{\pi(s,\ST(y),s')+\delta V(s')\}\bar{Q}_{\pd{\nu}}(\dee s'|s,\ST(y)) &\\&=\int_{S}\{\pi(s,\ST(y),s')+\delta V(s')\}\bar{Q}_{\pd{\nu}}(\dee s'|s,\ST(y))
\end{align*}
for all $y\in \NSE{X}$.

Thus, we have $x\in \argmax_{\hat{x}\in X}\int_{S}\{\pi(s,\hat{x},s')+\delta V(s')\}\bar{Q}_{\pd{\nu}}(\dee s'|s,\hat{x})$,
which implies that $x$ is optimal given $s$ in the MDP($\bar{Q}_{\pd{\nu}}$).
\end{proof}

\subsubsection{Optimality with Unbounded Payoff Function}

In this section, we establish optimality of of the candidate Berk-Nash equilibrium $\pd{m}$ with possibly unbounded payoff function under \cref{assumptionpayoffubd}, \cref{assumptionsint} and \cref{assumptioncontwass}.
Let $\|s\|, d_{S}$ denote the norm of an element $s\in S$ and the metric on $S$, respectively.  
Let $W(\mu, \nu)$ denote the Wasserstein distance between two probability measures $\mu$ and $\nu$.

\begin{lemma}\label{closewasser}
Suppose \cref{assumptioncontwass} holds. 
For every $\lambda\in \NSE{\PM{T_{\Theta}}}$ and every $(s, x)\in \NS{\NSE{S}}\times \NSE{X}$, $\NSE{W}\big(\bar{\NSE{Q}}_{\lambda}(s, x), \bar{\NSE{Q}}_{\NSE{\pd{\lambda}}}(\ST(s), \ST(x))\big)\approx 0$. That is, $\bar{\NSE{Q}}_{\lambda}(s, x)$ is in the monad of $\bar{Q}_{\pd{\lambda}}(\ST(s), \ST(x))$ with respect to the $1$-Wasserstein metric. 
\end{lemma}
\begin{proof}
Fix $\lambda\in \NSE{\PM{T_{\Theta}}}$ and $(s, x)\in \NS{\NSE{S}}\times \NSE{X}$.
Note that convergence in the Wasserstein metric is equivalent to weak convergence plus convergence of the first moments.
By \cref{Qbarlemmasigma}, it is sufficient to show that 

$$
\int_{\NSE{S}}\NSE{d_{S}}(t, s_0)\bar{\NSE{Q}}_{\lambda}(s, x)(\dee t)\approx \int_{S}d_{S}(t, \ST(s_0))\bar{Q}_{\pd{\lambda}}(\ST(s), \ST(x))(\dee t)
$$ for all $s_0\in \NS{\NSE{S}}$. 
By \cref{assumptioncontwass}, we have
\\
$$
\int_{\NSE{S}}\NSE{d_{S}}(t, s_0)\NSE{Q}_{\theta}(s, x)(\dee t)\approx \int_{S}d_{S}(t, \ST(s_0))Q_{\ST(\theta)}(\ST(s), ST(x))(\dee t).
$$\\
for all $\theta\in T_{\Theta}$. 
By \cref{pdint}, we have
\begin{align*}
\int_{\NSE{S}}\NSE{d_{S}}(t, s_0)\bar{\NSE{Q}}_{\lambda}(s, x)(\dee t)
&\approx \int_{\Theta}\int_{S}d_{S}(t, \ST(s_0))Q_{\theta}(\ST(s), \ST(x))(\dee t)\pd{\lambda}(\dee \theta)\\
&=\int_{S}d_{S}(t, \ST(s_0))\bar{Q}_{\pd{\lambda}}(\ST(s), \ST(x))(\dee t). 
\end{align*}
Hence, we have the desired result. 
\end{proof}

We now consider the Bellman equation.
$\label{unboundedbellman}
V(s)=\max_{x\in X}\int_{S}\{\pi(s,x,s')+\delta V(s')\}\bar{Q}_{\pd{\nu}}(\dee s'|s, x). 
$
For each $n\in \Nats$ and any two elements $g_1, g_2\in \mathcal{C}[S]$ (the set of continuous real-valued functions on $S$), let $d_{\sup, n}(g_1, g_2)=\sup_{s\in S_n}|g_1(s)-g_2(s)|$.\fn{
Recall that the uniform convergence topology on compact sets on $\mathcal{C}[S]$ can be generated from the metric $d_{\mathrm{unif}}(g_1, g_2)=\sum_{n\in \Nats}\frac{\min\{1, d_{\sup, n}(g_1, g_2)\}}{2^{n}}$. Note that $\mathcal{C}[S]$ equipped with $d_{\mathrm{unif}}$ is a complete metric space.} 
Let $B,D$ be constants in \cref{assumptionpayoffubd} and \cref{assumptionsint}, respectively.
Define
\begin{align*}
\mathcal{L}_{B,D}[S]=\{f\in \mathcal{C}[S]: (\exists E\in \PosReals)(\forall s\in S)(|f(s)|\leq E+(B+D)\|s\|)\},
\end{align*}
which is a complete metric space under the metric $d_{\mathrm{unif}}$. We present three lemmas, \cref{mapin}-\cref{nssolnsint}, proofs of which are provided in \cref{onlineomit}.

\begin{lemma}\label{mapin}
Suppose \cref{assumptionstate}, \cref{assumptionpayoffubd}, \cref{assumptionsint} and \cref{assumptioncontwass} hold.
The Bellman operator $F(g)(s)=\max_{x\in X}\int_{S}\{\pi(s,x,s')+\delta g(s')\}\bar{Q}_{\pd{\nu}}(\dee s'|s, x)$ maps every element in $\mathcal{L}_{B,D}[S]$ to some element in $\mathcal{L}_{B,D}[S]$.
\end{lemma}

We use $V$ to denote the unique solution of the Bellman equation, and fix this for the rest of this section.\fn{The Bellman operator $F(g)(s)=\max_{x\in X}\int_{S}\{\pi(s,x,s')+\delta g(s')\}\bar{Q}_{\pd{\nu}}(\dee s'|s, x)$
is a contraction mapping on $\mathcal{L}_{B,D}[S]$.  
Given any $g_0\in \mathcal{L}_{B,D}[S]$, let $\{g_n\}_{n\geq 0}$ be the sequence such that $g_{n+1}=F(g_n)$ for all $n\geq 0$. 
The sequence $\{g_n\}_{n\geq 0}\subset \mathcal{L}_{B,D}[S]$ is a Cauchy sequence with respect to the metric $d_{\mathrm{unif}}$. 
This is because, for every $n\in \Nats$, there exists some $K\in \Nats$ such that $d_{\sup, n}(g_K, g_{K+1})<1$. 
The limit of the sequence $\{g_n\}_{n\geq 0}$ is the unique fixed point of the Bellman operator, hence is the solution of the Bellman equation. }

\begin{lemma}\label{keyapproxlemmaunbd}
Suppose \cref{assumptionstate}, \cref{assumptionpayoffubd}, \cref{assumptionsint} and \cref{assumptioncontwass} hold.
For every $(s, x)\in \NS{\NSE{S}}\times \NSE{X}$:
$$\int_{\NSE{S_N}}\{\NSE{\pi}_N(s,x,s')+\delta \NSE{V}(s')\}\bar{\NSE{Q}}_{\nu}^{N}(\dee s'|s, x)\ \approx \int_{S}\{\pi(\ST(s), \ST(x), s')+\delta V(s')\}\bar{Q}_{\pd{\nu}}(\dee s'|\ST(s), \ST(x)).$$
\end{lemma}


The nonstandard Bellman equation is:
$$\label{nsunboundednsbellman}
\mathbb{V}(t)=\max_{x\in \NSE{X}}\int_{\NSE{S_N}}\{\NSE{\pi}_N(s,x,s')+\delta \mathbb{V}(s')\}\bar{\NSE{Q}}_{\nu}^{N}(\dee s'|s, x).
$$
Let $\NSE{\mathcal{C}_0}[\NSE{S_N}]$ is the set of $\NSE{}$bounded continuous functions on $\NSE{S_N}$.
Note that $\NSE{\pi}_{N}$ is an element in $\NSE{\mathcal{C}_0}[\NSE{S_N}]$.
By the transfer of the Banach fixed point theorem, there exists a unique solution $\mathbb{V}$ of the nonstandard Bellman equation, which we fix for the rest of this section.

\begin{lemma}\label{VVapproxsigmaunbd}
Suppose \cref{assumptionstate}, \cref{assumptionpayoffubd}, \cref{assumptionsint} and \cref{assumptioncontwass} hold.
Then $\NSE{V}(s)\approx \mathbb{V}(s)$ for all $s\in \NS{\NSE{S}}$. 
\end{lemma}

To complete the proof of the main result of this section, we need to show that the solution $\mathbb{V}$ of the nonstandard Bellman equation is S-integrable. 
Let 
$
\NSE{\mathcal{L}_{B,D}}[\NSE{S_N}]=\{f\in \NSE{\mathcal{C}_0}[\NSE{S_N}]: (\exists E\in \NSE{\PosReals})(\forall s\in \NSE{S_N})(|f(s)|\leq E+(B+D)\|s\|)\}. 
$
$\NSE{\mathcal{L}_{B,D}}[\NSE{S_N}]$ is a 
$\NSE{}$complete metric space under the $\NSE{}$metric $\NSE{d}_{\sup}$, since it is a $\NSE{}$closed subset of $\NSE{\mathcal{C}_0}[\NSE{S_N}]$ under $\NSE{d}_{\sup}$.
\begin{lemma}\label{nsmaplinear}
Suppose \cref{assumptionstate}, \cref{assumptionpayoffubd} and \cref{assumptionsint} hold. 
The nonstandard Bellman operator
$
G(f)(s)=\max_{x\in \NSE{X}}\int_{\NSE{S_N}}\{\NSE{\pi}_N(s,x,s')+\delta \NSE{f}(s')\}\bar{\NSE{Q}}_{\nu}^{N}(\dee s'|s, x)
$
maps every element in $\NSE{\mathcal{L}_{B,D}}[\NSE{S_N}]$ to some element in $\NSE{\mathcal{L}_{B,D}}[\NSE{S_N}]$.
\end{lemma}

Hence, we conclude that the solution $\mathbb{V}$ of the the nonstandard Bellman equation is an element of $\NSE{\mathcal{L}_{B,D}}[\NSE{S_N}]$. 

\begin{lemma}\label{nssolnsint}
Suppose \cref{assumptionstate}, \cref{assumptionpayoffubd}, \cref{assumptionsint} and \cref{assumptioncontwass} hold.
Then $\mathbb{V}$ is S-integrable with respect to $\bar{\NSE{Q}}_{\nu}^{N}(s, x)$ when $(s, x)\in \NS{\NSE{S}}\times \NSE{X}$. 
\end{lemma}


\begin{theorem}\label{optunbdpayoff}
Suppose \cref{assumptionstate}, \cref{assumptionpayoffubd}, \cref{assumptionsint} and \cref{assumptioncontwass} hold.
For every $(s, x)\in S\times X$ that is in the support of $\pd{m}$, $x$ is optimal given $s$ in MDP($\bar{Q}_{\pd{\nu}}$).
\end{theorem}
\begin{proof}
Pick $(s, x)\in S\times X$ that is in the support of $\pd{m}$. 
Then there exists some $(a, b)\in \NS{\NSE{S}}\times \NSE{X}$ such that $(a, b)\approx (s, x)$ and $m(\{(a, b)\})>0$. 
Thus, we have
\\
$$ b\in \argmax_{y\in \NSE{X}}\int_{\NSE{S_N}}\{\NSE{\pi}_N(a,y,s')+\delta \mathbb{V}(s')\}\bar{\NSE{Q}}_{\nu}^{N}(\dee s'|a, y). 
$$
\\
By \cref{nssolnsint}, $\mathbb{V}$ is S-integrable with respect to $\bar{\NSE{Q}}_{\nu}^{N}(\dee s'|a, y)$ for all $y\in \NSE{X}$.
Using similar argument, $\NSE{\pi}_{N}(a,y,\cdot)$ is also S-integrable with respect to $\bar{\NSE{Q}}_{\nu}^{N}(\dee s'|a, y)$ for all $y\in \NSE{X}$. 
Thus, by Arkeryd et al. (1997, Section 4, Theorem 6.2), \cref{VVapproxsigmaunbd}, \cref{Qbarlemmasigma} and \cref{pdint}:
\begin{align*}
\int_{\NSE{S_N}}\{\NSE{\pi}_N(a,y,s')+\delta \mathbb{V}(s')\}\bar{\NSE{Q}}_{\nu}^{N}(\dee s'|a, y)
&\approx \lim_{n\to \infty}\int_{S_n}\{\pi(s,\ST(y),s')+\delta V(s')\}\bar{Q}_{\pd{\nu}}(\dee s'|s,\ST(y))\\
&=\int_{S}\{\pi(s,\ST(y),s')+\delta V(s')\}\bar{Q}_{\pd{\nu}}(\dee s'|s,\ST(y)).
\end{align*}
for all $y\in \NSE{X}$. Thus, we have $x\in \argmax_{\hat{x}\in X}\int_{S}\{\pi(s,\hat{x},s')+\delta V(s')\}\bar{Q}_{\pd{\nu}}(\dee s'|s,\hat{x})$,
which implies that $x$ is optimal given $s$ in the MDP($\bar{Q}_{\pd{\nu}}$).
\end{proof}

\section{Supplementary Material - For Online Publication}
\label{sec:onlineappendix}

This supplementary material is divided into three subsections:
(i) proofs and statements that are omitted from Appendix A, (ii) asymptotic characterization of state-action frequencies, and (iii) a detailed analysis of the examples covered in the main paper.

\subsection{Omitted Proofs}\label{onlineomit}

Theorems \ref{BLintegral} and \ref{Lintegral} are invoked at several instances during the proofs of the main theorems in our paper. We list them for completeness here.

\begin{theorem}[{\citet[][Section.~4, Corollary~6.1]{NSAA97}}]\label{BLintegral}
Suppose $(\Omega, \cA, P)$ is an internal probability space, and 
$F: \Omega\to \NSE{\Reals}$ is an internally integrable function such that $\ST({F})$ exists everywhere.
Then $F$ is S-integrable. 
\end{theorem}

\begin{theorem}[{\citet[][Section.~4, Theorem~6.2]{NSAA97}}]\label{Lintegral} 
Suppose $(\Omega,\cA,P)$ is an internal probability space, and $F: \Omega\to \NSE{\Reals}$ is an internally integrable function such that $\ST({F})$ exists $\Loeb{P}$-almost surely. Then the following are equivalent:
\begin{enumerate}
\item $\ST({\int |F|\dee P})$ exists and it equals to $\lim_{n\to \infty}\ST({\int |F_n|\dee P})$ where for $n\in \Nats$, $F_n=\min\{F,n\}$ when $F\geq 0$ and $F_n=\max\{F,-n\}$ when $F\leq 0$;

\item For every infinite $K>0$, $\int_{|F|>K}|F|\dee P\approx 0$;

\item $\ST({\int |F|\dee P})$ exists, and for every $B$ with $P(B)\approx 0$, we have $\int_{B}|F|\dee P\approx 0$;

\item $F$ is S-integrable with respect to $P$. 
\end{enumerate}
\end{theorem}

\medskip

\nt We next provide a proof to \cref{eptapprox} which is used to prove \cref{KLapprox}.

\medskip

\nt{\textit{\textbf{Proof of Claim \ref{eptapprox}.}}}\label{dominate}
Pick $(s, x)\in T_S\times T_X$ such that $Q(\ST(s), \ST(x))$ is dominated by $Q_{\ST(\theta)}(\ST(s), \ST(x))$.
Then, $D_{\ST(\theta)}(\cdot|\ST(s), \ST(x))$ is the density function of $Q(\ST(s), \ST(x))$ with respect to $Q_{\ST(\theta)}(\ST(s), \ST(x))$.
Let $f: S\to \Reals$ be $f(t)=\ln \big(D_{\ST(\theta)}(t|\ST(s), \ST(x))\big)$. 
For $n\in \Nats$,  define $f_n: S\to \Reals$ to be:

$$
f_{n}(t)=
\begin{cases}
f(t) & \text{If $\frac{1}{n}\leq D_{\ST(\theta)}(t|\ST(s), \ST(x))\leq n$}\\
\frac{1}{n} & \text{if $D_{\ST(\theta)}(t|\ST(s), \ST(x))<\frac{1}{n}$}\\
n & \text{if $D_{\ST(\theta)}(t|\ST(s), \ST(x))>n$}\\
\end{cases}
$$
Note that $f_n$ is a bounded continuous function. Moreover, by \cref{KLint} of \cref{regsmdp}, we have $E_{Q(\cdot|\ST(s), \ST(x))}=\lim_{n\to \infty}\int_{S}f_{n}(t)Q(\dee t|\ST(s), \ST(x))$.
Let $F: T_S\to \NSE{\Reals}$ be $F(t)=\ln \big(\frac{\mathbb{Q}(t|s,x)}{\mathbb{Q}_{\theta}(t|s,x)}\big)$. 
For $n\in \Nats$,  define $F_n: T_S\to \NSE{\Reals}$ to be:

$$
F_{n}(t)=
\begin{cases}
F(t) & \text{If $\frac{1}{n}\leq \frac{\mathbb{Q}(t|s,x)}{\mathbb{Q}_{\theta}(t|s,x)}\leq n$}\\
\frac{1}{n} & \text{if $\frac{\mathbb{Q}(t|s,x)}{\mathbb{Q}_{\theta}(t|s,x)}<\frac{1}{n}$}\\
n & \text{if $\frac{\mathbb{Q}(t|s,x)}{\mathbb{Q}_{\theta}(s|t,x)}>n$}\\
\end{cases}
$$
By \cref{dsyapprox}, we know that $F_n(t)\approx f_n(\ST(t))$ for every $n\in \Nats$ and $t\in T_S$. 
For every $n\in \Nats$, we have
\begin{align*}
\int_{S}f_{n}(t)Q(\dee t|\ST(s), \ST(x))&=\int_{\NSE{S}}\NSE{f}_n(t)\NSE{Q}(\dee t|\ST(s),\ST(x))\\
&\approx \int_{\NSE{S}}\NSE{f}_n(t)\NSE{Q}(\dee t|s,x)\\
&\approx \sum_{i\in T_S}F_n(i)\mathbb{Q}(i|s, x)\\
\end{align*}
\nt To finish the proof, it remains to show that 
$\lim_{n\to \infty}\ST\big(\sum_{t\in T_S}F_n(i)\mathbb{Q}(t|s, x)\big)\approx \mathbb{E}_{\mathbb{Q}(\cdot|s, x)}[F(t)]$.

By \cref{Lintegral}, this is the same as establishing the S-integrability of $F(t)$ under $\mathbb{Q}(\cdot|s, x)$. 
Pick an infinite $K>0$ and let $I_{K}=\{t\in T_S: |F(t)|>K\}$. 
Let $I_{K}^{0}=\{t\in T_S: |F(t)|>K\wedge \frac{\mathbb{Q}(i|s,x)}{\mathbb{Q}_{\theta}(i|s,x)}\leq 1\}$
and $I_{K}^{\infty}=\{t\in T_S: |F(t)|>K\wedge \frac{\mathbb{Q}(i|s,x)}{\mathbb{Q}_{\theta}(i|s,x)}>1\}$.
It is easy to see that both $I_{K}^{0}$ and $I_{K}^{\infty}$ are internal sets and $I_{K}=I_{K}^{0}\cup I_{K}^{\infty}$. 
For all $t\in I_{K}^{0}$, we have $\frac{\mathbb{Q}(t|s,x)}{\mathbb{Q}_{\theta}(t|s,x)}\approx 0$. 
Then we have 

\begin{equation}\label{zinfinitesimal}
\sum_{t\in I_{K}^{0}}|F(t)|\mathbb{Q}(t|s, x)=\sum_{t\in I_{K}^{0}}|F(t)|\frac{\mathbb{Q}(t|s,x)}{\mathbb{Q}_{\theta}(t|s,x)}\mathbb{Q}_{\theta}(t|s, x)\approx 0.
\end{equation}

For all $t\in I_{K}^{\infty}$, $\frac{\mathbb{Q}(t|s,x)}{\mathbb{Q}_{\theta}(t|s,x)}>n$ for all $n\in \Nats$. 
By \cref{dsyapprox}, $D_{\ST(\theta)}(\ST(t)|\ST(s), \ST(x))=\infty$ for all $t\in I_{K}^{\infty}$. 
This implies that $Q_{\ST(\theta)}(\ST(s), \ST(x))\big(\ST(I_{K}^{\infty})\big)=0$.
By \cref{wkconverge}, we conclude that $\NSE{Q}_{\theta}(s, x)(\bigcup_{t\in I_{K}^{\infty}}B_{S}(t))\approx 0$. 
By \cref{KLint} in \cref{regsmdp}, we conclude that $\int_{\bigcup_{t\in I_{K}^{\infty}}B_{S}(t)}\big(\NSE{D}_{\theta}(s'|s, x)\big)^{1+r}\NSE{Q}_{\theta}(s, x)(\dee s')\approx 0$.
This implies that

\begin{equation}\label{iinfinitesimal}
\sum_{t\in I_{K}^{\infty}}|F(t)|\mathbb{Q}(t|s, x)=\sum_{t\in I_{K}^{\infty}}|F(t)|\frac{\mathbb{Q}(t|s,x)}{\mathbb{Q}_{\theta}(t|s,x)}\mathbb{Q}_{\theta}(t|s, x)\approx 0.
\end{equation}
Combining \cref{zinfinitesimal} and \cref{iinfinitesimal}, we have the desired result. 
\qed

\nt{\textit{\textbf{Proof of Lemma \ref{exnsexeq}}.}}\label{infinitesimaldominate}
Pick $\theta\in T_{\Theta}$ such that $\ST(\theta)\in \Theta_{\pd{m}}$ and $(s, x)\in \NS{\NSE{S}}\times \NSE{X}$ such that 
$Q(\ST(s), \ST(x))$ is dominated by $Q_{\ST(\theta)}(\ST(s), \ST(x))$. 
By \cref{densityapprox} and the fact that $\NSE{Q}_{\theta}(s, x)(\NSE{S_N})\approx 1$, we have

\begin{align*}
\NSE{\mathbb{E}}_{\NSE{Q}^{N}(\cdot|s, x)}\left[\ln \big(\mathbb{D}_{\theta}(s'|s,x)\big)\right]
\approx \int_{\NSE{S_N}}\NSE{D}_{\theta}(s'|s, x)\ln \big(\NSE{D}_{\theta}(s'|s,x)\big) \NSE{Q}_{\theta}(s, x)(\dee s'). \\
\end{align*}
By \cref{KLint} of \cref{regsmdp} and the fact that $s$ is near-standard, we conclude that 
$$\int_{\NSE{S_N}}\NSE{D}_{\theta}(s'|s, x)\ln \big(\NSE{D}_{\theta}(s'|s,x)\big) \NSE{Q}_{\theta}(s, x)(\dee s')\approx \int_{\NSE{S}}\NSE{D}_{\theta}(s'|s, x)\ln \big(\NSE{D}_{\theta}(s'|s,x)\big) \NSE{Q}_{\theta}(s, x)(\dee s').$$

\nt $D_{\ST(\theta)}(s'|\ST(s), \ST(x))\ln \big(D_{\ST(\theta)}(s'|\ST(s),\ST(x))\big)$ is a continuous and bounded on $S_n$ for each $n\in \Nats$.  So, for every $n\in \Nats$, we have

\begin{align*}
&\int_{S_n}D_{\ST(\theta)}(s'|\ST(s), \ST(x))\ln \big(D_{\ST(\theta)}(s'|\ST(s),\ST(x))\big)Q_{\ST(\theta)}(\ST(s), \ST(x))(\dee s')\\
&\approx \int_{\NSE{S_n}}\NSE{D}_{\ST(\theta)}(s'|\ST(s), \ST(x))\ln \big(\NSE{D}_{\ST(\theta)}(s'|\ST(s),\ST(x))\big)\NSE{Q}_{\theta}(s, x)(\dee s')\\
&\approx \int_{\NSE{S_n}}\NSE{D}_{\theta}(s'|s, x)\ln \big(\NSE{D}_{\theta}(s'|s,x)\big)\NSE{Q}_{\theta}(s, x)(\dee s').
\end{align*}
Note that, we have 
\begin{align*}
&\lim_{n\to \infty}\int_{S_n}D_{\ST(\theta)}(s'|\ST(s), \ST(x))\ln \big(D_{\ST(\theta)}(s'|\ST(s),\ST(x))\big)Q_{\ST(\theta)}(\ST(s), \ST(x))(\dee s')\\
&=\mathbb{E}_{Q(\cdot|\ST(s), \ST(x))}\left[\ln \big(D_{\ST(\theta)}(s'|\ST(s),\ST(x))\big)\right].
\end{align*}
\nt By \cref{KLint} of \cref{regsmdp} and the fact that $s$ is near-standard, we have
\begin{align*}
&\lim_{n\to\infty}\ST\big(\int_{\NSE{S_n}}\NSE{D}_{\theta}(s'|s, x)\ln \big(\NSE{D}_{\theta}(s'|s,x)\big)\NSE{Q}_{\theta}(s, x)(\dee s')\big)\\
&\approx \int_{\NSE{S}}\NSE{D}_{\theta}(s'|s, x)\ln \big(\NSE{D}_{\theta}(s'|s,x)\big)\NSE{Q}_{\theta}(s, x)(\dee s')\\
&\approx \NSE{\mathbb{E}}_{\NSE{Q}^{N}(\cdot|s, x)}\left[\ln \big(\mathbb{D}_{\theta}(s'|s,x)\big)\right].
\end{align*}
\smallskip
\nt Hence, we have the desired result. 
\qed

\bigskip

\nt{\textit{\textbf{Proof of Theorem \ref{stKLnsKL}.}}}
Pick $\theta\in T_{\Theta}$ such that $\ST(\theta)\in \Theta_{\pd{m}}$. 
Since $K_{Q}(\pd{m}, \ST(\theta))<\infty$,  $Q(s, x)$ is dominated by $Q_{\ST(\theta)}(s, x)$ for $\pd{m}$-almost all $(s, x)\in S\times X$. 
Let $\bar{\Reals}$ denote the extended real line and define $g: S\times X \to \bar{\Reals}$ to be  $g(s, x)=\mathbb{E}_{Q(\cdot|s, x)}\left[\ln \big(D_{\ST(\theta)}(s'|s,x)\big)\right]$ if $Q(s, x)$ is dominated by $Q_{\ST(\theta)}(s, x)$ and $g(s, x)=\infty$ otherwise. 
We have
$K_{Q}(\pd{m}, \ST(\theta))=\int_{S\times X}g(s, x)\pd{m}(\dee s, \dee x)=\lim_{n\to\infty}\int_{S_n\times X}g(s, x)\pd{m}(\dee s, \dee x).
$ Let $G: \NSE{S_N}\times \NSE{X}\to \NSE{\Reals}$ be $G(s, x)=\NSE{\mathbb{E}}_{\NSE{Q}^{N}(\cdot|s, x)}\left[\ln \big(\mathbb{D}_{\theta}(s'|s,x)\big)\right]$.
By \cref{exnsexeq} and \cref{pdint}, we have $\int_{S_n\times X}g(s, x)\pd{m}(\dee s, \dee x)\approx \int_{S_n\times X}G(s, x)m(\dee s, \dee x)$.
To finish the proof, it is sufficient to show that $G$ is S-integrable with respect to $m$. 
As $\theta\in T_{\Theta}\subset \NSE{\hat{\Theta}}$ and $m$ is $\NSE{}$stationary, by \cref{assumptionrebound}, $\NSE{\mathbb{E}}_{\NSE{Q}(\cdot|s, x)}\left[\ln \big(\NSE{D}_{\theta}(s'|s,x)\big)\right]$ is S-integrable with respect to $m$. 
By \cref{minbound} of \cref{assumptionstate}, $\NSE{\mathbb{E}}_{\NSE{Q}^{N}(\cdot|s, x)}\left[\ln \big(\mathbb{D}_{\theta}(s'|s,x)\big)\right]$ is S-integrable with respect $m$, completing the proof. 

\qed

\nt{\textit{\textbf{Proof of Lemma \ref{mapin}}.}}
Let $g$ be some element in $\mathcal{L}_{B,D}[S]$. 
Then there exists some $E\in \PosReals$ such that $|g(s)|\leq E+(B+D)\|s\|$ for all $s\in S$. 
We show that $F(g)$ is continuous. 
\begin{claim}\label{sintegrabclaim}
For every $(s, x)\in \NS{\NSE{S}}\times \NSE{X}$, $\NSE{\pi}(s, x, \cdot)+\delta \NSE{g}(\cdot)$ is S-integrable with respect to $\bar{\NSE{Q}}_{\NSE{\pd{\nu}}}(s, x)(\cdot)$.
\end{claim}

\nt{\textit{Proof of Claim \ref{sintegrabclaim}}.}
Let $(s_0, x_0)\in \NS{\NSE{S}}\times \NSE{X}$ be given. 
By \cref{assumptionpayoffubd}, it is sufficient to show that $\|s'\|$ is S-integrable with respect to $\bar{\NSE{Q}}_{\NSE{\pd{\nu}}}(s_0, x_0)(\cdot)$. By \cref{closewasser}, we have:
\begin{align*}
\int_{\NSE{S}}\|s'\|\bar{\NSE{Q}}_{\NSE{\pd{\nu}}}(\dee s'|s_0, x_0)&\approx \int_{S}\|s'\|\bar{Q}_{\pd{\nu}}(\dee s'|\ST(s_0), \ST(x_0))\\
&=\lim_{n\to \infty}\ST\big(\int_{\NSE{S_n}}\|s'\|\bar{\NSE{Q}}_{\NSE{\pd{\nu}}}(\dee s'|s_0, x_0)\big).\\ 
\end{align*}
By \cref{Lintegral}, $\|s'\|$ is S-integrable with respect to $\bar{\NSE{Q}}_{\NSE{\pd{\nu}}}(\dee s'|s_0, x_0)$. By \cref{sintegrabclaim} and \cref{pdint}, for every $(s, x)\in \NS{\NSE{S}}\times \NSE{X}$, we have
\begin{align*}
&\int_{\NSE{S}}\{\NSE{\pi}(s,x,s')+\delta \NSE{g}(s')\}\bar{\NSE{Q}}_{\NSE{\pd{\nu}}}(\dee s'|s, x)\\
&\approx \lim_{n\to \infty} \ST\big(\int_{\NSE{S_n}}\{\NSE{\pi}(s,x,s')+\delta \NSE{g}(s')\}\bar{\NSE{Q}}_{\NSE{\pd{\nu}}}(\dee s'|s, x)\big)\\
&=\lim_{n\to \infty}\int_{S_n}\{\pi(\ST(s),\ST(x),s')+\delta g(s')\}\bar{Q}_{\pd{\nu}}(\dee s'|\ST(s), \ST(x))\\
&=\int_{S}\{\pi(\ST(s),\ST(x),s')+\delta g(s')\}\bar{Q}_{\pd{\nu}}(\dee s'|\ST(s), \ST(x)).\\
\end{align*}
Hence, we have $\NSE{F}(\NSE{g})(s)\approx F(g)(\ST(s))$ for all $s\in \NS{\NSE{S}}$, so $F(g)$ is a continuous function. 
For every $s\in S$, by \cref{assumptionsint}, we have

\begin{align*}
|F(g)(s)|
&\leq \max_{x\in X}\int_{S}\{A+B\max\{\|s\|, \|s'\|\}+\delta (E+(B+D)\|s'\|)\}\bar{Q}_{\pd{\nu}}(\dee s'|s, x)\\
&\leq A+ \delta E+ \big(B+\delta(B+D)\big)C+(B+D)\|s\|\\
\end{align*}
Hence we have the desired result. 
\qed

\nt{\textit{\textbf{Proof of Lemma \ref{keyapproxlemmaunbd}}.}}
Pick $(s, x)\in \NS{\NSE{S}}\times \NSE{X}$. 
Note that $\bar{\NSE{Q}}_{\nu}(s, x)(\NSE{S_N})\approx 1$. 
Thus, we have
$
\int_{\NSE{S_N}}\{\NSE{\pi}_N(s,x,s')+\delta \NSE{V}(s')\}\bar{\NSE{Q}}_{\nu}^{N}(\dee s'|s, x)\approx \int_{\NSE{S_N}}\{\NSE{\pi}_N(s,x,s')+\delta \NSE{V}(s')\}\bar{\NSE{Q}}_{\nu}(\dee s'|s, x)
$

\begin{claim}\label{payoffsint}
$\NSE{\pi}_{N}(s, x, \cdot)+\delta \NSE{V}(\cdot)$ is S-integrable with respect to $\bar{\NSE{Q}}_{\nu}^{N}(s, x)$.
\end{claim}

\nt{\textit{Proof of Claim \ref{payoffsint}.}}
As $\bar{\NSE{Q}}_{\nu}(s, x)(\NSE{S_N})\approx 1$, by \cref{Lintegral}, it is sufficient to show $\NSE{\pi}_{N}(s, x, \cdot)+\delta \NSE{V}(\cdot)$ is S-integrable with respect to $\bar{\NSE{Q}}_{\nu}(s, x)(\cdot)$. 
By \cref{mapin}, there exists $E\in \PosReals$ such that $|V(s)|\leq E+(B+D)\|s\|$ for all $s\in S$. 
By \cref{assumptionpayoffubd} and \cref{assumptionsint}, it is sufficient to show that $\|s'\|$ is S-integrable with respect to $\bar{\NSE{Q}}_{\nu}(s, x)(\cdot)$. 
By \cref{closewasser} and \cref{pdint}, we have
\begin{align*}
\int_{\NSE{S}}\|s'\|\bar{\NSE{Q}}_{\nu}(\dee s'|s, x)&\approx \int_{S}\|s'\|\bar{Q}_{\pd{\nu}}(\dee s'|s, x)\\
&=\lim_{n\to \infty}\ST\big(\int_{\NSE{S_n}}\|s'\|\bar{\NSE{Q}}_{\nu}(\dee s'|s, x)\big). 
\end{align*}
By \cref{Lintegral}, $\|s'\|$ is S-integrable with respect to $\bar{\NSE{Q}}_{\nu}(s, x)(\cdot)$.  Thus, by Arkeryd et al. (1997, Section 4, Theorem 6.2) and \cref{pdint}, we have:
\begin{align*}
&\int_{\NSE{S_N}}\{\NSE{\pi}_{N}(s, x, s')+\delta \NSE{V}(s')\}\bar{\NSE{Q}}_{\nu}^{N}(\dee s'|s, x)\\
&\approx\lim_{n\to \infty}\int_{S_n}\{\pi(\ST(s), \ST(x), s')+\delta V(s')\}\bar{Q}_{\pd{\nu}}(\dee s'|\ST(s), \ST(x))\\
&=\int_{S}\{\pi(\ST(s), \ST(x), s')+\delta V(s')\}\bar{Q}_{\pd{\nu}}(\dee s'|\ST(s), \ST(x)).\\
\end{align*}
Hence, we have the desired result. 

\qed

\nt{\textit{\textbf{Proof of Lemma \ref{VVapproxsigmaunbd}}.}}
Let $V_0$ be the restriction of $\NSE{V}$ to $\NSE{S_N}$. 
For all $s\in \NS{\NSE{S}}$, by \cref{keyapproxlemmaunbd}, we have
\begin{align*}
&\max_{x\in \NSE{X}}\int_{\NSE{S_N}}\{\NSE{\pi}_N(s,x,s')+\delta V_0(s')\}\bar{\NSE{Q}}_{\nu}^{N}(\dee s'|s, x)\\
&\approx \max_{x\in X}\int_{S}\{\pi(\ST(s), \ST(x), s')+\delta V(s')\}\bar{Q}_{\pd{\nu}}(\dee s'|\ST(s), x)\\
&=V(\ST(s))\approx V_0(s). 
\end{align*}

Let $G(f)(s)=\max_{x\in \NSE{X}}\int_{\NSE{S_N}}\{\NSE{\pi}_N(s,x,s')+\delta \NSE{f}(s')\}\bar{\NSE{Q}}_{\nu}^{N}(\dee s'|s, x)$ for all $f\in \NSE{\mathcal{C}_0}(\NSE{S_N})$.
Note that $\NSE{\mathcal{C}_0}(\NSE{S_N})$ is a $\NSE{}$complete metric space under the $\NSE{}$metric $\NSE{d}_{\mathrm{unif}}$. 
Consider the following internal iterated process: start with $V_0$ and define a sequence  $\{V_n\}_{n\in \NSE{\Nats}}$ by $V_{n+1}=G(V_n)$.
As $\delta\in [0,1)$ and $\NSE{S_N}$ is a $\NSE{}$compact set, there exists some $K\in \NSE{\Nats}$ such that $\NSE{d}_{\sup}(V_K, V_{K+1})<1$. 
Hence the internal sequence $\{V_n\}_{n\in \NSE{\Nats}}$ is a $\NSE{}$Cauchy sequence with respect to the $\NSE{}$metric $\NSE{d}_{\mathrm{unif}}$.
As $\NSE{\mathcal{C}_0}(\NSE{S_N})$ is $\NSE{}$complete, the internal sequence $\{V_n\}_{n\in \NSE{\Nats}}$ has a $\NSE{}$limit. 
Note that $\NSE{d}_{\mathrm{unif}}(G(f_1), G(f_2))\leq \NSE{d}_{\mathrm{unif}}(f_1, f_2)$ for all $f_1, f_2\in \NSE{\mathcal{C}_0}(\NSE{S_N})$. 
So, $G$ is a $\NSE{}$continuous function, hence the $\NSE{}$limit of the internal sequence $\{V_n\}_{n\in \NSE{\Nats}}$ is the $\NSE{}$fixed point $\mathbb{V}$. 
As $\NSE{d_{\mathrm{unif}}}(V_0, \mathbb{V})\approx \NSE{d_{\mathrm{unif}}}(V_1, V_0)\approx 0$, we have $\NSE{V}(s)\approx \mathbb{V}(s)$ for all $s\in \NS{\NSE{S}}$.
\qed

\smallskip

\nt{\textit{\textbf{Proof of Lemma \ref{nsmaplinear}}.}}
Let $f$ be an arbitrary element in $\NSE{\mathcal{L}_{B,D}}(\NSE{S_N})$. 
Then, there is some $E\in \NSE{\PosReals}$ such that $|f(s)|\leq E+(B+D)\|s\|$ for all $s\in \NSE{S_N}$.  
By \cref{assumptionpayoffubd} and \cref{assumptionsint}:
\begin{align*}
\displaystyle
&|\int_{\NSE{S_N}}\!\!\!\!\!\!\!\!\{\NSE{\pi}_N(s,x,s')+\delta \NSE{f}(s')\}\bar{\NSE{Q}}_{\nu}^{N}(\dee s'|s, x)|\\
&\leq A+\delta E+\!\!\int_{\NSE{S}}\{B\max\{\|s\|,\|s'\|\}+\delta (B+D)\|s'\|\}\bar{\NSE{Q}}_{\nu}(\dee s'|s, x) \\
&\leq A+\delta E+ \big(B+\delta (B+D)\big)C+ (B+D)\|s\|.
\end{align*}
Thus, we have the desired result. 
\qed

\nt{\textit{\textbf{Proof of Lemma \ref{nssolnsint}}.}}
Pick some $(s, x)\in \NS{\NSE{S}}\times \NSE{X}$. 
As $\mathbb{V}\in \NSE{\mathcal{L}_{B,D}}[\NSE{S_N}]$, there exist some $E\in \NSE{\Reals}$ such that $|\mathbb{V}(t)|\leq E+(B+D)\|t\|$ for all $t\in \NSE{S_N}$. 
By \cref{VVapproxsigmaunbd}, $\NSE{V}(t)\approx \mathbb{V}(t)$ for all $t\in \NS{\NSE{S}}$. 
Hence, $E$ is near-standard. 
As $D\in \PosReals$, it is sufficient to show that $\|t\|$ is S-integrable with respect to $\bar{\NSE{Q}}_{\nu}(s, x)$. 
By \cref{closewasser} and \cref{pdint}:
\\
\begin{align*}
\int_{\NSE{S_N}}\|t\|\bar{\NSE{Q}}_{\nu}^{N}(\dee t|s, x)&\lessapprox \int_{\NSE{S}}\|t\|\bar{\NSE{Q}}_{\nu}(\dee t|s, x)\\
&=\lim_{n\to \infty}\ST\big(\int_{\NSE{S_n}}\|t\|\bar{\NSE{Q}}_{\nu}^{N}(\dee t|s, x)\big). 
\end{align*}
Note that $\int_{\NSE{S_N}}\|t\|\bar{\NSE{Q}}_{\nu}^{N}(\dee t|s, x)\gtrapprox \lim_{n\to \infty}\ST\big(\int_{\NSE{S_n}}\|t\|\bar{\NSE{Q}}_{\nu}^{N}(\dee t|s, x)\big)$. Hence, by Arkeryd et al. (1997, Section 4, Theorem 6.2), $\|t\|$ is S-integrable with respect to $\bar{\NSE{Q}}_{\nu}^{N}(s, x)$.
\qed

\medskip

\subsection{A Learning Foundation for Infinite Spaces}\label{sectioneqfd}In this subsection, we study the problem where the agent who faces a regular SMDP with compact state and action spaces, updates her belief in each period as a result of observing the current state, her action and the new state. Our aim is to show that the agent's steady state behavior is a Berk-Nash equilibrium. Throughout this section, we work with a regular SMDP $(\langle S,X,q_0,Q,\pi,\delta \rangle$, $\mathcal{Q}_{\Theta})$ as in \cref{regsmdp}. 
The agent who faces this regular SMDP has a prior $\mu_0\in \Delta(\Theta)$, which is assumed to have full support. 
Furthermore, throughout this section, we assume that the state space $S$ is compact.
We start with the following assumption:
\begin{assumption}\label{assumptionsctsmodel}
There is a referencing finite measure $\lambda$ on $(S, \BorelSets S)$ with full support such that 
\begin{enumerate}
    \item For all $\theta\in \Theta$ and $(s, x)\in S\times X$, $Q_{\theta}(s, x)$ is absolutely continuous with respect to $\lambda$;
    \item The density function $q_{(\theta, s, x)}(\cdot): S\to \Reals$ of $Q_{\theta}(s, x)$ with respect to $\lambda$ is a jointly continuous function on $\Theta\times S\times X\times S$;
    \item For all $(\theta,s,x)\in \Theta\times S\times X$, the density function $q_{(\theta, s, x)}(s')>0$ for all $s'\in S$. 
\end{enumerate}
\end{assumption}
Recall that $\Delta(\Theta)$ denote the set of probability measures on $(\Theta, \BorelSets \Theta)$, endowed with the Prokhorov metric. 
For $(s, x, s')\in S\times X\times S$, the Bayesian operator $B(s, x, s', \cdot): \Delta(\Theta)\to \Delta(\Theta)$ is defined as:
$
B(s, x, s', \mu)(A)=\frac{\int_{A}q_{(\theta,s,x)}(s')\mu(\dee \theta)}{\int_{\Theta}q_{(\theta,s,x)}(s')\mu(\dee \theta)}
$
for all $A\in \BorelSets \Theta$.

By the principle of optimality, the agent's problem can be cast recursively as:

\begin{equation}\label{blvbellman}
W(s, \mu)=\max_{x\in X}\int_{S}\{\pi(s,x,s')+\delta W(s', \mu')\}\bar{Q}_{\mu}(\dee s'|s, x)
\end{equation}
where $\bar{Q}_{\mu}(s, x)=\int_{\Theta}Q_{\theta}(s, x)\mu(\dee \theta)$ and $\mu'=B(s,x,s',\mu)$. 
Let $\cC[S\times \Delta(\Theta)]$ be the set of real-valued continuous functions on $S\times \Delta(\Theta)$, equipped with the $\sup$-norm. 
Assuming \cref{assumptionsctsmodel} holds, then the operator

\begin{equation}\label{beliefbellman}
L(g)(s, \mu)=\max_{x\in X}\int_{S}\{\pi(s,x,s')+\delta g(s', \mu')\}\bar{Q}_{\mu}(\dee s'|s, x)
\end{equation}
is a contraction mapping from $\cC[S\times \Delta(\Theta)]$ to $\cC[S\times \Delta(\Theta)]$, with the contraction factor $\delta$. 
Thus, by the Banach fixed point theorem, there exists a unique $W\in \cC[S\times \Delta(\Theta)]$ that is the solution of \cref{blvbellman}, which we fix for the rest of this section. 

\begin{definition}\label{defpolicy}
A policy function is a function $f: S\times \Delta(\Theta)\to \Delta(X)$, where $f(\cdot|s, \mu)$ is a probability measure on $X$ if she is in state $s$ and her belief is $\mu$. 
A policy function is optimal if, for all $s\in S$, $\mu\in \Delta(\Theta)$ and $x\in X$ such that $x$ is in the support $f(\cdot|s, \mu)$:
$
x\in \argmax_{\hat{x}\in X}\int_{S}\{\pi(s,\hat{x},s')+\delta W(s', B(s,\hat{x},s',\mu))\}\bar{Q}_{\mu}(\dee s'|s, \hat{x}).
$
\end{definition}

Let $h=(s_0,x_0,\dotsc,s_k,x_k,\dotsc)$ be an infinite history of state-action pairs and let $\mathbb{H}=(S\times X)^{\Nats}$ be the space of infinite histories. 
For every $k\in \Nats$, let $\mu_{k}:\mathbb{H}\to \Delta(\Theta)$ denote the agent's belief at time $k$, defined recursively by $\mu_{k}(h)=B(s_{k-1},x_{k-1},s_k,\mu_{k-1}(h))$.
When the context is clear, we drop $h$ from the notation. 

For a fixed $h\in \mathbb{H}$, in each period $k$, there is a state $s_{k}$ and a belief $\mu_k$. 
Given a policy function $f$, the agent chooses an action randomly according to $f(\cdot|s_k,\mu_k)$.
After an action $x_k$ is realized, the state $s_{k+1}$ is drawn according to the true transition probability $Q(\cdot|s_k, x_k)$. 
The agent then updates her belief to $\mu_{k+1}$ according to the Bayes operator. 
Thus, the primitives of the problem and the policy function $f$ induce a probability distribution $\mathbb{P}^{f}$ over $\mathbb{H}$. 

For every $k\in \Nats$, we define the frequency of the state-action pairs at time $k$ to be a function $m_{k}: \mathbb{H}\to \Delta(S\times X)$ such that $m_{k}(h)(A)=\frac{1}{k}\sum_{\tau=0}^{k}\mathbf{1}_{A}(s_{\tau}, x_{\tau})$
for all measurable $A\in \Delta(S\times X)$, where $\mathbf{1}_{A}$ denote the indicator function on $A$.

\begin{definition}\label{defunfcvg}
Let $H$ be a subset of $\mathbb{H}$.
The sequence $(m_k)_{k\in \Nats}$ is said to be uniformly converges to $m\in \Delta(S\times X)$ on $H$ in total variation distance if, 
for every $\epsilon>0$, there exists $k_{\epsilon}\in \Nats$ such that $\|m_k(h)-m\|_{\mathrm{TV}}<\epsilon$
for all $h\in H$ and all $k\geq k_{\epsilon}$. 
\end{definition}

\begin{remark}\label{unfcvgremark}
For $h\in \mathbb{H}$, the frequency of state-action pairs $m_k(h)$ is supported on a countable set. 
So, \cref{defunfcvg} implies that the support of $m$ is also countable. 
In conclusion, \cref{defunfcvg} is a reasonable assumption if both the state space $S$ and the action space $X$ are countable.  
\end{remark}

We now introduce the concept of identification and then present the main result Theorem \ref{blfupdatemain} of this section.
\begin{definition}\label{defident}
A SMDP is identified given $m\in \Delta(S\times X)$ if $\theta,\theta'\in \Theta_{Q}(m)$ implies $Q_{\theta}(\cdot|s, x)=Q_{\theta'}(\cdot|s, x)$ for all $(s, x)\in S\times X$.
\end{definition}

\begin{theorem}
\label{blfupdatemain}
Suppose \cref{assumptionsctsmodel} holds and the state space $S$ is compact. 
Let $f$ be an optimal policy function. 
Suppose: 
\begin{enumerate}
    \item $(m_{k})_{k\in \Nats}$ uniformly converges to some $m\in \Delta(S\times X)$ on some $H\in \mathbb{H}$ with $\mathbb{P}^{f}$-positive probability in total variation distance;
    \item The SMDP $(\langle S,X,q_0,Q,\pi,\delta \rangle$, $\mathcal{Q}_{\Theta})$ is identified given $m$.
\end{enumerate}
Then $m$ is a Berk-Nash equilibrium for the SMDP $(\langle S,X,q_0,Q,\pi,\delta \rangle$, $\mathcal{Q}_{\Theta})$.
\end{theorem}

\subsubsection{The Nonstandard Framework}

In this section, we present the nonstandard framework to prove \cref{blfupdatemain}.
Throughout this section, We work with a regular SMDP $(\langle S,X,q_0,Q,\pi,\delta \rangle, \mathcal{Q}_{\Theta})$ with a compact state space $S$. 

\begin{lemma}\label{bopects}
Suppose \cref{assumptionsctsmodel} holds.
Then the Bayesian operator $B$ is a continuous function from $S\times X\times S\times \Delta(\Theta)$ to $\Delta(\Theta)$. 
\end{lemma}

Since we are working with a regular SMDP, 
we can construct an associate hyperfinite SMDP $(\langle T_S, T_X, h_0, \mathbb{Q}, \Pi, \delta \rangle, \mathscr{Q}_{T_{\Theta}})$ as in \cref{sechypresent}, which will be fixed for the rest of the section. 
Let $\lambda$ be the finite measure on $S$ as in \cref{assumptionsctsmodel}. 
Define $\NSE{\lambda}_{T_S}$ to be the internal probability measure on $T_S$ such that $\NSE{\lambda}_{T_S}(\{s\})=\NSE{\lambda}(B_{S}(s))$ for all $s\in T_S$. 
Let $\NSE{\Delta}(T_{\Theta})$ denote the set of internal probability measures on $T_{\Theta}$. 
For $(s, x, s')\in T_S\times T_X\times T_S$, the hyperfinite Bayesian operator $\mathbb{B}(s, x, s', \cdot): \NSE{\Delta}(T_{\Theta})\to \NSE{\Delta}(T_{\Theta})$ is given by 
$
\mathbb{B}(s,x,s',\mu)(A)=\frac{\sum_{\theta\in A}\mathbb{Q}_{\theta}(s'|s,x)\mu(\{\theta\})}{\sum_{\theta\in T_{\Theta}}\mathbb{Q}_{\theta}(s'|s,x)\mu(\{\theta\})}
$
for all internal $A\subset T_{\Theta}$.

By the transfer principle and the principle of optimality, the agent's problem can be cast recursively as

\begin{equation}\label{hybeliefbellman}
\mathbb{W}(s,\mu)=\max_{x\in T_X}\sum_{s'\in T_S}\{\Pi(s,x,s')+\delta \mathbb{W}(s', \mu')\}
\bar{\mathbb{Q}}_{\mu}(s'|s, x)
\end{equation}
where $\bar{\mathbb{Q}}_{\mu}(s, x)=\sum_{\theta\in T_{\Theta}}\mathbb{Q}_{\theta}(s, x)\mu(\{\theta\})$, $\mu'=\mathbb{B}(s,x,s',\mu)$ and $\mathbb{W}: T_S\times \NSE{\Delta}(T_{\Theta})\to \NSE{\Reals}$ is the unique solution to the hyperfinite Bellman equation \cref{hybeliefbellman}. 
The existence of such a $\mathbb{W}$ is guaranteed by the transfer principle. 
The next theorem establishes a tight connection between the solution $\mathbb{W}$ of the hyperfinite Bellman equation \cref{hybeliefbellman} and the solution $W$ of the standard Bellman equation \cref{blvbellman}. 

\begin{theorem}\label{beliefBEclose}
Suppose \cref{assumptionsctsmodel} holds. 
For all $(s, \mu)\in T_S\times \NSE{\Delta}(T_{\Theta})$:
$
\mathbb{W}(s, \mu)\approx W(\ST(s), \pd{\mu}).
$
\end{theorem}

We now give the definition of hyperfinite policy functions. 

\begin{definition}\label{defhypolicy}
A hyperfinite policy function is an internal function $f: T_{S}\times \NSE{\Delta}(T_{\Theta})\to \NSE{\Delta}(T_X)$, where $f(x|s,\mu)$ denotes the probability that the agent chooses $x$ if she is in state $s$ and her belief is $\mu$. 
\end{definition}

We now discuss the agent's belief updating according to the hyperfinite SMDP $(\langle T_S, T_X, h_0, \mathbb{Q}, \Pi, \delta \rangle, \mathscr{Q}_{T_{\Theta}})$. 
The agent who faces the regular SMDP $(\langle S,X,q_0,Q,\pi,\delta \rangle$, $\mathcal{Q}_{\Theta})$ has a prior $\mu_0\in \Delta(\Theta)$, which is assumed to have full support. 
Let $\nu_0(\theta)=\NSE{\mu_0}(B_{\Theta}(\theta))$ for all $\theta\in T_{\Theta}$. 
As $\mu_0$ has full support, then $\nu_0(\theta)>0$ for all $\theta\in T_{\Theta}$.
The agent who faces the hyperfinite SMDP $(\langle T_S, T_X, h_0, \mathbb{Q}, \Pi, \delta \rangle, \mathscr{Q}_{T_{\Theta}})$ has the prior $\nu_0$. 
Let $h=(s_0,x_0,\dotsc,s_k,x_k,\dotsc)$ be an $\NSE{}$infinite hyperfinite history of state-action pairs and let $\NSE{\mathbb{H}}_{T_S\times T_X}=(T_S\times T_X)^{\NSE{\Nats}}$ be the space of infinite histories. 
It is clear that $\NSE{\mathbb{H}}_{T_S\times T_X}\subset \NSE{\mathbb{H}}=(\NSE{S}\times \NSE{X})^{\NSE{\Nats}}$. For two $h_1, h_2\in \NSE{\mathbb{H}}$, we write $h_1\approx h_2$ if every coordinates of $h_1$ and $h_2$ are infinitely close.  
For every $k\in \NSE{\Nats}$, let $\nu_{k}:\NSE{\mathbb{H}}_{T_S\times T_X}\to \NSE{\Delta}(T_\Theta)$ denote the agent's hyperfinite belief at time $k$, defined recursively by $\nu_{k}(h)=\mathbb{B}(s_{k-1},s_{k-1},s_k,\mu_{k-1}(h))$.
When the context is clear, we drop $h$ from the notation. 
Recall that we use $d_{P}$ to denote the Prokhorov metric on $\Delta(\Theta)$. 

For a fixed $h\in \NSE{\mathbb{H}}_{T_S\times T_X}$, in each period $k$, there is a state $s_{k}$ and a belief $\nu_k$. 
Given a hyperfinite policy function $F$, the agent chooses an action randomly according to $F(\cdot|s_k,\nu_k)$.
After an action $x_k$ is realized, the state $s_{k+1}$ is drawn according to the true hyperfinite transition probability $\mathbb{Q}(\cdot|s_k, x_k)$. 
The agent then updates her hyperfinite belief to $\nu_{k+1}$ according to the hyperfinite Bayes operator $\mathbb{B}$. 
Thus, the primitives of the problem and the hyperfinite policy function $F$ induce an internal probability measure $\NSE{\mathbb{P}}_{T_S\times T_X}^{F}$ over $\NSE{\mathbb{H}}_{T_S\times T_X}$. 

For every $k\in \NSE{\Nats}$, we define the hyperfinite frequency of the state-action pairs at time $k$ to be a 
function $M_{k}: \NSE{\mathbb{H}}_{T_S\times T_X}\to \NSE{\Delta}(T_S\times T_X)$ such that
$
M_{k}(h)(\{(s, x)\})=\frac{1}{k}\sum_{\tau=0}^{k}\mathbf{1}_{(s, x)}(s_{\tau}, x_{\tau}),
$
where $\mathbf{1}_{(s, x)}$ denote the indicator function on the point $(s, x)$. 

Recall that $\mathbb{K}_{\mathbb{Q}}: \NSE{\PM{T_S\times T_X}}\times T_{\Theta}\to \NSE{\NNReals}$ denote the hyperfinite weighted Kullback Leibler divergence, and the set of closest parameter values given $m\in \NSE{\PM{T_S\times T_X}}$ is the set
$
T_{\Theta}^{\mathbb{Q}}(m)=\argmin_{\theta\in T_{\Theta}}\mathbb{K}_{\mathbb{Q}}(m, \theta).
$
The set of almost closest parameter values given $m\in \NSE{\PM{T_S\times T_X}}$ is the external set
$
\hat{T}_{\Theta}^{\mathbb{Q}}(m)=\{\hat{\theta}\in T_{\Theta}: \mathbb{K}_{\mathbb{Q}}(m, \hat{\theta})\approx \min_{\theta\in T_{\Theta}}\mathbb{K}_{\mathbb{Q}}(m, \theta)\}. 
$
We now introduce the concept of S-identification for hyperfinite SMDP. 
\begin{definition}\label{defhyident}
The hyperfinite SMDP $(\langle T_S, T_X, h_0, \mathbb{Q}, \Pi, \delta \rangle, \mathscr{Q}_{T_{\Theta}})$ is S-identified given $m\in \NSE{\PM{T_S\times T_X}}$ if $\theta, \theta'\in \hat{T}_{\Theta}^{\mathbb{Q}}(m)$ implies that $\mathbb{Q}_{\theta}(\cdot|s, x)\approx \mathbb{Q}_{\theta'}(\cdot|s, x)$ for all $(s, x)\in T_S\times T_X$. 
\end{definition}

\subsubsection{Proof of \cref{blfupdatemain}}

In this section, we present a rigorous proof of \cref{blfupdatemain} via the hyperfinite SMDP constructed in the previous section.   
We start by proving the continuity of Bayesian operator.

\begin{proof}[\textbf{Proof of \cref{bopects}}]
Note that $S\times X\times S\times \Delta(\Theta)$ is a compact metric space. 
Pick $(s,x,s',\mu)\in \NSE{S}\times \NSE{X}\times \NSE{S}\times \NSE{\Delta(\NSE{\Theta})}$. 
Then, $\mu$ is an internal probability measure on $\NSE{\Theta}$. 
The standard part of $\mu$ in $\Delta(\Theta)$ with respect to the Prokhorov metric is simply the push-down of $\mu$, which we denote by $\pd{\mu}$. 
By \cref{assumptionsctsmodel} and \cref{pdint}, we have
$
\int_{\Theta}q_{(\theta,\ST(s),\ST(x))}(\ST(s'))\pd{\mu}(\dee \theta)\approx \int_{\NSE{\Theta}}\NSE{q}_{(\theta,s,x)}(s')\mu(\dee \theta).
$
Pick a set $A\in \BorelSets S$ such that $A$ is a continuity set of $B(\ST(s), \ST(x), \ST(s'), \pd{\mu})$. 
Then, by \cref{assumptionsctsmodel}, $A$ is a continuity set of $\pd{\mu}$, which implies that $\pd{\mu}(A)\approx \mu(\NSE{A})$. 
Hence, by \cref{assumptionsctsmodel} and \cref{pdint} again, we have
$
\int_{A}q_{(\theta,\ST(s),\ST(x))}(\ST(s'))\pd{\mu}(\dee \theta)\approx \int_{\NSE{A}}\NSE{q}_{(\theta,s,x)}(s')\mu(\dee \theta),
$
completing the proof. 
\end{proof}

The proof of \cref{beliefBEclose} relies on the following two lemmas:

\begin{lemma}\label{bdstclose}
Suppose \cref{assumptionsctsmodel} holds. 
For all $\theta\in T_{\Theta}$ and all $(s,x,s')\in T_S\times T_X\times T_S$:
$
q_{(\ST(\theta), \ST(s), \ST(x))}(\ST(s'))\approx \frac{\mathbb{Q}_{\theta}(s'|s, x)}{\NSE{\lambda}_{T_S}(\{s'\})}.
$
\end{lemma}
\begin{proof}
Pick $\theta\in T_{\Theta}$ and $(s,x,s')\in T_S\times T_X\times T_S$. 
Note that $\lambda$ has full support. 
By the construction of $T_S$,  we know that $\NSE{\lambda}_{T_S}(\{s'\})>0$. 
By the transfer principle, we have 
$
\mathbb{Q}_{\theta}(s'|s, x)=\NSE{Q}_{\theta}(s, x)(B_{S}(s'))=\int_{B_{S}(s')}\NSE{q}_{(\theta,s,x)}(y)\NSE{\lambda}(\dee y).
$
We also have 
$
\mathbb{Q}_{\theta}(s'|s, x)=\int_{B_{S}(s')}\frac{\mathbb{Q}_{\theta}(s'|s, x)}{\NSE{\lambda}_{T_S}(\{s'\})}\NSE{\lambda}(\dee y). 
$
By \cref{assumptionsctsmodel}, we have $\NSE{q}_{(\theta,s,x)}(y)\approx \frac{\mathbb{Q}_{\theta}(s'|s, x)}{\NSE{\lambda}_{T_S}(\{s'\})}$ for all $y\in B_{S}(s')$.
By \cref{assumptionsctsmodel}, we have $q_{(\ST(\theta), \ST(s), \ST(x))}(\ST(s'))\approx \frac{\mathbb{Q}_{\theta}(s'|s, x)}{\NSE{\lambda}_{T_S}(\{s'\})}$.
\end{proof}

Let $d_P$ denote the Prokhorov metric on $\Delta(\Theta)$. 
By the transfer principle, $\NSE{d}_P$ is the $\NSE{}$Prokhorov metric on $\NSE{\PM{\NSE{\Theta}}}$. 

\begin{lemma}\label{boptclose}
Suppose \cref{assumptionsctsmodel} holds. 
For all $\mu\in \NSE{\Delta}(T_{\Theta})$ and all $(s,x,s')\in T_S\times T_X\times T_S$:
$
\NSE{d}_{P}\big(\NSE{B}(\ST(s), \ST(x), \ST(s'), \pd{\mu}), \mathbb{B}(s,x,s',\mu)\big)\approx 0.
$
That is, the hyperfinite Bayesian operator $\mathbb{B}(s,x,s',\mu)$ is in the monad of the standard Bayesian operator $B(\ST(s), \ST(x), \ST(s'), \pd{\mu})$, with respect to the Prokhorov metric $d_{P}$. 
\end{lemma}
\begin{proof}
For all $\mu\in \NSE{\Delta}(T_{\Theta})$, as $\Theta$ is compact, $\pd{\mu}$ is a well-defined probability measure on $\Theta$, and $\mu$ is in the monad of $\pd{\mu}$ with respect to the Prokhorov metric. Then the result follows from \cref{bdstclose} and \cref{pdint}. 
\end{proof}

We now give a rigorous proof of \cref{beliefBEclose}

\begin{proof}[\textbf{Proof of \cref{beliefBEclose}}]
Let $W_0$ be the restriction of $\NSE{W}$ on $T_S\times \NSE{\Delta}(T_{\Theta})$. 
For all $(\mu, s, x)\in \NSE{\Delta}(T_{\Theta})\times T_S\times T_X$, by \cref{Qbarlemma}, \cref{bopects} and \cref{pdint}, we have
\begin{align*}
&\int_{T_S}\{\Pi(s,x,s')+\delta W_0(s', \mu')\}\bar{\mathbb{Q}}_{\mu}(\dee s'|s, x)\\
&\approx \int_{S}\{\pi(\ST(s),\ST(x),s')+\delta W(s', \pd{\mu}')\}\bar{Q}_{\pd{\mu}}(\dee s'|\ST(s), \ST(x)).
\end{align*}
Thus, we can conclude that
\begin{align*}
&\max_{x\in T_X}\int_{T_S}\{\Pi(s,x,s')+\delta W_0(s', \mu')\}\bar{\mathbb{Q}}_{\mu}(\dee s'|s, x)\\
&\approx \max_{x\in X}\int_{S}\{\pi(\ST(s),x,s')+\delta W(s', \pd{\mu}')\}\bar{Q}_{\pd{\mu}}(\dee s'|\ST(s), x)\\
&=W(\ST(s),\pd{\mu})\approx W_0(s, \mu).
\end{align*}
Let 
$
\mathbb{L}(g)(s)=\max_{x\in T_X}\int_{T_S}\{\Pi(s,x,s')+\delta g(s',\mu')\}\bar{\mathbb{Q}}_{\mu}(\dee s'|s, x)
$
for all internal function $g: T_S\times \NSE{\Delta}(T_{\Theta})\to \NSE{\Reals}$. 
Note that $\mathbb{L}$ is a contraction with the contraction factor $\delta$.
Moreover, we can find $\mathbb{W}$ as following: start with $W_0$ and define a sequence $\{W_n\}_{n\in \NSE{\Nats}}$ by $W_{n+1}=\mathbb{H}(W_n)$. 
Then $\mathbb{W}$ is the $\NSE{}$limit of $\{W_n\}_{n\in \NSE{\Nats}}$. 
Thus, we have
$
\NSE{d_{\sup}}(W_0, \mathbb{W})\leq \frac{1}{1-\delta}\NSE{d_{\sup}}(W_1, W_0)\approx 0. 
$
As $W$ is continuous, we have $\mathbb{W}(s, \mu)\approx W(\ST(s), \pd{\mu})$ for all $(s, \mu)\in T_S\times \NSE{\Delta}(T_{\Theta})$. 
\end{proof}
We now prove two important consequences of \cref{beliefBEclose}, which will be used in the proof of \cref{blfupdatemain}.
Let $Y$ be an arbitrary metric space and $U$ be a subset of $\NSE{Y}$. 
The \textit{nonstandard hull} of $U$, denoted by $\hat{U}$, is the collection of all points in $\NSE{Y}$ that are infinitely close to some point in $U$. 
That is:
$
\hat{U}=\{y\in \NSE{Y}: (\exists u\in U )(\NSE{d}_{Y}(y, u)\approx 0)\}. 
$
\begin{lemma}\label{nshulllemma}
Suppose  \cref{assumptionsctsmodel} holds. 
Let $(s, \mu)\in T_S\times \NSE{\PM{T_{\Theta}}}$. 
Suppose
$
y\in \hat{\argmax_{x\in T_X}}\int_{T_S}\{\Pi(s,x,s')+\delta \mathbb{W}(s', \mathbb{B}(s,x,s',\mu))\}\bar{\mathbb{Q}}_{\mu}(\dee s'|s, x).
$
Then
\begin{align*}
&\int_{T_S}\{\Pi(s,y,s')+\delta \mathbb{W}(s', \mathbb{B}(s,y,s',\mu))\}\bar{\mathbb{Q}}_{\mu}(\dee s'|s, y)\\
&\approx \max_{x\in T_X}\int_{T_S}\{\Pi(s,x,s')+\delta \mathbb{W}(s', \mathbb{B}(s,x,s',\mu))\}\bar{\mathbb{Q}}_{\mu}(\dee s'|t, x).
\end{align*}
\end{lemma}
\begin{proof}
Pick 
$
x_0\in \argmax_{x\in T_X}\int_{T_S}\{\Pi(s,x,s')+\delta \mathbb{W}(s', \mathbb{B}(s,x,s',\mu))\}\bar{\mathbb{Q}}_{\mu}(\dee s'|s, x)
$
such that $y\approx x_0$. 
As the SMDP $(\langle S,X,q_0,Q,\pi,\delta \rangle$, $\mathcal{Q}_{\Theta})$ is regular, we have $\Pi(s, y, s')\approx \Pi(s, x_0, s')$ for all $s'\in T_S$. By regularity again, the $\NSE{}$Prokhorov distance between $\bar{\mathbb{Q}}(\cdot|s, y)$ and $\bar{\mathbb{Q}}(\cdot|s, x_0)$ is infinitesimal. 
The result then follows from \cref{boptclose} and \cref{beliefBEclose}. 
\end{proof}

\begin{lemma}\label{nshullstar}
Suppose \cref{assumptionsctsmodel} holds.
Let $(s_1, \mu)\in \NSE{S}\times \NSE{\PM{\NSE{\Theta}}}$ and $(s_2, \nu)\in T_S\times \NSE{\PM{T_{\Theta}}}$ such that $s_1\approx s_2$ and the $\NSE{}$Prokhorov distance between $\mu$ and $\nu$ is infinitesimal. 
Suppose 
$
x\in \argmax_{\hat{x}\in \NSE{X}}\int_{\NSE{S}}\{\NSE{\pi}(s_1,\hat{x},s')+\delta \NSE{W}(s', \NSE{B}(s_1,\hat{x},s',\mu))\}\NSE{\bar{Q}}_{\mu}(\dee s'|s_1, \hat{x}).
$
Then, for all $y\in T_X$ such that $y\approx x$:

\begin{align*}
&\int_{T_S}\{\Pi(s_2,y,s')+\delta \mathbb{W}(s',\mathbb{B}(s_2,y,s',\nu))\}\bar{\mathbb{Q}}_{\nu}(\dee s'|s_2, y)\\
&\approx \max_{\hat{x}\in T_X}\int_{T_S}\{\Pi(s_2,\hat{x},s')+\delta \mathbb{W}(s', \mathbb{B}(s_2,\hat{x},s',\nu))\}\bar{\mathbb{Q}}_{\nu}(\dee s'|s_2, \hat{x}).
\end{align*}
\end{lemma}
\begin{proof}
By \cref{boptclose}, \cref{beliefBEclose} and the fact that $(\langle S,X,q_0,Q,\pi,\delta \rangle$, $\mathcal{Q}_{\Theta})$ is a regular SMDP, we have
\begin{align*}
&\argmax_{\hat{x}\in \NSE{X}}\int_{\NSE{S}}\{\NSE{\pi}(s_1,\hat{x},s')+\delta \NSE{W}(s', \NSE{B}(s_1,\hat{x},s',\mu))\}\NSE{\bar{Q}}_{\mu}(\dee s'|s_1, \hat{x})\\
&\approx \argmax_{\hat{x}\in T_X} \int_{T_S}\{\Pi(s_2,\hat{x},s')+\delta \mathbb{W}(s', \mathbb{B}(s_2,\hat{x},s',\nu))\}\bar{\mathbb{Q}}_{\nu}(\dee s'|s_2, \hat{x}). \\
\end{align*}
Moreover, we have 
\begin{align*}
&y\in \hat{\argmax_{z\in \NSE{X}}}\int_{\NSE{S}}\{\NSE{\pi}(s_1,z,s')+\delta \NSE{W}(s', \NSE{B}(s_1,z,s',\mu))\}\NSE{\bar{Q}}_{\mu}(\dee s'|s_1, z)\\
&=\hat{\argmax_{z\in T_X}} \int_{T_S}\{\Pi(s_2,z,s')+\delta \mathbb{W}(s', \mathbb{B}(s_2,z,s',\nu))\}\bar{\mathbb{Q}}_{\nu}(\dee s'|s_2, z).\\
\end{align*}
By \cref{nshulllemma}, we have the desired result. 
\end{proof}
We are now at the place to prove \cref{blfupdatemain}. 
We start with the following lemma, which shows that the agent's belief $\mu_k$ and the agent's hyperfinite belief $\nu_k$ remains close for some infinite steps. 
\begin{lemma}\label{updateclose}
Suppose \cref{assumptionsctsmodel} holds.
Let $\tilde{h}\in \NSE{\mathbb{H}}_{T_S\times T_X}$ and $h\in \NSE{\mathbb{H}}$ be such that $\tilde{h}\approx h$. 
Then, for every $k\in \Nats$, $\NSE{d}_{P}(\NSE{\mu}_{k}, \nu_k)\approx 0$. 
Hence, there exists some $k_0\in \NSE{\Nats}\setminus \Nats$ such that $\NSE{d}_{P}(\NSE{\mu}_{k}, \nu_k)\approx 0$ for all $k\leq k_0$. 
\end{lemma}
\begin{proof}
The second claim follows from the first claim and saturation. 
We now prove the first claim by induction. 
Clearly, we have $\NSE{d}_{P}(\NSE{\mu}_{0}, \nu_0)\approx 0$. 
The inductive case follows from \cref{bopects} and \cref{boptclose}. 
\end{proof}

If the frequency of state-action pairs $(m_k)_{k\in \Nats}$ uniformly converges in total variation distance to some $m\in \Delta(S\times X)$ for all $h$ in some set $H\subset \mathbb{H}$, then the hyperfinite frequence of state-action pairs $(M_k)_{k\in \NSE{\Nats}}$ almost converges to some $M\in \NSE{\Delta}(T_S\times T_X)$ for all $\tilde{h}$ in some internal $\tilde{H}\subset \NSE{\mathbb{H}}_{T_S\times T_X}$. 
As one would expect, $M$ and $\tilde{H}$ are closely related to $m$ and $H$, respectively. 

\begin{lemma}\label{mtcvglemma}
Let $H\in \mathbb{H}$ be such that $(m_{k}(h))_{k\in \Nats}$ converges in total variation distance to some $m\in \Delta(S\times X)$ for all $h\in H$.
Let $M\in \NSE{\Delta}(T_S\times T_X)$ be $M(\{(s, x)\})=\NSE{m}(B_{S}(s)\times B_{X}(x))$. 
Let $\tilde{H}$ be the internal subset of $\NSE{\mathbb{H}}_{T_S\times T_X}$ consisting of $\tilde{h}=(\tilde{s}_0, \tilde{x}_0,\dotsc, \tilde{s}_k, \tilde{x}_k, \dotsc)\in \tilde{H}$ such that
$
(\exists h=(s_0,x_0,\dotsc, s_k, x_k, \dotsc)\in \NSE{H})(\forall k\in \NSE{\Nats})(s_k\in B_{S}(\tilde{s}_k)\wedge x_k\in B_{X}(\tilde{x}_k)).
$
Then $(M_{k}(\tilde{h}))_{k\in \NSE{\Nats}}$ $\NSE{}$converges to $M$ for all $\tilde{h}\in \tilde{H}$. 
Moreover, if $(m_k)_{k\in \Nats}$ uniformly converges to $m$ on $H$ in total variation distance, then 
$
\|M_{k}(\tilde{h})-M\|_{\mathrm{TV}}\approx 0
$
for all $\tilde{h}\in \tilde{H}$ and all $k\in \NSE{\Nats}\setminus \Nats$. 
\end{lemma}
\begin{proof}
Pick some $\tilde{h}=(\tilde{s}_0, \tilde{x}_0,\dotsc, \tilde{s}_k, \tilde{x}_k, \dotsc)\in \tilde{H}$. 
By the construction of $\tilde{H}$, there exists some $h=(s_0,x_0,\dotsc, s_k, x_k, \dotsc)\in \NSE{H}$ such that $s_k\in B_{S}(\tilde{s}_k)$ and $x_k\in B_{X}(\tilde{x}_k)$ for all $k\in \NSE{\Nats}$. 
By the transfer principle,
$
\big(\NSE{m}_{k}(h)(B_{S}(s)\times B_{X}(x))\big)_{k\in \NSE{\Nats}}\ \text{$\NSE{}$converges to}\ \NSE{m}(B_{S}(s)\times B_{X}(x))
$
for all $(s, x)\in T_S\times T_X$. For all $(s, x)\in T_S\times T_X$, note that $(\tilde{s}_k, \tilde{x}_k)=(s, x)$ if and only if $(s_k, x_k)\in B_{S}(s)\times B_{X}(x)$. 
Hence, we conclude that $\big(M_{k}(\tilde{h})(\{(s, x)\})\big)_{k\in \NSE{\Nats}}$ $\NSE{}$converges to $M(\{(s, x)\})$ for all $(s, x)\in T_S\times T_X$. 
Now, suppose that $(m_k)_{k\in \Nats}$ uniformly converges to $m$ on $H$ in total variation distance, by saturation, we have $\|\NSE{m}_{k}(h)-\NSE{m}\|_{\mathrm{TV}} \approx 0$
for all $h\in \NSE{H}$ and all $k\in \NSE{\Nats}\setminus \Nats$. 
By the construction of $\tilde{H}$, $T_S$ and $T_X$, we have the desired result. 
\end{proof}

Using essentially the same proof as in Lemma 2 of EP, we have:

\begin{lemma}\label{inacpvset}
Let $F$ be a hyperfinite policy function. 
Suppose that $\|M_k(\tilde{h})-M\|_{\mathrm{TV}}\approx 0$
for some $k\in \NSE{\Nats}$ and all $\tilde{h}$ in some internal $\tilde{H}\subset \NSE{\mathbb{H}}_{T_S\times T_X}$ such that $\NSE{\mathbb{P}}_{T_S\times T_X}^{F}(\tilde{H})>0$. 
Then, for any internal set $A\supset \hat{T}_{\Theta}^{\mathbb{Q}}(m)$, $\nu_k(A)\approx 1$ on some internal set $\tilde{H}'\subset \tilde{H}$ with $\NSE{\mathbb{P}}_{T_S\times T_X}^{F}(\tilde{H}')>0$.
\end{lemma}

We now study the connection between identification and S-identification.
The following lemma follows from essentially the same proof of \cref{beliefmain}.

\begin{lemma}\label{almstparamet}
Let $m\in \Delta(S\times X)$ and let $M\in \NSE{\PM{T_S\times T_X}}$ be the same as in \cref{mtcvglemma}. 
Then, for every $\hat{\theta}\in \hat{T}_{\Theta}^{\mathbb{Q}}(M)$, $\ST(\hat{\theta})\in \Theta_{Q}(m)$. 
\end{lemma}

The hyperfinite SMDP is S-identified if the SMDP is identified. 

\begin{lemma}\label{idthyidt}
Suppose \cref{assumptionsctsmodel} holds and the SMDP $(\langle S,X,q_0,Q,\pi,\delta \rangle$, $\mathcal{Q}_{\Theta})$ is identified given $m\in \Delta(S\times X)$. 
Let $M\in \NSE{\PM{T_S\times T_X}}$ be the same as in \cref{mtcvglemma}.
Then the hyperfinite SMDP $(\langle T_S, T_X, h_0, \mathbb{Q}, \Pi, \delta \rangle, \mathscr{Q}_{T_{\Theta}})$ is S-identified given $M$. 
\end{lemma}
\begin{proof}
By \cref{almstparamet}, 
for every $\theta\in \hat{T}_{\Theta}^{\mathbb{Q}}(m)$, we have $\ST(\theta)\in \Theta_{Q}(m)$. 
Thus, for $\theta, \theta'\in \hat{T}_{\Theta}^{\mathbb{Q}}(m)$, by \cref{assumptionsctsmodel}, we have 
$
\NSE{Q}_{\theta}(\cdot|s, x)\approx \NSE{Q}_{\ST(\theta)}(\cdot|s, x)=\NSE{Q}_{\ST(\theta')}(\cdot|s, x)\approx \NSE{Q}_{\theta'}(\cdot|s, x)
$
for all $(s, x)\in \NSE{S}\times \NSE{X}$. 
This immediately implies that $\mathbb{Q}_{\theta}(\cdot|s, x)\approx \mathbb{Q}_{\theta'}(\cdot|s, x)$
for all $(s, x)\in T_S\times T_X$, completing the proof. 
\end{proof}

We now prove the main result, \cref{blfupdatemain}. 

\begin{proof}[\textbf{Proof of \cref{blfupdatemain}}]
Let $M\in \NSE{\PM{T_S\times T_X}}$ be the same as in \cref{mtcvglemma}.
Let $\tilde{H}\subset \NSE{\mathbb{H}}_{T_S\times T_X}$ be the same internal set as in \cref{mtcvglemma}. 
By \cref{mtcvglemma}, we have:
$
\|M_{k}(\tilde{h})-M\|_{\mathrm{TV}}\approx 0
$
for all $\tilde{h}\in \tilde{H}$ and all $k\in \NSE{\Nats}\setminus \Nats$.
Let $F$ be a hyperfinite policy function such that $\NSE{\mathbb{P}}_{T_S\times T_X}^{F}(\tilde{H})>0$. 
Pick $\tilde{H}'\subset \tilde{H}$ as in \cref{inacpvset}. 
For the rest of the proof, we fix $\tilde{h}=(\tilde{s}_0, \tilde{x}_0,\dotsc,\tilde{s}_k, \tilde{x}_k,\dotsc)\in \tilde{H}'$. 
By the construction of $\tilde{H}$, there exists $h=(s_0, x_0, \dotsc, s_k, x_k, \dotsc)\in \NSE{H}$ such that
$s_k\in B_{S}(\tilde{s}_k)$ and $x_k\in B_{X}(\tilde{x}_k)$ for all $k\in \NSE{\Nats}$. 
For every $k\in \NSE{\Nats}$, let $\NSE{\mu}_k$ denote the updated $\NSE{}$belief and $\nu_k$ denote the updated hyperfinite belief at time $k$, according to $h$ and $\tilde{h}$, respectively.  
Henceforth, we omit the hyperfinite history from the notation. 

Recall that we use $d_{P}$ to denote the Prokhorov metric on $\Delta(\Theta)$. 
By \cref{updateclose}, there exists some $k_0\in \NSE{\Nats}\setminus \Nats$ such that $\NSE{d}_{P}(\NSE{\mu}_k, \nu_k)\approx 0$ for all $k\leq k_0$. 
Let $(s, x)\in T_S\times T_X$ be such that $M(\{(s, x)\})>0$. 
By the construction of $M$, $(\ST(s), \ST(x))$ is in the support of $m$. 
Thus, there exists $k_1\in \NSE{\Nats}\setminus \Nats$ such that 
\begin{enumerate}
    \item $k_1\leq k_0$;
    \item $\tilde{s}_{k_1}\approx s_{k_1}\approx s$ and $\tilde{x}_{k_1}\approx x_{k_1}\approx x$. 
\end{enumerate}
As $f$ is an optimal policy function, by the transfer principle, we have
\begin{align*}
x_{k_1}\in \argmax_{\hat{x}\in \NSE{X}}\int_{\NSE{S}}\{\NSE{\pi}(s_{k_1}, \hat{x}, s')+\delta \NSE{W}(s', \NSE{B}(s_{k_1}, \hat{x}, s', \NSE{\mu}_{k_1}))\}\bar{\NSE{Q}}_{\NSE{\mu}_{k_1}}(\dee s'|s_{k_1}, \hat{x}).
\end{align*}
\nt As $\NSE{d}_{P}(\NSE{\mu}_{k_1}, \nu_{k_1})\approx 0$, by \cref{nshullstar}, we have:
\begin{align*}
&\int_{T_S}\{\Pi(s,x,s')+\delta \mathbb{W}(s',\mathbb{B}(s,x,s',\nu_{k_1}))\}\bar{\mathbb{Q}}_{\nu_{k_1}}(\dee s'|s, x)\\
&\approx \max_{\hat{x}\in T_X}\int_{T_S}\{\Pi(s,\hat{x},s')+\delta \mathbb{W}(s', \mathbb{B}(s,\hat{x},s',\nu_{k_1}))\}\bar{\mathbb{Q}}_{\nu_{k_1}}(\dee s'|s, \hat{x}).
\end{align*}
As the SMDP is identified given $m$, 
by \cref{idthyidt}, the hyperfinite SMDP is S-identified. 
This implies that there is $\mathbb{Q}_{M}$ such that, for all $\nu\in \NSE{\PM{T_{\Theta}}}$ with support being a subset of $\hat{T}_{\Theta}^{\mathbb{Q}}(M)$, $\mathbb{Q}_{\nu}\approx \mathbb{Q}_{M}$. 
By \cref{inacpvset}, the support of $\nu_{k_1}$ is a subset of $\hat{T}_{\Theta}^{\mathbb{Q}}(M)$. 
Note that the support of the posterior of $\nu_{k_1}$ generated from the hyperfinite Bayesian operator is a subset of the support of $\nu_{k_1}$. 
Hence, we have:
\begin{align*}
&\int_{T_S}\{\Pi(s,x,s')+\delta \mathbb{V}(s')\}\mathbb{Q}_{M}(\dee s'|s, x)\\
&\approx \max_{\hat{x}\in T_X}\int_{T_S}\{\Pi(s,\hat{x},s')+\delta \mathbb{V}(s')\}\mathbb{Q}_{M}(\dee s'|s, \hat{x}).
\end{align*}
By \cref{mtcvglemma}, $(M_k)_{k\in \NSE{\Nats}}$ $\NSE{}$converges to $M$ for all hyperfinite histories in $\tilde{H}$. 
By the transfer principle (or use essentially the same proof as in Theorem 2 of EP, $M$ is $\NSE{}$stationary. 
Thus, $M$ is a Berk-Nash S-equilibrium, as in \cref{sberknote}, for the hyperfinite SMDP with the hyperfinite belief $\nu_{k_1}$ (or any $\nu\in \NSE{\PM{T_{\Theta}}}$ such that the support of $\nu$ is a subset of $\hat{T}_{\Theta}^{\mathbb{Q}}(M)$). 
Hence, by \cref{sberknote}, $m$ is a Berk-Nash equilibrium for the SMDP with the belief being any $\mu\in \Delta(\Theta)$ such that the support of $\mu$ is a subset of $\Theta_{Q}(m)$.  
\end{proof}

\subsection{Detailed Analysis of Examples}\label{detailexample}

Here we present the complete analysis of the examples in the main body of the paper. To recap, Example 1 deals with the optimal consumption-savings environment while Example 2 is about a  producer with misspecified costs. Examples 3 and 4 study a misspecified AR(1) process  with Example 5 dealing with a misspecified revenue problem for a producer.

\nt  \textbf{Example \ref{stochasticgrowth}} For this optimal savings problem, we solve for optimality, belief restriction and stationarity. The Bellman equation for the agent is

$$
V(y, z)=\max _{0 \leq x \leq y} z \ln (y-x)+\delta E\left[V\left(y^{\prime}, z^{\prime}\right) \mid x\right],
$$
and let us guess that the form of the value function is $V(y,z)=a(z)+b(z)\ln (y).$ This provides us a guess for the optimal strategy which is to invest a fraction of wealth that depends on the utility shock and the unknown parameter $\beta$, i.e., $x=A_{z}(\beta) \cdot y$, where $A_{z}(\beta)=\displaystyle\frac{\delta \beta E[b(z')]}{(z+\delta \beta E[b(z')])}$ where $b(z)$ satisfies $b(z)=z+\delta \beta E[b(z')]$ for $z\in [0,1].$ Solving for $b(z),$ we get $b(z)=z+\dfrac{\delta\beta}{1-\delta\beta} E[z]$ which gives $A_{z}(\beta)=\dfrac{0.5\delta\beta}{(1-\delta\beta)z+0.5\delta\beta}$ where $E[z]=0.5.$\fn{This corrects a typo in EP for the policy function.} The stationarity condition is met because of $0\leq\beta^{*}<1,$ which prevents the process from drifting away. The belief restriction and the rest of the problem for $\beta^{m}$ is solved analogously as in EP. 

\nt \textbf{Example \ref{costsproduction}} This example assumes that the agent knows the per-period payoff function  and the transition function but has a misspecified cost function. We follow EP in framing cost be a part of the state variable. we simply let the cost $c$ be part of the state as follows:
$$
V(z, c)=\max _{x} \int_{Z\times C}\left(z f(x)-c^{\prime}+\delta V\left(z^{\prime}, c^{\prime}\right)\right) Q\left(\dee z^{\prime} \mid z\right) Q^{C}\left(\dee c^{\prime} \mid x\right)
$$
The variable $c^{\prime}$ is the unknown cost of production at the time the agent has to choose $x$. Its distribution is given by $Q^{C}\left(\dee c^{\prime} \mid x\right)$, which is the distribution of $c^{\prime}=c(x)$ as described above. The agent knows $Q$, but does not know $Q^{C}$. In particular, the agent has a parametric family of transitions, where $Q_{\theta}^{C}\left(\dee c^{\prime} \mid x\right)$ is the distribution of $c^{\prime}=c_{\theta}(x)$. The action space, $X=\Bigg[\varepsilon, \max\Bigg[\Bigg((\expect_{d^{*}}[\epsilon])/4\Bigg)^{2/3}, \Bigg((\expect_{d^{*}}[\epsilon])/\sqrt{K}\Bigg)^{2/3}\Bigg]+1\Bigg], \varepsilon>0.$ The parameter space $\Theta$ is compact,  $\Theta=\Bigg[0,\sqrt{\dfrac{K+1}{2}}\expect_{d^{*}}[\epsilon]+1\Bigg]$, where $K=\dfrac{1-e^{-k}}{1-k(k+1)e^{-k}}.$
Given this, suppose the true cost function $\phi(x)$ is quadratic i.e. $\phi(x)=x^{2}.$ Then the Berk-Nash equilibrium is characterized by the minimizer $\theta^{*}$ given by
$\theta^{*}=\sqrt{\dfrac{K}{2}\expect_{d^{*}}[\epsilon]}$ and the action $x^{*}=\displaystyle\sqrt{\frac{2K}{\expect_{d^{*}}(\epsilon)}}z$, whereas for the agent with the correctly specified model is $x^{opt}=\displaystyle\frac{\sqrt z}{\sqrt{2\expect_{d^{*}}(\epsilon)}}$. Indeed, note that the true transition probability function $Q(s)$ has a unique stationary measure $\mu$. Therefore, the Berk-Nash equilibrium for this SMDP is $\mu\times \delta_{x^{*}}$, supported by the belief $\delta_{\theta^{*}}$.

We now solve for the equilibrium. First, we solve for the optimal $x*$ as a function of the parameter. Suppose the agent has a degenerate belief on some $\theta.$ Here, as in the original example, the agent's optimization problem reduces to a static optimization problem $\max _{x} z \ln x-x E_{\theta}[\epsilon]$. Noting that $E_{\theta}[\epsilon]=\dfrac{\theta}{K}$, it follows that the optimal input choice in state $z$ is $
x^{*}=K z / \theta.$ Next, the stationarity condition implies that the marginal of $m$ over $\mathbb{Z}$ is equal to the stationary distribution over $z$, which is $q$, a uniform distribution, $U [0,1]$. Therefore, the stationary distribution over $\mathbb{X}$, denoted by $m_{\mathrm{X}}$, is a uniform distribution, $U[0,\frac{K}{\theta}].$ Finally, following the steps as in Example 3.3, we get our corresponding $x^{*}$ and $\theta^{*}.$

\medskip

\nt \textbf{Example \ref{example-Unit_Root}}\ In this example, we study an AR(1) process and show that a Berk-Nash equilibrium exists if and only if the AR(1) process has no unit root. Recall that the SMDP in this problem is defined as:
\begin{itemize}
    \item The state space $S=\Reals$, the action space $X=\{0\}$, and the payoff function $\pi: S\times X\times S\to \Reals$ is the constant function $0$;
    \item For every $s\in S$, the true transition probability function $Q(s)$ is the distribution of $a_0s+b_0\xi$, where $a_0\in [0, 2], b_0\in [0, 1]$ and $\xi=\mathcal{N}(0,1)$ has the standard normal distribution;
    \item he parameter space $\Theta$ is $[0, 2]\times [0, 1]$ and for every $(a, b)\in \Theta$, the transition probability function $Q_{(a, b)}(s)$ is the distribution of $as+b\xi$.
\end{itemize}
We first consider the degenerate case $b_0=0$.
The true transition $Q_{(a_0,0)}$ is absolutely continuous with respect to $Q_{(a,b)}$ if and only if $a=a_0$ and $b=b_0 = 0$. When $a_0<1$, the Markov process has a unique stationary distribution, namely the Dirac measure $\delta_0$ at zero. So the Berk-Nash equilibrium is $\delta_{(0,0)}$ with the belief $\delta_{(a_0, 0)}$. 
When $a_0 = 1$, the Dirac measure $\delta_s$ is a stationary distribution for every $s \in S$, and $\delta_{(s,0)}$ is a Berk-Nash equilibrium supported by the belief $\delta_{(1,0)}$. 
When $a_0>1$, there is no stationary distribution hence no Berk-Nash equilibrium. 

For the non-degenerate case $b_0>0$, following Example \ref{example-Unit_Root}, we focus on the case $0\leq a_0<1$.
We now provide rigorous verification for \cref{assumptiontight} and \cref{assumptionrebound}:
\begin{itemize}
    \item We apply the Lyapunov condition to verify \cref{assumptiontight} by taking the Lyapunov function $V(s)=|s|$. 
          Clearly, this $V$ is a non-negative, continuous and norm-like function as defined in \cref{assumptiontight}. 
          Moreover, we have $S_n=\{s\in S: V(s)\leq n\}$ for all $n\in \Nats$.  
          By the properties of the folded normal distribution, we have: 
         $
         \int_{S}|s'|Q(s)(\dee s')=b_0\sqrt{\frac{2}{\pi}}e^{-\frac{a_{0}^2s^2}{2b_{0}^2}}+a_0s(1-2\phi(-\frac{a_0s}{b_0}))
         $
         where $\phi$ is the cumulative distribution function of the standard normal distribution. 
         Thus, for all $s\geq 0$, we have $\int_{S}|s'|Q(s)(\dee s')\leq a_0|s|+\sqrt{\frac{2}{\pi}}$.
         Hence, by choosing $\alpha=1-a_0$ and $\beta=\sqrt{\frac{2}{\pi}}$, \cref{lyapueq} is satisfied. Hence, \cref{assumptiontight} is satisfied;
    \item For $(a, b)\in [0, 2]\times (0, 1]$ and $s\in S$, the relative entropy from $Q_{(a, b)}(s)$ to $Q(s)$ is:
          $
          \mathcal{D}_{\mathrm{KL}}(Q(s), Q_{(a, b)}(s))=\ln(\frac{b}{b_0})+\frac{b_0^2+(a_0s-as)^2}{2b^2}-\frac{1}{2}.
           $
           Note that the true transition probability function $Q(s)$ has a unique stationary measure $\mu=\mathcal{N}(0, \frac{b_0^2}{1-a_0^2})$. It is then straightforward to show that \cref{assumptionrebound} is satisfied.
\end{itemize}

\medskip

\nt \textbf{Example \ref{pricingshocks}}\ This example assumes that the agent knows the per-period payoff function  and the transition function but has a misspecified revenue function. We follow EP in framing the price shock be a part of the state variable. The Bellman can be written as,
\begin{equation}
V(z, \epsilon)=\max _{x} \int_{Z\times [0,1]}\left(z f(x)\epsilon'-c+\delta V\left(z^{\prime}, \epsilon^{\prime}\right)\right) Q\left(d z^{\prime} \mid z\right) Q^{R}\left(\dee \epsilon^{\prime} \mid x\right)
\end{equation}
The variable $\epsilon^{\prime}$ is the unknown price shock to the  revenue, $r=f(x)\epsilon',$ at the time the agent has to choose $x$. Its distribution is given by $Q^{R}\left(d \epsilon^{\prime} \mid x\right)\sim d_{\theta}$. The agent knows $Q$, but does not know $Q^{R}$. In particular, the agent has a parametric family of transitions, where $Q_{\theta}^{R}\left(d \epsilon^{\prime} \mid x\right)$ is the distribution of $\epsilon^{\prime}$. The parameter space $\Theta$ is compact, that is, $\Theta= [0, 2KE_{d^{*}}[\epsilon]+1]$ and the action space, $X=[0,\max\Bigg[\Bigg((E_{d^{*}}[\epsilon])/4\Bigg)^{2/3}, \Bigg((E_{d^{*}[\epsilon])}/\sqrt{K}\Bigg)^{2/3}\Bigg]+1], \mbox{ where } K=(1-e^{-k})/(1-k(k+1)e^{-k}).$ Given this,  suppose the true production function is given by $f^{*}(x)x^{1/2}$ and  hence, concave. Then the minimizer,  $\theta^{*}=2(K\expect_{d^{*}}[\epsilon])^{2/3}$ and the corresponding optimal action for the misspecified agent is $x^{*}=z\theta^{*}/2K,$\fn{ For $k> 0,$ let $K=\dfrac{1-e^{-k}}{1-k(k+1)e^{-k}}$ which is always finite and asymptotes to 1 as $k\rightarrow\infty.$ }  whereas, for the agent with the correctly specified model is, $x^{opt}=\Bigg(\displaystyle\frac{z\expect_{d^{*}}[\epsilon]}{4}\Bigg)^{2/3}.$ We first solve for the optimal action as a function of model primitives. Suppose the agent has a degenerate belief on some $\theta.$ Here, as in the original example, the agent's optimization problem reduces to a static optimization problem $\max _{x}zxE_{\theta}[\epsilon]-x^{2}$. Noting that $E_{\theta}[\epsilon]=\dfrac{\theta}{K}$,\footnote{$K=\dfrac{1-e^{-k}}{1-k(k+1)e^{-k}}.$} it follows that the optimal input choice in state $z$ is $ x^{*}=\dfrac{z\theta}{2K}.$ Next, the stationarity condition implies that the marginal of $m$ over $\mathbb{Z}$ is equal to the stationary distribution over $z$, which is $q$, a uniform distribution, $U [0,1]$. Therefore, the stationary distribution over $\mathbb{X}$, denoted by $m_{\mathrm{X}}$, has a uniform support over $[0,\dfrac{\theta}{2K}]$. Finally, we optimize $\theta$ for the weighted KLD,
$$
\begin{aligned}
\int_{x} E_{Q(\cdot \mid x)}\left[\log Q_{\theta}^{R}\left(f^{\prime} \mid x\right)\right] m_{\mathrm{X}}(x) \dee x &=\int_{[0,1]} E_{Q\left(\cdot \mid x\right)}\left[\log d_{\theta}\left(\epsilon'\right)\right] m_{\mathrm{X}}(x) \dee x \\
&=\int_{[0,1]} E_{Q\left(\cdot \mid x\right)}\!\left(-\frac{1}{\theta}\left(\epsilon'\right)-\ln \theta - \ln(1-\exp^{-k})\right) m_{\mathrm{X}}(\dee x) 
\end{aligned}
$$

Then minimizing the above expression with respect to $\theta$ gives us the minimizing $\theta^{*}$ and the corresponding $x^{*}.$

\newpage
\section{Bibliography}
\printbibliography[heading=none]

\medskip

\medskip

\medskip

\end{document}